\documentclass[aps,prb,onecolumn,groupedaddress,showpacs,superscriptaddress,amssymb,amsmath,longbibliography,]{revtex4-2}
\usepackage[utf8]{inputenc}
\usepackage{amsmath,amsthm}
\usepackage{amssymb}
\usepackage[margin=1in]{geometry}
\usepackage{mathtools}
\usepackage{dsfont}
\usepackage{xcolor}
\usepackage{algorithm2e}
\usepackage{graphicx}
\usepackage{subcaption}
\usepackage{hyperref}
\usepackage{cleveref}
\usepackage{nicefrac}

\usepackage{thmtools}
\usepackage{thm-restate}

\usepackage{mathrsfs}
\usepackage{qcircuit}
\usepackage{yquant}
\usepackage{nccmath}
\usepackage{placeins}
\usepackage{physics}
\usepackage{comment}

\newtheorem{theorem}{Theorem}
\newtheorem{lemma}[theorem]{Lemma}

\newtheorem{definition}[theorem]{Definition}

\usepackage{todonotes}

\newcommand{\bea}{\begin{eqnarray}}
\newcommand{\eea}{\end{eqnarray}}
\newcommand{\fidel}{\mathscr{F}}
\newcommand{\tracedist}{\mathscr{T}}
\newcommand{\capU}{\mathcal{U}}
\newcommand{\bigo}[1]{\mathcal{O}\left( #1 \right)}
\newcommand{\trotterchan}[2]{{S_{#1}{\parens{ #2 }}}}
\newcommand{\qdchan}[1]{\mathcal{Q}{( #1 )}}
\newcommand{\tschan}[2]{\mathcal{T}^{#1}{( #2 )}}

\newcommand{\achan}[1]{\mathcal{T}_{A}^{2k}{( #1 )}}
\newcommand{\bchan}[1]{\mathcal{Q}_{B}{( #1 )}}

\newcommand{\hilbSpace}{\mathscr{H}}
\newcommand{\evolchan}[1]{\mathcal{U}{ ( #1 ) }}
\newcommand{\diamondnorm}[1]{\left| \left| #1 \right| \right|_\diamond}

\newcommand{\liouv}{\mathcal{L}}
\newcommand{\prodform}{S_{2k}}
\newcommand{\indone}[1]{\norm{ #1 }_{1\rightarrow 1}}
\newcommand{\indonedef}[1]{\max_{\rho : \norm{\rho}=1} \norm{ #1 }_1}

\newcommand{\Braket}[3]{\langle #1|#2|#3\rangle}
\newcommand{\parens}[1]{\left( #1 \right)}
\newcommand{\brackets}[1]{\left[ #1 \right]}

\newcommand{\set}[1]{\left\{ #1 \right\}}
\newcommand{\ceil}[1]{\left\lceil #1 \right\rceil}

\newcommand{\expect}[1]{\mathbb{E}\brackets{#1}}
\DeclareMathOperator\supp{supp}

\begin{document}

\title{Composite QDrift-Product Formulas for Quantum and Classical Simulations in Real and Imaginary Time} 

\author{Matthew Pocrnic}
\affiliation{Department of Physics, 60 Saint George St., University of Toronto, Toronto, Ontario,  M5S 1A7, Canada}

\author{Matthew Hagan}
\affiliation{Department of Physics, 60 Saint George St., University of Toronto, Toronto, Ontario,  M5S 1A7, Canada}
\author{Juan Carrasquilla}
\affiliation{Department of Physics, 60 Saint George St., University of Toronto, Toronto, Ontario,  M5S 1A7, Canada}
\affiliation{Vector Institute, MaRS Centre, Toronto, Ontario, M5G 1M1, Canada}
\affiliation{Department of Physics and Astronomy, University of Waterloo, Ontario, N2L 3G1, Canada}

\author{Dvira Segal}
\affiliation{Department of Chemistry and Centre for Quantum Information and Quantum Control,
University of Toronto, 80 Saint George St., Toronto, Ontario, M5S 3H6, Canada}
\affiliation{Department of Physics, 60 Saint George St., University of Toronto, Toronto, Ontario,  M5S 1A7, Canada}
\author{Nathan Wiebe}
\affiliation{Department of Computer Science, University of Toronto, Toronto ON, Canada}
\affiliation{Pacific Northwest National Laboratory, Richland Wa, USA}
\affiliation{Canadian Institute for Advanced Research, Toronto ON, Canada}

\date{June 27, 2023}


\begin{abstract}
Recent work has shown that it can be advantageous to implement a composite channel that partitions the Hamiltonian $H$ for a given simulation problem into subsets $A$ and $B$ such that $H=A+B$, where the terms in $A$ are simulated with a Trotter-Suzuki channel and the $B$ terms are randomly sampled via the QDrift algorithm. Here we show that this approach holds in imaginary time, making it a candidate classical algorithm for quantum Monte-Carlo calculations. We upper-bound the induced Schatten-$1 \to 1$ norm on both imaginary-time QDrift and Composite channels. Another recent result demonstrated that simulations of Hamiltonians containing geometrically-local interactions for systems defined on a finite lattices can be improved by decomposing $H$ into subsets that contain only terms supported on that subset of the lattice using a Lieb-Robinson argument. Here, we provide a quantum algorithm by unifying this result with the composite approach into ``local composite channels" and we upper bound the diamond distance. We provide exact numerical simulations of algorithmic cost by counting the number of gates of the form $e^{-iH_j t}$ and $e^{-H_j \beta}$ to meet a certain error tolerance $\epsilon$. We show constant factor advantages for a variety of interesting Hamiltonians, the maximum of which is a $\approx 20$ fold speedup that occurs for a simulation of Jellium.
    
\end{abstract}
\maketitle

\section{Introduction}
Quantum simulations of quantum systems,  suggested by Feynman in 1982 \cite{feynman2018simulating}, provides what is perhaps the most natural application for quantum computers with expected exponential advantages over the best classical algorithms \cite{lloyd1996universal}. The quantum simulation problem can be summarized as follows: Given a Hamiltonian $H$, a time $t$, and an error tolerance $\epsilon$, we wish to implement, on a quantum computer, an approximate function $f$ to the \textit{time evolution operator} $e^{-iHt}$ such that $\norm{e^{-iHt} - f(iHt)}_\kappa \leq \epsilon$ for some distance norm $\kappa$. The significance of this operator comes from quantum mechanics where $e^{-iHt}$ is the solution to the Schr\"{o}dinger Equation with a time-independent Hamiltonian: $i \partial_t \ket{\Psi} = H\ket{\Psi}$ in units where $\hbar=1$. Motivation for quantum simulation is drawn from the fact that it is a necessary subroutine for phase estimation, which allows one to learn the eigenvalues of $H$ \cite{kitaev1995quantum, nielsen2002quantum}. This of particular interest in quantum chemistry and materials \cite{reiher2017elucidating, bauer2020quantum, babbush2018low} where Hamiltonians are large and analytic solutions to the eigenvalue problem are not known. In addition, quantum simulation also allows for the study of dynamical quantum observables \cite{jordan2012quantum}. Classically, implementing the time evolution operator is a difficult problem due to how rapidly the matrices grow with respect to the size of the system of interest. As well, quantum effects such as interference are hard to simulate classically. This problem has received extensive attention over the last few decades, and many different solutions have been proposed. These solutions can in some sense be broken down into two families; those built on product formulas including Trotter-Suzuki \cite{lloyd1996universal, berry2007efficient, childs2021theory} formulas and QDrift \cite{campbell2019random, wan2022randomized, ouyang2020compilation, berry2020time} and those implementing more intricate Hamiltonian transformations such as Linear Combinations of Unitaries \cite{childs2012hamiltonian, berry2015simulating}, and Qubitization \cite{low2019hamiltonian, gilyen2019quantum, martyn2021grand}. Put simply, if we have a Hamiltonian that can be expressed as a sum of terms $H=\sum_j H_j$ with a time evolution operator $e^{-it \sum_j H_j}$, a product formula is an approximation to $e^{-iHt} \approx \prod_j e^{-iH_j t}$ which is not exact in general due to non-zero commutators in the Taylor series. In the absence of fault-tolerant quantum computers, powerful methods have also been discovered to study the properties of $H$ on classical computers, particularly in the field of Quantum Monte-Carlo \cite{suzuki1993quantum, foulkes2001quantum}. Here, the propagator of interest is often the time evolution operator which has been Wick transformed to imaginary time $e^{-\beta H}: it \rightarrow \beta$. This is implemented both as a projection method to study the ground state of $H$, and to study the partition function $Z = \Tr e^{-\beta H}$ from statistical physics where $\beta$ is the inverse temperature of the system. Trotter-Suzuki formulas found some of their first applications in Monte Carlo calculations prior to becoming a popular ancilla free quantum simulation algorithm. In this paper, we take the reverse approach and consider how Composite Trotter-QDrift simulation strategies \cite{hagan2022composite} from quantum computing carry over to imaginary time. We also extend these composite strategies to account for locality of $H$ based on Lieb-Robinson bounds \cite{haah2021quantum} in terms of their application for quantum simulation (in real time). \\

Within this manuscript, we provide three contributions. First, we show that QDrift and Composite QDrift-Trotter algorithms \cite{hagan2022composite} can be extended to imaginary time with the same asymptotic scaling as in real time. The basis of these algorithms is a partitioning of the Hamiltonian such that $H=A+B$, where we simulate the subset $A$ with Trotter-Suzuki formulas and $B$ with QDrift, to take advantage of the strengths of each. This promotes them to a candidate classical algorithm for quantum Monte-Carlo calculations. In doing so, we provide new bounds on QDrift and Trotter-Suzuki formulas in imaginary time. Next, we show that Composite algorithms can be applied to geometrically local lattices, by first decomposing the lattice into local subsets \cite{haah2021quantum}, and then partitioning the Hamiltonian terms supported on said subsets into localized QDrift and Trotter simulations. In comparison to standard Composite simulations, the Local Composite simulations recover the form of the original error bounds with an additional error term that is exponentially small. Finally, we provide a library to numerically evaluate the performance of Composite channels in both real and imaginary time, as well as the performance of the Local Composite quantum algorithm in real time. Here, we also numerically introduce new optimization schemes for conducting this partitioning based on heuristics and intuition that comes from the error bounds of each respective algorithm. Rather than computing bounds, we evaluate the exact exponential gate cost (the number of exponentials $e^{-iH_j t'}$ with time slice $t'$ that appear in the product formula) at simulation time $t$ for implementing an approximation to the time evolution operator. Equivalent counts are performed in imaginary time for $e^{-\beta H}$ at inverse temperature $\beta$. This is done for a series of interesting Hamiltonians which are drawn from quantum chemistry and spin systems, the former of which are generated using OpenFermion \cite{mcclean2020openfermion}. We find significant constant factor advantages using our composite approaches, with the maximum being a factor of 18.8 for a 6-site Jellium Hamiltonian. We make the observation that composite methods seem to work better for Hamiltonians containing a large number of terms, and for those with sharply peaked spectral norm distributions. \\

The paper is organized as follows: We begin with detailed preliminary material in Section \ref{sec:prelim} that introduces the algorithms in question as well as some details regarding how they are analyzed analytically and numerically. Here we also provide some background on the Lieb-Robinson bound and define geometrically-local Hamiltonians for which they apply. In Section \ref{sec:analysis}, we show how to extend composite channels to imaginary time and we provide bounds on the error for said channels. In doing so, we show that the QDrift algorithm holds in imaginary time and that the same $\epsilon$ scaling can be achieved, suggesting it can also be used effectively for classical computation. We also provide a quantum algorithm that allows us to include locality as means for partitioning in the ``local composite" channel, and show that doing so introduces exponentially-small errors. In Section \ref{sec:numerical_sim}, we outline the details of our numerical simulations and describe the general features of our library. We then provide plots of the gate cost, or the number of exponential operators needed to meet a chosen error tolerance. We conclude in Sec. \ref{sec:summary} with a discussion and consider open questions. 

\section{Preliminaries}
\label{sec:prelim}

In this section, we introduce the general physical setup, notation and simulation approaches that will be used throughout the paper. We then review and summarize the algorithmic cost of 
the Trotter-Suzuki and the QDrift product formulas of real time time evolution, as well as their hybrid approach, which was introduced in Ref. \citenum{hagan2022composite}.

\subsection{General Notation and Assumptions}

Before introducing the simulation algorithms, it is important to outline the notation that will be used throughout the paper, as well as the assumptions that remain relevant for calculations and proofs. The following assumptions will remain valid unless otherwise stated. 
We assume that Hamiltonian operators $H$ are time-independent, bounded, and that they live in a finite $2^n$ dimensional Hilbert space $\hilbSpace$. 
We further assume that these operators can be written as a finite sum of terms, $H = \sum_{i=1}^L H_i$ with each $H_i \in \mathbb{C}^{{2^n \times 2^n}}$. 
Later, in Section \ref{sec:local_composite},
we will discuss algorithms 
for geometrically-local Hamiltonians in which case each term in this sum will only be supported on a finite subset of $\hilbSpace$; we save those details for later. 
Note that we can normalize the terms in the above sum by factoring out their respective spectral norms (Schatten infinity norm on matrices), 
$h_i = ||H_i||$ such that $H = \sum_i^L h_i H_i$. We also impose $h_i \geq 0$  by pulling any phase factors into $H_i$; this can always be done without loss of generality. The Schatten p-norms will be a useful concept throughout, and they are defined as $\norm{H}_p = \parens{\sum_i s_i^p}^{\frac{1}{p}}$, where $s_i$ are the singular values of the matrix $H$. From this definition, we observe that the Schatten infinity norm coincides with the spectral norm. Other norms will be discussed throughout the paper. Wherever a norm $\norm{\cdot}$ is lacking a subscript, it will be assumed to be the Schatten infinity norm. All other norms will be subscripted appropriately. In what follows, Hamiltonian operators are taken to be time independent, and they generate dynamics through the time evolution operator, $U = e^{-iHt}$.  \\

Throughout, we make use of the density matrix formalism to describe mixed quantum states, with quantum channels being the object that evolve states, rather than the evolution operator $e^{-iHt}$. By quantum channels we mean a mapping $\mathcal{A}(\rho) = \sum_i E_i \rho E_i^\dagger$ that is completely positive and preserves the trace, $\Tr(\rho)$. For this to be true, the operators $E_i$ must satisfy the condition $\sum_i E_i E_i^\dagger = \openone$. These conditions are required so that information regarding the channel operation is preserved and not lost to an external environment. For example, consider performing a measurement on the output density matrix $\mathcal{A}(\rho)$. If we find that $\Tr(\mathcal{A}(\rho)) \leq 1$, we conclude that the mapping does not provide a complete description of the processes involved, since there must be some other process(es) and therefore measurement outcomes that occur with non-zero probability \cite{nielsen2002quantum}. For easy distinction, we reserve calligraphic fonts for channels $\mathcal{A}$ while operators $A$ are written in the standard way. Given a density matrix $\rho(t_0)$ that describes the state at some initial time $t_0$, the time evolution channel is then defined in the following way
\begin{equation}
    \evolchan{\rho(t_0), t}:= e^{-iH  (t-t_0)} \rho(t_0) e^{iH (t-t_0)} \rightarrow \rho(t).
\end{equation}
  
When not specifically defining evolutionary channels, we will often use $t_0 = 0$ and write $\rho(t=0) = \rho$ to be the initial state. If time evolution is broken into $m$ time intervals  $\{\Delta t_1, \Delta t_{2}, ..., \Delta t_m\}$, 

the state vector evolution looks like $e^{-iH\Delta t_m} ... e^{-iH\Delta t_{2}} e^{-iH\Delta t_1}  \ket{\psi}$. However, in the channel description we do not take products of channels, rather we apply a channel multiple times in a channel composition which we denote with $\circ$. For example, $\evolchan{t_{i+1}} \circ \evolchan{t_i} (\rho)= e^{-iHt_{i+1}} e^{-iHt_i} \rho e^{iHt_i} e^{iHt_{i+1}}$.

Note that in this expression, $t_i$ corresponds to the time interval during which the channel acts. 

A significant portion of this paper is dedicated to the investigation of imaginary-time channels $\evolchan{\rho, \beta} = \frac{1}{\mathscr{N}} e^{-\beta H} \rho e^{-\beta H}$ and the algorithms that closely approximate them. Upon first glance, the notion of imaginary time seems to violate the trace preserving property of a quantum channel, so we introduce $\mathscr{N}$ as a normalization factor. To avoid confusion, we use $\beta$ throughout these discussions instead of $t$, motivated by statistical mechanics where $\beta$ is considered an inverse temperature. In these cases, we are not considering a dynamical evolution of the system in the physical sense, but are rather interested in calculating properties of the Hamiltonian $H$, such as its partition function $Z = \Tr e^{-\beta H}$ or its ground state. This is because time evolution of a quantum system is isomorphic to the cooling of a statistical ensemble via a Wick rotation. We wish to make clear that all discussions pertaining to algorithms in imaginary time are not being proposed for {\it digital quantum simulation}, but as {\it classical algorithms} for use in calculations like Quantum Monte-Carlo. For a study of imaginary time evolution on a quantum computer, see \cite{motta2020determining}

\subsection{Algorithmic Cost Model} \label{sec:cost_model}

Given that we are dealing with channels containing only product formulas, a natural cost model is to count the number of exponential gates of the form $e^{-iH_jt'}$, required to achieve some error $\epsilon$ in a simulation. 
Here, $H_i$ is a specific term chosen from a sum of Hamiltonian terms and $t'$ is an appropriately-chosen time slice of the total simulation time $t$; the following sections on the algorithms of interest explain how to choose these time slices to meet a desired simulation accuracy $\epsilon$. 
Throughout, the aforementioned cost will be denoted $C$ and often subscripted with the algorithm of interest. This cost model is convenient analytically as the algorithms we explore contain an exact amount of iterations or samples of terms that can be easily equated to a count of exponential gates. 
This situation is different from common forms of cost in the literature such as a query complexity, which requires an assumed access to an oracle that necessarily has an implementation overhead. In this case, the overhead would come within the implementation of each $e^{-iH_jt'}$, which depends on the form of the Hamiltonian terms. In the case of common Pauli and Fermionic Hamiltonians, circuits are well known and relatively simple to analyze \cite{whitfield2011simulation}. 
For example, it is common in Heisenberg-like models to have interactions that are a products of Pauli operators such as $\sigma^\nu_i \sigma^\nu_{i+1} \sigma^\nu_{i+2}$. The time evolution of these operators $e^{-it ( \sigma^\nu_i \sigma^\nu_{i+1} \sigma^\nu_{i+2})}$ can be simply implemented using rotation gates $R_\nu (\theta) = exp\parens{-i\frac{\theta}{2} \sigma_\nu}$. 
In fact, using the following circuit we can implement the time evolution for any product of Pauli operators using just $R_z(\theta)$:


\begin{figure}[h!]
\centering

\begin{tikzpicture}
\begin{yquant}
qubit {$\ket{j_{\idx}}$} i[3];
cnot i[1] | i[0];
cnot i[2] | i[1];

box {$R_z(\theta)$} i[2];

cnot i[2] | i[1];
cnot i[1] | i[0];

\end{yquant}
\end{tikzpicture}
\caption{A simple quantum circuit showing an implementation of an exponentiated produce of Pauli spin operators.}
\end{figure}
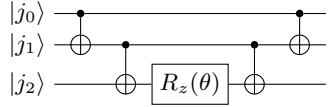

This circuit effectively implements $\exp(-i \frac{\theta}{2} \sigma^z_i \sigma^z_{i+1} \sigma^z_{i+2})$, and we can build a product of any length by continuing this pattern. Switching this to a product of $\sigma_x$ or $\sigma_y$ is straightforward given that $R_x = \hat{H}R_z \hat{H}$ and $R_y = T^2 R_z (T^2)^\dagger$. We can then adjust the simulation time $t$ through our choice of angle $\theta$. Fermionic Hamiltonians follow in a similar way, given that we can map them to a Pauli Hamiltonian using the Jordan-Wigner transformation, which will be outlined in Section \ref{sec:hydrogen}. 
The exampled circuit also reinforces the need for product formula approximations since while we can simply implement the above mentioned products of operators, implementing $e^{-it(\sigma^x+\sigma^z)}$ is highly non-trivial, and we approximate it as $e^{-it\sigma^x} e^{-it\sigma^z}$ as we detail in  the following sections. In our cost model, this approximation would be counted as 2 gates. Numerically, counting exponential gates is also a convenient cost model since it allows for straightforward counting of gates as they are applied.

\subsection{Trotter-Suzuki Formulas} \label{sec:trotter_intro}
Given an operator that can be written as a finite sum of terms $A = \sum_i^L A_i$, the simplest approximation for the corresponding operator exponential is then $e^{\sum_i^L A_i} \approx \prod_i^L e^{A_i}$. While this holds as an equality for exponential functions for variables, it does not in general hold for operator exponentials since individual terms of $A$ do not necessarily commute. This can be easily seen by considering the Taylor series of a simple exponential approximation $e^{(A + B)x} \approx e^{Ax} e^{Bx}$ to second order in $x$:

\begin{align}
    e^{(A+B)x} - e^{Ax} e^{Bx} &= \brackets{\openone + Ax+Bx + \frac{(A+B)^2 x^2}{2!} +...} - \brackets{(\openone + Ax + \frac{A^2x^2}{2!} + ...)(\openone + Bx+ \frac{B^2x^2}{2^!} + ...)} \\
    & = \frac{(AB + BA)x^2}{2} - ABx^2 + ...\\
    &= \frac{[B, A]x^2}{2} + ...
\end{align}
An error clearly arises due to a non-vanishing commutator, meaning the approximation is exact for commuting algebras. However, the nature of this error opens up the possibility for improving this approximation. If one can symmetrize this formula, the second order commutator will vanish. In reference to the previous example, a perhaps obvious symmetrization is $e^{\frac{Ax}{2}}e^{B}e^{\frac{Ax}{2}}$, which annuls the error term $\frac{[B, A]x^2}{2}$, leaving only error terms in $\bigo{x^3}$. More generally, for an operator $A = \sum_i A_i$, the 2nd order Trotter-Suzuki decomposition is given by
\begin{equation}\label{eq:trot2}
    \trotterchan{2}{Ax} = \prod_i^L e^{\frac{A_i x}{2}} \prod^i_L e^{\frac{A_i x}{2}}
\end{equation}
where the flipping of indices in the two products is indicative of reversing the ordering of products. One can readily imagine that it is then possible to continue this strategy and build higher order formulas and cancel out arbitrarily high commutator errors in the product formula approximation. Indeed, we can build these highly non-trivial decompositions in the following way  \cite{suzuki1990fractal, Hatano_2005}:

\begin{definition}[Trotter-Suzuki Product Formula] Given a linear operator $A$ acting on a finite dimensional Hilbert space that can be represented as a finite sum of linear operators $A = \sum_i^L A_i$, a Trotter-Suzuki product formula approximation  $\trotterchan{2k}{Ax}$ for $A$ of order $2k$ ($k>1$) can be constructed in the following recursive fashion:
\begin{equation}
    \trotterchan{2k}{Ax}:= (\trotterchan{2k-2}{Axs_{2k}})^2 \trotterchan{2k-2}{(1-4s_{2k})Ax} (\trotterchan{2k-2}{Axs_{2k}})^2 ,
\end{equation}
where $s_{2k} = \frac{1}{4 - 4^{\frac{2k-1}{2}}}$ .
\end{definition}
Hence, in order to obtain a desired product formula of order $2k$, we start with the second order formula $\trotterchan{2}{Ax}$ given in Equation (\ref{eq:trot2}) and recursively build the higher order formulas. To reiterate, the advantage of doing so is that for a product formula of order $2k$, the error of the approximation is $\in \bigo{t^{2k+1}}$ with $t$ here as the time interval \cite{Hatano_2005}. 

When the general operator $e^{Ax}$ is replaced by the Hamiltonian time evolution operator $e^{-iHt}$, where $H = \sum_i^L H_i$, it becomes obvious how product formulas can be used to evolve quantum states $\ket{\psi}$. However, 
we also wish to generalize this to evolve probabilistic mixtures of states, which are represented by density matrices $\rho = \sum_i p_i \ketbra{\psi_i}{\psi_i}$.  
A Trotter-Suzuki evolution formula for a channel can be written as follows:
\begin{definition}[Trotter-Suzuki Channel] Given a Hamiltonian $H$, a density matrix $\rho$, times $t$ and $t_0$ ($t> t_0 \geq 0$), and order 2k, then a Trotter-Suzuki channel $\tschan{2k}{\rho(t_0), \Delta t}$ performs the operation $\tschan{2k}{\rho(t_0), \Delta t} \rightarrow \rho(t)$ and can be defined as:
\begin{equation}
    \tschan{2k}{\rho(t_0), t} := \trotterchan{2k}{-iH \Delta t} \rho \trotterchan{2k}{-iH \Delta t}^\dagger
\end{equation}
Where $\trotterchan{2k}{-iHt}$ represent the product formulae from Equation (\ref{eq:trot2}) and $\Delta t = t-t_0$.
\end{definition}

Since the error of this approximation accumulates considerably for long iteration times, we introduce the iteration parameter $r$, and apply the channel $r$ times with time intervals $\nicefrac{t}{r}$. As a result, longer time simulations become more expensive in the number of iterations $r$ required to manage the error that accumulates to do the simulation time $t$. The composition of $r$ Trotter-Suzuki channels is straightforward, and applying $r$ channels each for time $\nicefrac{t}{r}$ can be written as
\begin{equation}
    \tschan{2k}{\rho, t}^{\circ r} = \prod^r_{j=1}\trotterchan{2k}{-iH \frac{t}{r}} \rho \prod^r_{j=1}\trotterchan{2k}{-iH \frac{t}{r}}^\dagger .
\end{equation}
The cost of this algorithm for first order decomposition is simply $rL$ where $L$ is the number of Hamiltonian terms. For order $2k$, the cost is $\Upsilon Lr$ where $\Upsilon = 2 \times 5^{k-1}$. 

Looking at cost formulas for different orders $2k$, we desire information about $r$, specifically, the minimum value of $r$ required to meet some error tolerance $\epsilon$ where $\norm{\tschan{2k}{\rho, \nicefrac{t}{r}}^{\circ r}-\evolchan{\rho, t}}_\kappa \leq \epsilon$ for some norm $\kappa$. In Ref. \citenum{hagan2022composite},  the diamond norm is selected, which is the completely bounded trace norm between channels tensored with the environment. The diamond norm is defined as follows: $\diamondnorm{\capU(t) - \mathcal{V}(t)} :=  \sup_{\rho: \norm{\rho}\leq 1}\norm{\parens{\capU(\rho, t) - \mathcal{V}(\rho, t)}\otimes \openone}_1$. This is a good analytical norm to choose, given that it maintains the interpretation of distinguishability between states from the trace distance (see Section \ref{sec:tracedist}). However, it is a bound on channels (maximized over possible input states $\rho$) and thus yields the interpretation of distinguishability between quantum channels. This norm is also sub-additive and sub-multiplicative, which makes it convenient to work with. The procedure in Ref. \citenum{hagan2022composite} upper bounds this norm, sets the value equal to $\epsilon$, and solves the expression for $r$ to arrive at the following first and $2k$th order costs
\begin{align}
    C_{TS}^1(H, t, \epsilon, 1) &= L r \leq L \ceil{\frac{t^2}{2 \epsilon} \sum_{i, j} h_i h_j \norm{[H_i, H_j]}_{\infty}}\label{eq:TrotterCost}\\
    C_{TS}^{2k}(H, t, \epsilon, 2k) &= \Upsilon L r \leq \Upsilon L \ceil{\frac{(\Upsilon t)^{1+1/2k}}{\epsilon^{1/2k}} \parens{\frac{4 \alpha_{comm}(H, 2k)}{2k+1}^{1/2k}}},\label{eq:TrotterCost2k}
\end{align}
where $C$, to reiterate, is the number of operator exponentials that need be applied rather than an exact gate count where gates are drawn from some universal gate set. Here, $$\alpha_{comm}(H, 2k) := \sum_{\gamma_i \in \set{1,\ldots L}} \parens{\prod h_{\gamma_i}} \norm{[H_{\gamma_{2k+1}}, [H_{\gamma_{2k}},\ldots[H_{\gamma_2}, H_{\gamma_1}]\ldots]}_{\infty}.$$ This is a nested commutator error that arises from $2k$ symmetrizations in the product formula. The main takeaway from this formula is the structure of the Hamiltonian that the simulation cost depends on. Trotter-Suzuki channels clearly depend on the number of terms in the Hamiltonian $L$ and on the commutator structure of the Hamiltonian. The simulation cost does not directly depend on the {\it magnitude} of the terms $\norm{H_i}$, only indirectly through their commutators. This is important to keep in mind when examining the next algorithm and its cost dependence.

\subsection{QDrift Random Compiler} \label{sec:qdrift_intro}
The QDrift algorithm was discovered by Campbell \cite{campbell2019random}. It generates random product formulas allowing for the evolution of a system to {\it stochastically} drift (through the Hilbert space) towards the true evolution with high probability. Using the introduced quantum channel notation, the algorithm works as follows: 

\begin{definition}[QDrift Channel] \label{def:QD} Given a Hamiltonian $H = \sum_{i=1}^L h_i H_i$, density matrix $\rho$, times $t$ and $t_0$ ($t > t_0 \geq 0$), and let $p_i = \frac{h_i}{\lambda}$ with $\lambda = \sum_{i=1}^L h_i$ be a probability distribution from which Hamiltonian terms are drawn, then a QDrift channel $\qdchan{\rho(t_0), \Delta t} \rightarrow  \rho(t)$ can be defined as:
\begin{equation}
    \qdchan{\rho(t_0), t} := \sum_{i=1}^L p_i e^{-iH_i \lambda \Delta t} \rho e^{iH_i \lambda \Delta t} ,
\end{equation}
where $\Delta t = t - t_0$. 
\end{definition}

This channel mixes unitaries with a single sample
from the distribution $p_i$. Similarly to the Trotter-Suzuki approach, the accuracy of this approach improves with the number of samples $N$ that go into our ``random product formula". To write the channel for multiple samples, it is useful to think of sampling a vector $\mathbf{j} = \{j_1, j_2, ... j_N\}$ of length $N$, i.i.d from the distribution $P_\mathbf{j} = \lambda^{-N}\prod_{k=1}^N h_{j_k}$. This corresponds to implementing the following unitary $\mathbf{V_j} = \prod^N_{k=1} e^{-iH_{j_k} \tau}$ with $\tau = \nicefrac{t \lambda}{N}$. With  this in mind, we can neatly write down the expression for the channel with arbitrary samples $N$ as
\begin{equation}
    \qdchan{\rho, t}^{\circ N} = \sum_{\mathbf{j}} P_\mathbf{j} \mathbf{V_j} \rho \mathbf{V_j}^\dagger,
\end{equation}
 where the sum is performed over all possible vectors $\mathbf{j}$ of length $N$. In terms of preparing these gates with a quantum circuit, this algorithm can be thought of as a linear combination of unitaries under classical control. 
 Similarly to the Trotter-Suzuki formula, we are also interested in the algorithmic 
 cost of QDrift. A diamond distance upper bound is provided in Ref. \citenum{campbell2019random}, which is then used via the same procedure to provide a cost function in Ref. \citenum{hagan2022composite}, by using $N = 4 \lambda^2 t^2 /\epsilon$ and restricting $\epsilon \in (0, \lambda t \ln(2)/2)$,
 \begin{equation} \label{eq:QDcost}
     C_{QD}(H, t, \epsilon) = \frac{4 \lambda^2 t^2}{\epsilon}.
 \end{equation}
The important take home message here is the difference in the structure of the 
algorithmic cost between the Trotter-Suzuki formula [Eq. \eqref{eq:TrotterCost}] and QDrift [Eq. \eqref{eq:QDcost}]. In QDrift, the cost does not depend on the number of Hamiltonian terms $L$, or the commutator structure. In contrast, it depends on the size of the terms in the sum of the spectral norms $\lambda$. It is important to keep this in mind when constructing composite channels.

\subsection{Composite Simulation Formulas}
\label{subsec:comp}

Equipped with the Trotter-Suzuki and QDrift channel formulas, we wish to hybridize these approaches to form a {\it composite} time evolution channel. 
The eventual goal is to do so in such a way that we can minimize the ``side effects" of each algorithm while taking advantage of the strengths of each approach. 
Following Ref. \cite{hagan2022composite}, we define the composite channel as follows. We consider the Hamiltonian $H = A+B$, where $A = \sum_{j=1}^{L_A} A_j$ and $B = \sum_{j=1}^{L_B} B_j$. Other terms of interest, such as the number of summands in a set $L$ are written with a subscript $A$ to indicate that they belong to a set of Hamiltonian terms that will be simulated by a Trotter-Suzuki channel, and the same goes for the $B$ terms, for a QDrift channel. The composite channel takes the following form
\begin{equation}\label{eq:comp12k}
    \mathcal{X}^{2k}(\rho, t) = \bchan{t} \circ \achan{t} (\rho),
\end{equation}
where the channel subscripts indicate the subset of the Hamiltonian that it is simulating. Note that there is substantial freedom in the construction of this channel. The first obvious freedom is the partitioning into the sets $A$ and $B$. Section \ref{sec:partition_scheme} is devoted to this task. The next choice lies in the Trotter-Suzuki order $2k$. Further, we can construct an outer loop as in Ref. \citenum{hagan2022composite} which uses a Trotter-Suzuki style symmetrization of the channels themselves. For example, a composite channel with outer order $2k = 2$ looks like: 
\begin{equation} 
    \mathcal{X}^{2k}(\rho, t) = \bchan{t/2} \circ \achan{t} \circ \bchan{t/2} (\rho).
\end{equation}
Such an ``outer-loop" with a respective ``outer-order" was introduced as a strategy to make more detailed analysis of simulation costs. The outer loop can be thought of as a product formula for composite channels, symmetrizing the channels and allocating the respective time slice, while the inner Trotter order does so on the operator exponentials as before. However, when introducing the outer-order in  Ref. \citenum{hagan2022composite} 
it was enforced to match the same order as the inner factorization. In this paper, one of our goals is to numerically evaluate the exact costs of these composite algorithms, and this gives us the freedom to independently set
the inner and outer orders. However, an actual application of the outer loop is situational as it requires specific knowledge about the commutator structure for it to be of use. We further discuss this in later sections where we will resort to the notation $\mathcal{X}^{2k, 2l}$ for \textit{ $2k$ inner, $2l$ outer} orders respectively. If only a single superscript is given, it is understood to be the inner order with the outer order fixed at 1.\\

We now restate the first order cost result for composite channels from Ref. \citenum{hagan2022composite}, where the $A$ terms are placed in a Trotter channel and the $B$ terms are placed in QDrift: 
\begin{align}
    C_{Comp}(A, B, t, \epsilon) = (L_A + N_B) r = (L_A + N_B) \ceil{\frac{t^2}{\epsilon} \parens{\frac{1}{2}\sum_{i,j} a_i a_j \norm{[A_i, A_j]}_{\infty} +  \frac{1}{2} \sum_{i,j} a_i b_j\norm{ [A_i, B_j]}_{\infty} + \frac{4 \lambda_B^2}{N_B}}}. \label{eq:first_order_comp_cost}
\end{align}
We only restate the first order Trotter-Suzuki result due to the fact that this is the expression from which we will draw most of intuition to build our heuristics for partitioning. 
As well, the second order cost formula is of a similar structure with extra $\alpha_{comm}$ and $\Upsilon$ terms, but it is based on channels with matching inner-order and outer-order, which our numerical work does not follow in general. Inspecting Equation (\ref{eq:first_order_comp_cost}), it clearly inherits the structure from Trotter in its dependence on the commutators and number of terms $L_A$ in the set $A$ (first term), and on QDrift in the sum of the spectral norms $\lambda_B$ and number of samples $N_B$ of the terms in the set $B$ (third term). The second term adds to the cost due to the hybridization of the Trotter and QDrift channels. 
Clearly, in order to gain as much advantage as possible using this algorithm over the non-composite approaches, we desire that the Hamiltonian of interest has a structure that we can exploit via a good choice of partitioning. 
An interesting structure might be of an $H$ that contains terms of largely varying magnitude, such that there might be a sharp contrast between the small and large terms, rather than their spectral norms following say a Gaussian distribution. Having small magnitude terms that largely outnumber the large-magnitude terms is also likely desirable. 
This understanding naturally leads into the next topic of how to choose the partitioning. This task can be motivated by the structure of $H$ and its commutators, if known a-priori. For example, inspecting Eq. (\ref{eq:first_order_comp_cost}), if we have a small number of large commuting terms and numerous small non-commuting terms, it seems natural to place the large terms into the Trotter channel, but sample the smaller magnitude terms with QDrift. The question of choosing effective partitioning will be further addressed in Section \ref{sec:partition_scheme}.

\subsection{Local Lattice Hamiltonians} 

Given the observation of commutator-dependent error in previous sections, an interesting system to study are geometrically-local Hamiltonians. These systems are defined on a lattice, and they have only local interactions, such as nearest-neighbour terms. It turns out, via a Lieb-Robinson bound (Section \ref{sec:LR}), that commutators of interaction terms fall off exponentially with the separation distance in the lattice. This motivates the search for a 
local product decomposition to simulate the time evolution of local systems. In fact, computational advantages have been demonstrated for local Hamiltonians in Ref. \citenum{haah2021quantum}. We follow this work in defining geometrically-local Hamiltonians as
\begin{definition}
    Given a D-dimensional lattice $\Gamma \in \mathbb{Z}^D$, a local Hamiltonian can be written in the following way:
    \begin{equation}
        H = \sum_{X \subseteq \Gamma} H_{X},
    \end{equation}
Where $H_X$ is only supported on the subset $X$ that contains only immediately adjacent lattice points, and it acts as $\openone$ on $\mathbb{Z}^D \in \Gamma \backslash X$. 
Equivalently, $H_X = 0$ if $diam(X) > 1$, where $diam(X)=\max_{x,x'\in X} dist(x,x')$ and $dist$ is the graph distance between the lattice indices. Each term in the Hamiltonian is also normalized such that $\norm{H_X} \leq 1$. 
\end{definition}
 In our analysis below, we will explicitly state when we are referring to a local Hamiltonian. These Hamiltonians can be used to describe a wide variety of physical systems, perhaps most famously is the quantum Heisenberg model, which will be investigated numerically in Section \ref{sec:numerical_sim}. 
 If locality is not explicitly mentioned then the Hamiltonian is of the aforementioned more general structure. 

\subsection{Lieb-Robinson Bound}
\label{sec:LR}

A Lieb-Robinson bound is essentially a bound on the speed at which information propagates through a quantum system that has interactions governed by a local Hamiltonian. This result is particularly interesting due to the fact that it holds for non-relativistic quantum systems, meaning that in no part of the system is the finite speed of light $c$ enforced. Instead, the locality of the system's interactions, as well as the geometry of the lattice leads to the emergence of a Lieb-Robinson velocity $v_{LR}$, which bounds the speed of causality. The central idea behind this bound is that the commutator between the time evolution of an operator $A_X$ and another operator $B_Y$, such that the operators are supported on disjoint sets $X \cap Y = \varnothing$, have an exponentially small commutator. More formally, a general Lieb-Robinson bound is often given in the following form \cite{nachtergaele2006lieb, nachtergaele2011much}: 
\begin{equation}
    ||[e^{iH_\Gamma t} A e^{-iH_\Gamma t}, B]|| \leq 2||A|| ||B|| F e^{-\mu (dist(X,Y) - v_{LR} |t|)}, 
\end{equation}
where $\mu$ is a constant that depends on the lattice and $F$ depends on the cardinality of the sets $X$ and $Y$. 
It is also important to note how the distance between sets is defined, and it appears in the exponent above. The distance between the sets $X$ and $Y$ is $dist(X,Y) = \min_{x\in X, y\in Y} dist(x,y)$, where $dist(x,y)$ is the distance between elements in the set; in a lattice this measure is just the absolute value of the difference of their indices. From a simulation viewpoint, and specifically regarding our Trotter-Suzuki product formulas, we know that the error in these approximations depend on commutators of terms. Therefore, carefully-chosen product decompositions may have exponentially-small error. 
This understanding  motivated the algorithm in Ref. \citenum{haah2021quantum}, and in Sec. \ref{sec:local_composite} we combine this algorithm with the composite framework. 

\section{Imaginary Time Channels Algorithmic Analysis} 

\label{sec:analysis}

Within this section, we show that the composite algorithms proposed by Hagan and Wiebe in Ref. \cite{hagan2022composite}, for the purposes of digital real-time quantum simulation, can be applied onto
the {\it imaginary time} case with similar error bounds and asymptotics. As mentioned previously, the algorithms in this section are formulated with the intention of calculating properties of the Hamiltonian on {\it classical} computers. These algorithms can therefore be considered as quantum-inspired classical algorithms, a topic unrelated to the problem of quantum imaginary time evolution \cite{motta2020determining}. 
Within this section, we introduce the imaginary-time QDrift, Trotter-Suzuki, and the composite channels, and we bound their distance norm with respect to the ideal unitary evolution.

We first introduce the notion of an imaginary time evolution channel, given in the definition below, and demonstrate how it can be viewed as a ground state preparation, or ``cooling" procedure.
\begin{definition}[Imaginary Time Evolution Channel] \label{def:ideal_imaginary_time_channel}
    Given a Hamiltonian $H$ and input state $\rho$, we define the imaginary time evolved state as the action of the following map
    \begin{equation}
        \frac{\evolchan{\rho, \beta}}{\Tr \evolchan{\rho, \beta}} = \frac{e^{-\beta H} \rho e^{-\beta H}}{\Tr e^{-\beta H} \rho e^{-\beta H}}.
    \end{equation}
    
\end{definition}
This can be seen as an imaginary time evolution channel most straightforwardly when considering quantum state vectors as opposed to density matrix. The time-independent Sch{\"o}dinger equation tells us $\ket{\psi(t)} = e^{-i H t} \ket{\psi(0)}$. If we take this expression and perform a Wick rotation that sends $t \to -i \beta$ we get that $\ket{\psi(\beta)} = e^{- \beta H} \ket{\psi(0)}$. If we consider the action of this matrix $e^{- \beta H}$ on the density matrix of the state $\ket{\psi(0}$ we get $e^{-\beta H} \ketbra{\psi(0)}{\psi(0)} e^{-\beta H}$ due to the fact that $e^{-\beta H}$ is a Hermitian operator. Relaxing the input from a pure state to a density matrix and normalizing by the trace yields the channel provided in Definition \ref{def:ideal_imaginary_time_channel}.

Often when imaginary time is discussed in quantum mechanics the state vector representation is used, in which case the imaginary time exponential operator has the following property
\begin{equation}
    \lim_{\beta \to \infty} \frac{e^{-\beta H} \ket{\psi_i}}{\norm{e^{-\beta H} \ket{\psi_i}}_2} = \ket{\psi}_{GS}.
\end{equation}
This clearly only holds whenever our initial state $\ket{\psi}_i$ has nonzero overlap with the ground state $\ket{\psi}_{GS}$, which is expressed as $\braket{\psi_i}{\psi_{GS}}\neq 0$. In the density matrix picture, we have the following equivalent property:
\begin{equation}
    \lim_{\beta \to \infty} \frac{e^{-\beta H} \rho_i  e^{-\beta H}}{\Tr e^{-\beta H} \rho_i  e^{-\beta H}} = \rho_{GS},
\end{equation}
such that $\supp (\rho_i) \cap \supp (\rho_{GS}) \neq \varnothing$. This property can be seen by expanding $\rho$ into an arbitrary mixed state in the energy eigenbasis. Now that the maps of interest have been made clear, we can consider algorithms that closely approximate them and analyze their performance analytically by bounding their error. An important detail is the choice of distance norm to quantify this error. In comparison to real time composite algorithms, the aforementioned diamond norm is not a good choice for imaginary time, given that this norm is about distinguishing quantum operations on a system and ancillary space, whereas here we are simply performing classical calculations. For this reason, we bound the following induced Schatten $p \rightarrow q$ norm on a map $\Phi$: $\norm{\Phi }_{p\rightarrow q} = \max_{\norm{\rho}_q =1} \norm{\Phi(\rho)}_p$. Here, $\Phi$ will be the difference between the ideal map and that which is generated by the composite algorithm, and we will investigate the case of the induced trace norm where $p=q=1$. Despite the fact that the maps being investigated in this section are non-linear, due to the trace operation in the denominator, the outputs of the maps are Schatten-class linear operators, and thus the induced Schatten norm is well defined. Further, by analyzing $1 \to 1$ norms for two density matrices, such as $\max_{\norm{\rho}_1  =1} \norm{\rho - \sigma}_1 \leq \epsilon$, we can guarantee from properties of the trace distance that the measurement statistics for $\rho$ and $\sigma$ will deviate by at most $\epsilon$. 

\subsection{Imaginary Time QDrift}
Introduced in Section \ref{sec:qdrift_intro}, QDrift was introduced as a quantum simulation algorithm that, in real-time simulations, has a unique property in that its error does not depend on the number Hamiltonian summands, nor on the commutators between them, only on the total size of these operators. To the best of our knowledge, this algorithm has not been applied to imaginary time calculations. We recover all the characteristic properties of the QDrift channel in imaginary time with renormalizing trace operations in the following Theorem. 

\begin{restatable}[Imaginary Time QDrift]{thm}{imagQD}\label{thm:imagQD}
Given a Hamiltonian $H = \sum_j h_j H_j$ with $\lambda = \sum_j h_j$,  imaginary time $\beta$, number of samples $N$, and a density matrix initial state $\rho$, then the induced Schatten $1\rightarrow 1$ norm of the difference between the imaginary time QDrift channel and exact imaginary time evolutionary channel has the following bound:
\begin{equation}
    \indone{\frac{\mathcal{U}(\rho, \beta)}{\Tr \mathcal{U}(\rho, \beta)} - \parens{\frac{\qdchan{\rho, \nicefrac{\beta}{N}}}{\Tr \qdchan{\rho, \nicefrac{\beta}{N}}}}^{\circ N}} \leq \frac{4 \beta^2 \lambda^2}{N} \frac{e^{\nicefrac{2 \lambda \beta }{N}} }{2 - e^{\nicefrac{2 \beta \lambda}{N}}} + 2 \frac{\frac{4\beta^4 \lambda^4}{N^3} e^{\nicefrac{6\lambda \beta}{N}}}{1-\frac{2\beta^2 \lambda^2}{N^2} e^{\nicefrac{2\lambda \beta}{N}}} +  2\frac{2\beta^2 \lambda^2}{N} e^{\nicefrac{4\lambda \beta}{N}},
\end{equation}
given that $|e^{\nicefrac{2 \beta \lambda}{N}} - 1| < 1 $ which is satisfied if $N > \frac{2 \beta \lambda}{\ln 2}$. Further, If the constraint $\frac{\lambda}{N} \leq 0.01$ is satisfied, then the bound simplifies to the following:
\begin{equation}
    \indone{\frac{\mathcal{U}(\rho, \beta)}{\Tr \mathcal{U}(\rho, \beta)} - \parens{\frac{\qdchan{\rho, \nicefrac{\beta}{N}}}{\Tr \qdchan{\rho, \nicefrac{\beta}{N}}}}^{\circ N}} \leq \frac{C\beta^2 \lambda^2}{N},
\end{equation}
where the constant $C\approx 29.71747$.
\end{restatable}

For a proof of this theorem see Appendix \ref{sec:imag_proof}. The proof consists of repeated applications of the triangle inequality and sub-multiplicative property of the diamond norm, as well as manipulations with the geometric series and a tail bound on a Taylor series. The geometric series arises when one uses a Taylor expansion on the terms in the denominator and takes the trace of the zeroth order term, which is just $\Tr \rho = 1$. A constant theme throughout the analytics in this paper are multiple exponential factors that arise due to the techniques used to bound trace terms. Multiple techniques are used to accomplish this throughout, and the exponential factors are unavoidable in imaginary time, given that the exponential operator is no longer unitary and has an exponentially decaying norm. For clearer interpretation of the bounds we use linearization techniques to either show that they disappear under some constraints when proving an inequality, or showing they become $\bigo{1}$ in some limit for an asymptotic bound. The important takeaway from this theorem is that the QDrift error scales essentially the same was in imaginary time as it does in real time, thus making it a candidate algorithm to be applied in Quantum-Monte Carlo calculations. As well, this promise on the ratio $\frac{\lambda}{N}$ is reasonable to make as this algorithm excels in the regime of small terms. 

\subsection{Imaginary Time Trotter-Suzuki} \label{sec:imag_TS}
Different from QDrift, work has already been done to provide bounds on non-unitary product formulas \cite{childs2021theory}. However, these bounds are given in terms of operator norms rather than channel distances, as well as they are given without renormalizing trace operations. However, these existing results are integral in proving the theorem below.

\begin{restatable}[Imaginary-Time Trotter-Suzuki Channels] {thm}{imagTS}\label{thm:imagTS}
    Given a Hamiltonian $H = \sum_j h_j H_j$ with $\lambda = \sum_j h_j$,  imaginary time $\beta$, and a density matrix initial state $\rho$, then the induced Schatten $1\rightarrow 1$ norm of the difference between the imaginary time Trotter-Suzuki channel $\tschan{2k}{\rho, \beta}$ and the exact imaginary time evolutionary $\evolchan{\rho, \beta}$ channel has the following bound:

    \begin{equation}
        \indone{\frac{\evolchan{\rho, \beta}}{\Tr(\evolchan{\rho, \beta})} - \parens{\frac{\tschan{2k}{\rho, \nicefrac{\beta}{r}}}{\Tr(\tschan{2k}{\rho, \nicefrac{\beta}{r}})}}^{\circ r}} \leq 2\Upsilon^{2k+1} \frac{\alpha_{comm}(H,2k)}{(2k+1)!} \frac{\beta^{2k+1}}{r^{2k}} e^{\nicefrac{4\Upsilon \beta \lambda}{r}  }\parens{e^{\nicefrac{4\beta \lambda}{r} } e^{\nicefrac{2\beta \norm{H}}{r} } + e^{\nicefrac{2\beta \norm{H}}{r} }},
    \end{equation}
    which yields the following asymptotic bound:
    \begin{equation}\indone{\frac{\evolchan{\rho, \beta}}{\Tr(\evolchan{\rho, \beta})} - \parens{\frac{\tschan{2k}{\rho, \nicefrac{\beta}{r}}}{\Tr(\tschan{2k}{\rho, \nicefrac{\beta}{r}})}}^{\circ r}} \in \bigo{ \Upsilon^{2k+1}\frac{\alpha_{comm}(H,2k)}{(2k+1)!}\frac{\beta^{2k+1}}{r^{2k}} }.
    \end{equation} 
    Here $\Upsilon$ is the number of stages of the product formula, $2k$ is the order of the product formula, and $\lambda$ is the sum of the spectral norms of the Hamiltonian summands. Here $\bigo{\cdot}$ is understood in the infinite limit of its arguments.
\end{restatable}

For a proof of this theorem, see Appendix \ref{sec:imag_proof}. The proof follows similar strategies to that of imaginary-time QDrift, however, one cannot simply expand a geometric series in the denominator here given that we have a product instead of a sum of terms. So we instead bound the trace terms using a combination of von Neumann's trace inequality and Weyl's inequality regarding the singular values of matrices. Similar to the QDrift case, we obtain a similar looking bound to that obtained in \cite{hagan2022composite} for the real time case, with some additional exponential factors as expected. Considering asymptotics, the bounds behave the same in both cases where we again see the dependence of the error on a nested commutator sum. 

\subsection{Imaginary Time Composite Channels}
Now equipped with error bounds on both QDrift and Trotter-Suzuki channels in imaginary time, we can proceed with the analysis of channels composed of the two. Given the difficulty of analyzing the renormalizing operations in the trace terms, we will again make use of asymptotic notation as was done in the Trotter-Suzuki analysis in the previous section. We will once again write a partitioning of the Hamiltonian as $H = A+B$, where $A = \sum_{i=1}^{L_A} a_i A_i$ are the terms placed in the Trotter-Suzuki channel and $B = \sum_{i=1}^{L_B} b_i B_i$ are placed in the QDrift channel.

\begin{restatable}[Imaginary Time Composite Channels]{thm}{imagCOMP}\label{thm:imagCOMP}
    Given a partitioned Hamiltonian $H = \sum_j a_j A_j + \sum_i b_i B_i$ with $\lambda_A = \sum_j a_j$, $\lambda_B = \sum_j b_j$, imaginary time $\beta$, QDrift samples $N_B$ and a density matrix initial state $\rho$, then the induced Schatten $1\rightarrow 1$ norm of the difference between the imaginary time Composite channel $\mathcal{X}^{2k}(\rho, \beta)$ of order $2k$ and the exact imaginary time evolutionary channel $\evolchan{\rho, \beta}$ has the following bound:
    \begin{equation}
    \indone{\frac{\evolchan{\rho, \beta} }{\Tr \evolchan{\rho, \beta} }- \frac{\mathcal{X}^{2k}(\rho, \beta)}{\Tr \mathcal{X}^{2k}(\rho, \beta)}} \in \bigo{\frac{\beta^2 \lambda_B^2}{N_B r} + \Upsilon^{2k+1} \frac{\alpha_{comm}(A,2k)}{(2k+1)!} \frac{\beta^{2k+1}}{r^{2k}}   + \norm{[A,B]} \frac{\beta^2}{r}}
    \end{equation}
    given that $r\geq \Upsilon \beta \lambda$. Here, all parameters in this bound correspond to the same quantities in Theorems \ref{thm:imagQD} and \ref{thm:imagTS}, where the subscripts $A$ and $B$ indicate their belonging corresponding set. 
\end{restatable}
For a proof of this bound see Appendix \ref{sec:imag_proof}. The condition on $r$ is not overly restrictive, nor arbitrary. In general, we do not expect to simulate the imaginary time evolution with $r$ sub-linear in either the imaginary time $\beta$ or the norms of the Hamiltonian $\lambda$, and $\Upsilon$ is an artefact of using higher order Trotter-Suzuki formulas. This result yields three error terms that we might expect. It is a sum of the QDrift and Trotter error, plus the error that arises from the initial partitioning of the Hamiltonian into two subsets. Therefore, it becomes apparent that if we have information about the terms in the Hamiltonian, we can potentially exploit the form of this error bound by choosing a good partition. By good partition, we mean one that is close to the optimal one that minimizes the error. We later show in Section \ref{sec:partition_scheme} how to possibly go about choosing a partition using intuition from physics, as well as provide a machine learning regressor to find the optimal partition. The later is more for a proof of principle rather than something you would run in practice due to the large overhead. 

\subsection{Examples and Implementation}
Within this section we illustrate an implementation of the imaginary time composite channel with an example. Approximating the imaginary time propagator $e^{-\beta H}$, which is a function of a Hamiltonian $H$ that is hard to exponentiate in general, can be achieved through the algorithm written in pseudo-code below. This algorithm returns a list of propagators (of which their product approximates $e^{-\beta H}$) that are instead easily diagonalizable. 

\begin{algorithm}[H]
\SetKwInput{Input}{Input}
\SetKwInput{Output}{Output}
\SetKwFunction{PARTITION}{PARTITION}
\SetKwFunction{TROTTER}{TROTTER}
\SetKwFunction{SAMPLE}{SAMPLE}
\caption{Pseudo-code for implementation of an imaginary time composite simulation using a high-order Trotter formula and a partitioning heuristic to divide the Hamiltonian terms between the two channels.}\label{pseudocode}

\Input{Simulation parameters: Trotter-Suzuki inner-order $2k$, QDrift samples $N_B$, iterations $r$, and imaginary time $\beta$. A list of Hamiltonian terms $H = \sum_j h_j H_j$, two deterministic classical functions: \PARTITION{$H$} which returns two lists of Hamiltonian terms $A = \sum_q a_q A_q$ and $B = \sum_p b_p B_p$, and \TROTTER{$A, 2k, \beta, r$} that returns a Trotter-Suzuki product formula of order $2k$ with time-slice $\frac{\beta}{r}$. We also require a classical oracle function \SAMPLE{} which returns a value $j$ from the probability distribution $p_j = \frac{b_j}{\sum_l b_l}$.}
\Output{A \textit{composite imaginary time product formula} in the form of a list of propagators, with \textit{inner-order} $2k$ and \textit{outer-order} $1$, which closely approximates the imaginary time propagator $e^{-\beta H}$.}

\SetAlgoLined
\SetKwProg{Fn}{Function}{:}{}
\Fn{CompositeListCompilation($H, \beta, r, 2k, N_B$)}{
    $A, B \leftarrow$ \PARTITION{$H$}\;
    $m \leftarrow 1$\;
    $\mathcal{X} \leftarrow$ \{ \}(empty ordered list)\;
    $\mathcal{T} \leftarrow$ \TROTTER{$A, 2k, \beta, r$}\;
    \While{$m \leq r$}{
        $\mathcal{Q} \leftarrow$ \{ \}(empty ordered list)\;
        $\lambda_B \leftarrow \sum_l b_l$\;
        $j \leftarrow 1$\;
        \While{$j \leq N_B$}{
            $j \leftarrow$ \SAMPLE{}\;
            append $e^{-\beta \lambda_B B_j / N_B r}$ to $\mathcal{Q}$\;
            $j \leftarrow j + 1$\;
        }
        append $\mathcal{T}$ to $\mathcal{X}$\;
        append $\mathcal{Q}$ to $\mathcal{X}$\;
        $m \leftarrow m + 1$\;
    }
    \Return $\mathcal{X}$\;
}
\end{algorithm}

\FloatBarrier

The algorithm above is written for illustrative purposes and simply constructs a composite product formula in the form of an ordered list. This can then be used as a subroutine for other algorithms in computational physics, such as those from the field of Quantum Monte-Carlo (QMC). For example, in statistical physics calculating the partition function $Z$ is of central interest. Given this returned list, which for clarity, is not a channel, we can approximate $Z$ as follows:
\begin{align}
    Z &= \Tr (e^{-\beta H}) = \Tr ((e^{-\frac{\beta}{r} A_1} ...  e^{-\frac{\beta}{r} A_{L_{A}}} e^{-\frac{\beta \lambda_B}{rN_B} B_1} ... e^{-\frac{\beta \lambda_B}{rN_B} B_{N_B}})^r )+ \epsilon,
\end{align}
where $\epsilon$ is bounded by a set of functions implied by Theorem \ref{thm:imagCOMP}. Next, we insert the resolution of the identity between each propagator such that the trace disappears
\begin{align}
    Z = \sum_{j_1, j_2, ... j_{n}} \bra{j_n} e^{-\frac{\beta}{r} A_1} \ket{j_1} ... \bra{j_{L_A -1}} e^{-\frac{\beta}{r} A_{L_{A}}} \ketbra{j_{L_A}} e^{-\frac{\beta \lambda_B}{rN_B} B_1}  \ket{j_{L_A +1}} ... \bra{j_{n-1}}e^{-\frac{\beta \lambda_B}{rN_B} B_{N_B}} \ket{j_n} + \epsilon,
\end{align} 
given that $n = |\mathcal{X}| r$ or the number of terms in the list multiplied by the number of iterations. Note the indices are written such that there $L_A r$ Trotter terms and $N_B r$ QDrift terms, with the difference that the subscripts indexing the QDrift terms are dummy indices, as these propagators are importance sampled at random, as shown in the pseudo-code. Now, with the above expression, we can either diagonalize these terms and Monte-Carlo sample this sum, or impose a classical mapping that allows us to apply some of the well known QMC algorithms. These classical mappings depend on the system in play, for example, if we are interested in Heisenberg or XXZ chains, it is possible to map the model onto a higher dimensional Ising model and apply the method of world lines \cite{suzuki1993quantum}. The motivation for using this more complicated propagator list, as opposed to the standard Trotterization, is that a propagator with a lower error requires a smaller $r$ to achieve the desired precision, and can therefore reduce the path length of the Monte-Carlo path integral required to approximate $Z$ above. Our numerics show that this framework can achieve a substantially lower error in the right conditions via attaining a smaller $C$ in multiple cost plots in Section \ref{subsubsec:iTime_results}, therefore, highlighting the potential applications of this algorithm to improve QMC simulations. \\

In this example, it was chosen to show arbitrary inner-order and outer-order 1, given that the goal is simplicity. This can be simply generalized to higher outer-orders by also requiring an additional subroutine that symmetrizes $\mathcal{T}$ and $\mathcal{Q}$, and adjusts their time-slices accordingly. We also remark that the classical function \texttt{PARTITION()} is presented here somewhat as a black box. This is done for generality, given that the user is free to use a method of their choosing. In section \ref{sec:partition_scheme}, we present an explicit algorithm for this function, which requires a single additional input $\omega_c$ called the \textit{chop threshold}, and simply involves a looping \textbf{if statement} that sorts Hamiltonian summands into $A$ or $B$ depending on whether their spectral norm is greater or less than $\omega_c$.

\section{Local Composite Channels for Real-Time Evolution}
\label{sec:local_composite}

We introduce here an algorithm for {\it real time evolution}, which utilizes the {\it local} structure of the Hamiltonian and is {\it composite}. It combines intuition from both the Trotter-QDrift composite channels as discussed in Sec. \ref{subsec:comp} and the Lieb-Robinson decomposition, Sec. \ref{sec:LR}.
Our main result is Eq. (\ref{eq:LCC}), which bounds the error of a local composite channel. 

The basic idea behind the method goes as follows: We decompose the Hamiltonian evolution into local blocks and build composite channels out of the operators supported on these local blocks. Each block is thus assigned a local partition $\{A, B\}$ and a sample number $N_B$. 

Based on prior intuition, we expect this approach to perform well on disordered models yet with local interactions, such that there is a structure for the local composite channels to exploit.  
However, by adding QDrift sampling to the local framework, with the speed limit $v_{LR}$ in mind we are ``fuzzing"  
the light-cone in the quantum circuit evolution of the system. 
For example, in Figure \ref{fig:cirquit} we show a circuit that illustrates the Trotter time evolution of a local Hamiltonian where the $H_{i,j}$ terms only connect neighbouring qubits within and between each disjoint subset $A$ and $B$ 
of the 1-dimensional lattice. 

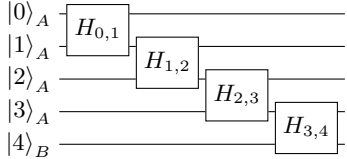
\begin{figure}[htbp!]
\centering
\begin{tikzpicture}
\begin{yquant}

qubit {$\ket{\idx}_A$} a[4];
qubit {$\ket{4}_B$} b[1];

box {$H_{0,1}$} (a[0-1]);
box {$H_{1,2}$} (a[1-2]);
box {$H_{2,3}$} (a[2-3]);
box {$H_{3,4}$} (a[3]-b[0]);

\end{yquant}
\end{tikzpicture}
\caption{Brickwork circuit diagram to highlight the emergent light-cones from nearest neighbour interactions. The notion of ``fuzzing" the light-cone occurs when we use QDrift to sample terms along these light-cone boundaries in the circuit.} \label{fig:cirquit}
\end{figure}

 The light-cone analogy is now made clear: A term $H_{0,1}$ influences the most adjacent qubit in $B$ only once 3 ``layers" of gates have been applied (intermittent gates cannot improve this speed). This behavior is analogous to the Lieb-Robinson velocity that upper-bounds the propagation of information flow in the lattice. 
 We note though that in a local composite channel, some of the gates are being randomly sampled in the QDrift channel, thus the evolution of the light-cone is not exactly reproduced since not all terms are necessarily implemented due to sampling. Keep in mind that local composite channel does include more complicated subsets of gates, some of which evolve boundary region terms like $H_{3,4}$ backwards in time; the circuit above is presented simply for conceptual purposes. 

 Towards bounding the error of a local composite channel, we build on the result from Ref. \citenum{haah2021quantum}, which concerns with bounds on time evolution operators,
 and reformulate it as a bound on channels. The motivation for the reformulation of the bound into a channel norm is so that QDrift could be incorporated into this algorithm. This would later allow us to bound the overall error. Since QDrift produces mixed states, we need to work with the density matrix formalism. 
 
\begin{theorem} 
\label{thm:LocalResult}
Given a time $t$ and a local Hamiltonian $H = \sum_X H_X$ that generates a time evolution unitary $U$ and an evolutionary channel $\capU$, $\exists \; (\mu>0) \; |$ for any disjoint regions $A,B,C$ we have the following operator norm bound: 
\begin{equation}
    ||\capU_{A\cup B} \capU^\dagger_B \capU_{B\cup C} - \capU_{A\cup B \cup C}|| \in \bigo{e^{-\mu dist(A,C)}} \sum_{X:bd(AB,C)} ||H_X||,
\end{equation}
where the evolutionary channels contain only Hamiltonians terms supported on their subscripted sets of the lattice, and X:bd(AB,C) indicates $X\subseteq A\cup B\cup C$ and $X \not \subseteq A\cup B$ and $X \not \subseteq C$.
\begin{proof}
    We first restate the main result from Ref. \citenum{haah2021quantum} on time evolution operators,
    \begin{equation}
        ||U_{A\cup B} U^\dagger_B U_{B\cup C} - U_{A\cup B \cup C}|| \in \bigo{e^{-\mu dist(A,C)}} \sum_{X:bd(AB,C)} ||H_X||.
    \end{equation} 
Now, write $V:= U_{A\cup B} U^\dagger_B U_{B\cup C}$ and $U := U_{A\cup B \cup C}$ and consider the following induced channel infinity norm
\begin{align}
\label{eq:b1}
        \max_{\rho : \norm{\rho}_1 \leq 1} ||V\rho V^\dagger - U \rho U^\dagger|| &= \max_{\rho : \norm{\rho}_1 \leq 1} ||V\rho V^\dagger - U\rho V^\dagger   + U\rho V^\dagger - U\rho U^\dagger|| \nonumber\\
        &= \max_{\rho : \norm{\rho}_1 \leq 1}||(V-U)\rho V^\dagger + U\rho(V^\dagger-U^\dagger) || \nonumber\\
        & \leq \max_{\rho : \norm{\rho}_1 \leq 1} ||(V-U)\rho V^\dagger|| + \max_{\rho : \norm{\rho}_1 \leq 1}||U\rho(V-U)^\dagger|| 
        \nonumber\\
        & \leq 2||V-U||.
\end{align}
In the last line we used the fact that the infinity norm of a density matrix is upper bounded by 1, and  the unitary invariance of the Schatten norms. Inserting the original definitions for $V$ and $U$ into $\max_{\rho : \norm{\rho}_1 \leq 1} ||\mathcal{V} - \capU|| \leq 2||V-U||$, where $\mathcal{V}$ is just $V\rho V^\dagger$, 
and similarly $\mathcal{U}= U\rho U^\dagger$, 
we successfully upper bounded the desired channel norm with the result   
with the constant factor disappearing due to 
$\mathcal{O}$  notation.
\end{proof}
\end{theorem}

The bound (\ref{eq:b1}) we have proven above is still not of the form of that given in QDrift.
QDrift bound utilizes the diamond distance, which is  a completely bounded trace norm or Schatten 1-norm defined as $\diamondnorm{\mathcal{V}-\capU} = \max_{\rho : \norm{\rho}_1 \leq 1}||(\mathcal{V}-\capU)\otimes \openone||_1$. 
While this is an upper bound on the infinity norm in general, it is not immediately clear how tight this bound may be as factors of dimensionality of the system may come into play, which we wish to avoid. This can be made clear by a simple example: Consider an $N\times N$ identity matrix $\openone_N$. This matrix has $||\openone_N||_1 = N$ while $||\openone_N||_\infty = 1$, where $N$ is a factor of dimensionality. Therefore, if possible, we would also like to convert 
the  result (\ref{eq:b1}) to a diamond norm. Building on Ref. \citenum{hagan2022composite}, we accomplish this in the following lemma:

\begin{lemma} \label{lem:LocalDiamond}
Given a time $t$ and a local Hamiltonian $H = \sum_X H_X$ that generates a time evolution unitary $U$ and an evolutionary channel $\capU$, $\exists \; (\mu>0) \; |$ for any disjoint regions $A,B,C$ we have the following diamond distance bound
\begin{equation}
\label{eq:b2}
    \diamondnorm{\capU_{A\cup B} \capU^\dagger_B \capU_{B\cup C} - \capU_{A\cup B \cup C}} \in \bigo{e^{-\mu dist(A,C)}} \sum_{X:bd(AB,C)} ||H_X||
\end{equation}
\begin{proof}
We use Equations (11-17) from Ref. ~\citenum{hagan2022composite}  replacing their Trotter channel with the local block channel $\mathcal{V} = \capU_{A\cup B} \capU^\dagger_B \capU_{B\cup C}$ and set $\capU = \capU_{A\cup B \cup C}$ with $V$ and $U$ representing the unitaries that induce each channel,
\begin{align}
        \diamondnorm{\capU(t) - \mathcal{V}(t)} :=& \norm{\parens{\capU(t) - \mathcal{V}(t)}\otimes \openone}_1 \label{eq:diamond_to_spectral_start} \nonumber\\
        =& \max_{\rho : \norm{\rho}_1 \leq 1} \norm{e^{-iHt}\otimes \openone \rho e^{i H t}\otimes \openone - V(t)\otimes \openone \rho V^\dagger(t)\otimes \openone}_1
        \nonumber\\
        \leq& \max_{\rho : \norm{\rho}_1 \leq 1} \norm{e^{-iHt}\otimes \openone \rho e^{i H t}\otimes \openone -e^{-i H t} \otimes \openone \rho V^\dagger(t) \otimes \openone}_1 
        \nonumber\\
        \text{ }& +\max_{\rho : \norm{\rho}_1 \leq 1} \norm{ e^{-i H t} \otimes \openone \rho V^\dagger(t)  - V(t)\otimes \openone \rho V^\dagger(t)\otimes \openone}_1 
        \nonumber\\
        =& \max_{\rho : \norm{\rho}_1 \leq 1} \norm{\rho \parens{e^{-i H t} - V^\dagger(t)}\otimes \openone}_1 + \max_{\rho : \norm{\rho}_1 \leq 1} \norm{\parens{e^{-i H t} - V(t)}\otimes \openone \rho}_1 
        \nonumber\\
        \leq & 2 \norm{e^{-i H t} - V(t)}_\infty \max_{\rho : \norm{\rho}_1 \leq 1} \norm{\rho}_1 \nonumber\\
        = & 2 \norm{e^{-i H t} - V(t)}_\infty 
        \nonumber\\
        = & 2 \norm{U(t) - V(t)}_\infty.
    \end{align}
In the last line, we inserted the Theorem \ref{thm:LocalResult} to complete the proof. We observe the same bound on our Diamond distance as appeared in the induced infinity norm with the constant factor 2 once again absorbed into $\mathcal{ O}$.
\end{proof}
\end{lemma}

Onward, the idea of the the local-composite algorithm is to start with some Hamiltonian $H = \sum_X H_X$ defined on some lattice $\Gamma$ with $X \subseteq \Gamma$. Begin by breaking up the lattice into two regions: $A\cup B = \Gamma$. Now we take a smaller region of the lattice $Y \subseteq A$ and use the result to approximate the time evolution operator in the following way,
\begin{equation} \label{eq:blocking}
    e^{-iHt} \approx e^{-iH_At} (e^{-iH_Yt})^\dagger e^{-iH_{Y\cup B }t}
\end{equation}
Note that this breakup is distinct from Trotterization in that we are decomposing the operator into a product formula of which the operator exponentials still contain sums of operators. However, the operator sums now only contain terms in each local ``block". We can re-curse this process to make $m$ blocks of near-equal size. Next, we decide which simulation method to employ to approximate and implement the block-evolution operators. In the previous theorem 
(\ref{eq:b2}),
a bound was proved for only one round of decomposition, but note that when defining $V$, the proof holds regardless of how many unitary blockings the operator is composed of. 
In Ref. \citenum{haah2021quantum}, it was stated that the error in Theorem \ref{thm:LocalResult} for $D$ rounds of decomposition in a $D$-dimensional lattice is $\in \bigo{e^{-\mu l} \nicefrac{DL^D}{l}}$ for a lattice with chain length $L$ and blocks with overlap $l$. The overlap for the block in the above example would be the diameter of the set $Y$, which can be written $diam(Y)=\max_{y,y'\in Y} dist(y,y')$;  the latter distance is again the distance between lattice indices. 

Before investigating this channel numerically, we  provide a bound on the diamond norm between this local composite channel and the {\it exact} evolutionary channel. Along with the fact that the diamond norm is sub-additive and sub-multiplicative, we will make use of the following two properties:
\begin{align*} \label{eq:normsum}
    \diamondnorm{\mathcal{A}^{\circ r} - \mathcal{B}^{\circ r}} &\leq r \diamondnorm{\mathcal{A} - \mathcal{B}} \\
    \diamondnorm{\mathcal{A}\mathcal{B} - \mathcal{C} \mathcal{D}} &\leq \diamondnorm{\mathcal{A} - \mathcal{C}} + \diamondnorm{\mathcal{B}-\mathcal{D}},
\end{align*}
Where $\mathcal{A}, \mathcal{B}, \mathcal{C}$ and ${\mathcal{D}}$
represent channels.  

These properties were shown in Ref. \citenum{watrous2018theory}. For our purposes, we require a more general property that holds for the composition of $k$ channels $\mathcal{A}_i$. We provide a simple proof in the following lemma:

\begin{lemma} \label{lem:multisub}
    Given a composition of $k$ quantum channels $\mathcal{A}_i$, and another composition of $k$ quantum channels $\widetilde{\mathcal{A}_i}$, the diamond norm of the difference between these two channels has the following bound
    \begin{equation}
        \diamondnorm{\bigcirc^k_{i=1} \mathcal{A}_i - \bigcirc^k_{i=1} \mathcal{B}_i} \leq \sum_{i=1}^k \diamondnorm{\mathcal{A}_i-\mathcal{B}_i} . \label{eq:normsum}
    \end{equation}
    \begin{proof}
The proof is achieved by induction and it uses the fact that $\diamondnorm{\mathcal{A}_i} \leq 1$ for any quantum channel, which is due to the fact that the spectral norm is unitary invariant and density matrices have a maximum eigenvalue of 1. Starting from the left and suppressing the composition operation $\circ$ above, we have
\begin{align}
    \diamondnorm{\mathcal{A}_k \mathcal{A}_{k-1} ...  \mathcal{A}_1 - \mathcal{B}_k  \mathcal{B}_{k-1} ...  \mathcal{B}_1} &= \diamondnorm{\mathcal{A}_k \mathcal{A}_{k-1} ...  \mathcal{A}_1 - \mathcal{B}_k \mathcal{A}_{k-1} \mathcal{A}_{k-2}...\mathcal{A}_1 + \mathcal{B}_k \mathcal{A}_{k-1} \mathcal{A}_{k-2}...\mathcal{A}_1 - \mathcal{B}_k  \mathcal{B}_{k-1} ...  \mathcal{B}_1} 
    \nonumber\\
    &= \diamondnorm{(\mathcal{A}_k - \mathcal{B}_k)\mathcal{A}_{k-1}...\mathcal{A}_1 + \mathcal{B}_k(\mathcal{A}_{k-1}\mathcal{A}_{k-2}...\mathcal{A}_1 - \mathcal{B}_{k-1}\mathcal{B}_{k-2}...\mathcal{B}_1)}\nonumber\\
    & \leq \diamondnorm{\mathcal{A}_k - \mathcal{B}_k}\diamondnorm{\mathcal{A}_{k-1}...\mathcal{A}_1} + \diamondnorm{\mathcal{B}_k}\diamondnorm{\mathcal{A}_{k-1}\mathcal{A}_{k-2}...\mathcal{A}_1 - \mathcal{B}_{k-1} \mathcal{B}_{k-2}...\mathcal{B}_1} 
    \nonumber\\
    & \leq \diamondnorm{\mathcal{A}_k - \mathcal{B}_k} + \diamondnorm{\mathcal{A}_{k-1}\mathcal{A}_{k-2}...\mathcal{A}_1 - \mathcal{B}_{k-1}\mathcal{B}_{k-2}...\mathcal{B}_1} .
    \end{align}
Repeating this procedure, we decompose the rightmost term into a sum of two channel differences by induction, allowing one to prove the original identity in Equation (\ref{eq:normsum}).
\end{proof}
\end{lemma}

Next, we utilize Theorem \ref{thm:LocalResult}, where we converted the spectral norm result on local block unitaries to a bound on a local decomposition channel. 
It is now made clear why this was done; we require a channel description to describe the QDrift algorithm and relate their bounds. Using Lemma \ref{lem:multisub}, we proceed with the initial goal of bounding the diamond norm between the local composite channel and the exact evolutionary channel. Here, for compactness we adopt the notation of writing blocked composite local channels as $\mathcal{X}_{Y_i}$ and ``exact" local channels (prior to hybridization into Trotter and QDrift) $\mathcal{U}_{Y_i}$. Here, $Y_i$ indicates a set of lattice points $Y_i \subseteq \Gamma$ such that each $Y_i$ is constructed via the algorithm laid out earlier in this Section. Therefore, only certain Hamiltonian interaction terms are supported on the lattice subsets $Y_i$. Additionally, let $\mathcal{U}_{Y_i}$ contain terms that have also undergone the appropriate conjugation (given that the intersections of some blocks are evolved backwards in time). Finally, let each set $Y_i \in Y$ such that $Y = \{Y_1, Y_2, ... , Y_m\}$ and $m$ is the number of blockings constructed by the Lieb-Robinson decomposition. This provides convenient notation for the following theorem:

\begin{theorem} \label{thm:localcomp_diamond}
    Given a time $t$, local channel iterations $r$, and a local Hamiltonian $H_\Gamma$ on a $D$-dimensional lattice $\Gamma$, let $\mathcal{X}_Y(t)$ be a local-composite channel with $m$ local blocks $Y_i$, and partition $\{A, B\}_Y$ on each block, then the diamond norm of the difference between the local composite channel and the evolutionary channel has the following bound:
    \begin{align}
        &\diamondnorm{\bigcirc_Y \parens{\mathcal{X}_Y(\nicefrac{t}{r})}^{\circ r} - \capU(t)} \leq m \max_Y\diamondnorm{\parens{\mathcal{X}_Y(\nicefrac{t}{r})}^{\circ r} - \mathcal{U}_{Y}(t)} + \epsilon_{LR}
        \nonumber\\
        &\quad= \frac{t^2 m}{2r} \max_Y \parens{\parens{\sum_{i,j} a_{Y_i} a_{Y_j} \norm{[A_{Y_i}, A_{Y_j}]} + \sum_{i, j} a_{Y_i} b_{Y_j} \norm{[A_{Y_i}, B_{Y_j}]}} +  \frac{4 \lambda_{B_Y}^2 t^2 }{N_{B_Y} r}} + \bigo{e^{-\mu l} DL^D/l}
        \label{eq:LCC}
    \end{align} 
    where $A_Y, B_Y$ are once again the terms of $H_Y$ in Trotter and QDrift with their respective spectral norms $a_Y, b_Y$, $N_B$ is the number of QDrift samples and $\lambda$ is the sum of the spectral norms of the terms in QDrift. It should also be stressed that $r$ is a local channel iteration, which means that we do not iterate the channel that is a composition of blocks but rather each individual local composite channel built out of the block terms. The reason for this is that the error in the local decomposition is independent of time for time independent $H$ with $t \in \bigo{1}$. We also require $m \in \bigo{(L/l)^D}$ \cite{haah2021quantum}. 
    \begin{proof}
        We begin with the norm we wish to bound, and then apply the above identities to format the norm such that known bounds from \cite{haah2021quantum} and \cite{hagan2022composite} can be applied. 
        \begin{align}
            \diamondnorm{\bigcirc_{i=1}^m \parens{\mathcal{X}_Y(\nicefrac{t}{r})}^{\circ r} - \capU(t)} &= \diamondnorm{\bigcirc_{i=1}^m  \parens{\mathcal{X}_{Y_i}(\nicefrac{t}{r})}^{\circ r} - \bigcirc_{i=1}^m \mathcal{U}_{Y_i}(t) + \epsilon_{LR}} \\
            &\leq  \sum_{i=1}^m \diamondnorm{ \parens{\mathcal{X}_{Y_i}(\nicefrac{t}{r})}^{\circ r} - \mathcal{U}_{Y_i}(t)} +\epsilon_{LR} \\
            & \leq  m \max_{Y} \diamondnorm{\parens{\mathcal{X}_Y(\nicefrac{t}{r})}^{\circ r} -  \mathcal{U}_{Y}(t)} +\epsilon_{LR}
        \end{align}
        Then by inserting the bound given in \cite{hagan2022composite} Equation 31, we arrive at the original expression. It then immediately follows that we can bound local composite channels of any order $2k$ by inserting the result of \cite{hagan2022composite} Equation 58, however, we omit this as it follows trivially from above.
    \end{proof}
\end{theorem}
The takeaway from this is that the block-local composite channel is bounded by an error that is very similar to that of the general composite channel asides from the addition of an extra term that is asymptotically exponentially small. However, while we expect to see improvement from this localization scheme, getting more out of this bound does not seem feasible due to the amount of freedom in this algorithm as almost every parameter is dependent on the ``worst" subset of the set $Y$. 

\section{Numerical Simulations} \label{sec:numerical_sim}
In this section we highlight the details of configuring our numerical investigations by defining Hamiltonians of interest. We then discuss partitioning strategies and choice of an error measure before presenting numerical results. The choice of error measure usually does not require a dedicated subsection. However, when considering composite algorithms, there is a subtle feature in that not all algorithms may be treated on equal footing by the same error measure, which can lead to inconsistencies. For our proposes, the trace distance is the most useful measure of choice. Rather than computing upper-bounds, the numerical results we provide here are exact calculations of the cost $C$, or the precise counts on the number of quantum gates with the form $e^{-i\frac{t}{r} \bigotimes_j \sigma_j^\nu}$ required to meet some desired precision $\epsilon$. We write gates in this form to highlight that the Hamiltonian summands in our simulations are always written as a tensor product of Pauli operators, allowing for a nice parallel to the well known rotation gates in the context of quantum computing. For imaginary time, although this form still holds in our numerics, the same analogy does not hold, as we are simply constructing $e^{-\frac{\beta}{r} H_j}$ on a classical computer which can be of somewhat arbitrary form. \\ 

In order to calculate $C$, we have constructed a library to compile any desired product formula simulation, given a list of Hamiltonian terms, partition, QDrift samples, simulation time and desired precision. These simulator objects are handled by external functions that can partition the simulators, calculate errors and exact costs, or approximate simulation cost via Monte-Carlo methods and more. The library is built on NumPy, but also contains conversion functions to load Hamiltonians generated by quantum chemistry packages OpenFermion \cite{mcclean2020openfermion} and PySCF \cite{sun2018pyscf} to simulate systems of interest. It also contains methods for geometrically local simulations to compute blockings using the required set logic.

\subsection{Hamiltonians of Interest}
Within this section, we introduce the Hamiltonians in our numerical investigation of composite algorithms, as well as briefly outline methods used for their generation. 

\subsubsection{Electronic Structure Hamiltonians in Second Quantization} \label{sec:hydrogen} The electronic structure problem is perhaps one of the most famous classically intractable problems that has vast applications in quantum chemistry. In order to write down these Hamiltonians, it is first necessary to introduce Fermionic creation and annihilation operators. Fermions are particles with half-integer spin that obey Fermi-Dirac statistics, meaning they obey the following anti-commutation relations:
\begin{align}
    \{a_m, a_n^\dagger\}:&= a_ma_n^\dagger + a_n^\dagger a_m = \delta_{mn} \\
    \{a_m, a_n\} &= \{a_m^\dagger, a_n^\dagger\} = 0
\end{align}

We work in Fock space where the subscripts of operators indicate the excitation number or the atomic orbital of an electron. Molecular electronic structure Hamiltonians then take the following form: 
\begin{equation}\label{eq:secondquant}
    H = \sum_{mn} h_{mn} a_m a_n^\dagger + \sum_{pqrs} h_{pqrs} a_p a_q a_r^\dagger a_s^\dagger 
\end{equation} 
Here the first term represents single excitations and the second term keeps track of double excitations or ``hopping" amongst orbitals. The coefficients $h$ are molecular integrals that depend on the basis of choice to describe the molecule \cite{whitfield2011simulation}. Further, more can be done to aid in the implementation of this problem on a quantum computer. The Jordan-Wigner transformation provides a one-to-one mapping between the fermionic and spin operators. This will allow us to write down the time evolution operator in terms of the universal rotation gates (and CNOT gates). To understand the transformation, first observe: 
\begin{align}
    a^\dagger &= \begin{bmatrix}
    0 & 0\\
    1 & 0 
    \end{bmatrix} = \frac{\sigma^x - i\sigma^y }{2}:= \sigma^- \\
    a &= \begin{bmatrix}
    0 & 1\\
    0 & 0 
    \end{bmatrix} = \frac{\sigma^x + i\sigma^y }{2}:= \sigma^+
\end{align}
Now to build in the desired commutation relations and  generalize this to a Hilbert space for $N$ qubits, or the tensor (Kronecker) product of $N$ 2-dimensional Hilbert spaces $\hilbSpace = \bigotimes_{i=1}^N \hilbSpace_i$:

\begin{align}
    a_n^\dagger & \Leftrightarrow \openone^{\otimes n -1} \otimes  \sigma^- \otimes (\sigma^z)^{\otimes N- n -1} \\
    a_n & \Leftrightarrow \openone^{\otimes n -1} \otimes  \sigma^+ \otimes (\sigma^z)^{\otimes N- n -1}
\end{align}

In order to build these Hamiltonians numerically, we use the OpenFermion \cite{mcclean2020openfermion} and PySCF \cite{sun2018pyscf} packages for quantum chemistry. OpenFermion is a library that allows for the easy manipulation of fermionic operators that arise in quantum chemistry, as well as it interfaces with a variety of electronic structure packages that perform molecular integrals in the basis of choice to generate Hamiltonians in the form of Equation \ref{eq:secondquant}. Further, OpenFermion also has the Jordan Wigner transform built in, allowing one to construct this Hamiltonian in the Pauli basis. PySCF was our electronic structure package of choice to compute molecular integrals. \\


Given the form of Equation \ref{eq:secondquant}, we observe that our Hilbert space needs to be truncated. An active space calculation does exactly this; the Hamiltonian is written in a space such that only so many orbitals are ``active" or such that an electron can be excited to occupy active orbitals. We generate all of our electronic structure Hamiltonians in the minimal basis where we use a the number of qubits equal to the total period of the molecule. For our numerical investigation, we provide a function to generate chains of hydrogen atoms given a very simple input; bond length and number of atoms. The function uses PySCF to compute the molecular integrals, and then uses the data to build the Hamiltonian in an active space implied by the minimal basis, and using a minimal spin configuration. 


\subsubsection{Jellium Uniform Electron Gas}
Jellium is a model of a uniform electron gas that captures the interactions between delocalized electrons in a solid with uniformly distributed positive potentials serving as Nuclei. It is not only a system of interest in Materials Science, but also as a benchmark system in quantum simulation. More compact representations of this Hamiltonian have been proposed as a candidate for experimental simulation on near-term hardware \cite{babbush2018low}. The system Hamiltonian has a closed form representation and does not require any additional molecular integrals to construct:
\begin{equation}
    H = \frac{1}{2} \sum_{p, \sigma} k^2_p a^\dagger_{p, \sigma} a_{p, \sigma} - \frac{4\pi}{\Omega}\sum_{p\neq q, j, \sigma} \parens{\zeta_j \frac{e^{ik_{q-p} \cdot R_j}}{k^2_{p-q}}} a^\dagger_{p, \sigma} a_{q, \sigma} + \frac{2\pi}{\Omega} \sum_{(p,\sigma)\neq (q, \sigma '), \nu \neq 0} \frac{a^\dagger_{p, \sigma} a^\dagger_{q, \sigma '} a_{q+\nu, \sigma '} a_{p -\nu, \sigma}}{k_\nu^2}
\end{equation}
where the $j$th nuclei has position $R_j$ and atomic number $\zeta_j$, and $k_\nu=\frac{2\pi \nu}{\Omega^{\nicefrac{1}{3}}}$ with cell volume $\Omega$ and $\sigma$ containing both up and down spins. For the derivation of this Hamiltonian see Appendix B of \cite{babbush2018low}. Conveniently, OpenFermion also provides simple functions to quickly generate this Hamiltonian, and we do so in the momentum plane wave basis with periodic boundary conditions. We elect not to use the more compact plane wave dual basis representation presented in \cite{babbush2018low}, due to the fact that we are using this Hamiltonian as a benchmark, rather than studying the outputs of the simulation. For the composite simulation, Jellium provides many Hamiltonian terms and a very sharply peaked distribution (see Figure  for a system of size equal to that of the spin models we study. Given that system size is more of a limiting factor than term number in our numerical study, this presents an opportunity to see how a composite channel performs on a system with greater $L$. To limit the system size we also use a spinless model, and then perform the Jordan-Wigner transformation on the second quantized Hamiltonian to represent our Hamiltonian as a sum of Pauli operators. This Hamiltonian is constructed with the necessary transformations using OpenFermion \cite{mcclean2020openfermion}.

\begin{figure}[h!]
    \centering
        \begin{subfigure}[b]{.49\textwidth}
            \includegraphics[width=1\textwidth]{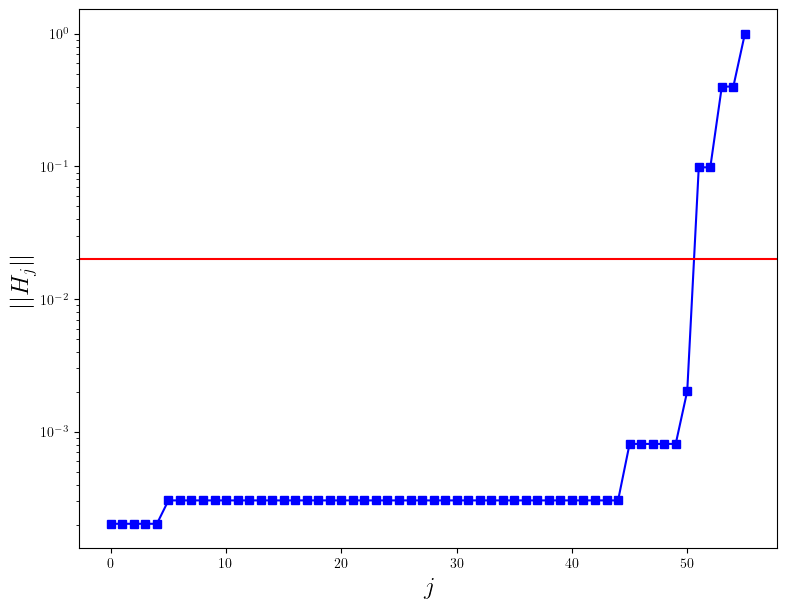}
            \caption{}
        \end{subfigure}
        \begin{subfigure}[b]{.49\textwidth}
            \includegraphics[width=1\textwidth]{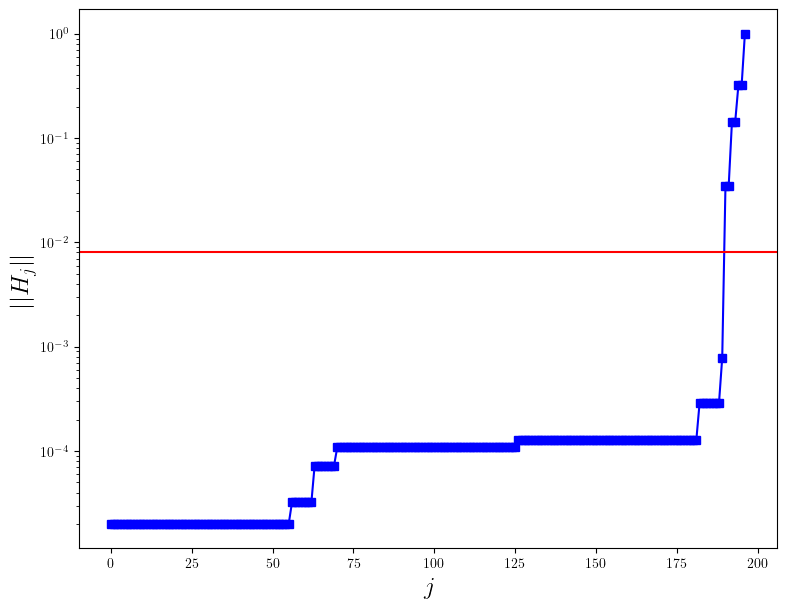}
            \caption{}
        \end{subfigure}
        \caption{\textit{Jellium Spectral Norm Distribution:} Semi-log plots of the sorted normalized spectral norms versus Hamiltonian index for 5 and 7 site Jellium models in figures (a) and (b) respectively. The plots show the increases in number of terms as well as how the distributions become increasingly sharply peaked. In red we provide a potential choice of $\omega_c$ for the partitioning heuristic.} \label{fig:Jelliumspec}
\end{figure}
\FloatBarrier

\subsubsection{Graph Hamiltonian Model}
The Hamiltonian we explore here involves a spin chain imposed on a lattice $\Gamma \in \mathbb{Z}^D$ with a graph distance metric $dist(\mathbf{u},\mathbf{v}) = |\mathbf{u}-\mathbf{v}|_1$ where $\mathbf{u}$ and $\mathbf{v}$ are vectorized coordinates on the graph with dimension equal to $D$. For our investigation we only examine lattices with $D=1$, given that for a fixed number of sites, this gives the most sharply peaked spectral distribution than any other $D$:

\begin{equation}
    H = \sum_{i>j} e^{-dist(i,j)}\alpha_{ij} \sigma^x_i \sigma^x_j + \sum_k \beta_k \sigma_k^z
\end{equation}

This system is similar to the quantum transverse field Ising model but with interactions that fall of exponentially with graph distance. The coefficients $\alpha_{ij}$ and $\beta_k$ are site dependant coupling constants that allow for the introduction of more disorder and/or structure in the Hamiltonian. To add some disorder to the model, we draw these coefficients pseudo-randomly from a Gaussian distribution with mean 0 and variance 1. \\

\subsubsection{Heisenberg Model}
The Heisenberg model describes a quantum spin system in a magnetic field with nearest-neighbour interactions. The Hamiltonian takes the following form:
\begin{equation} \label{eq:heisenberg}
    H = \sum_j \left(J_x \sigma_j^x \sigma_{j+1}^x + J_y \sigma_j^y \sigma_{j+1}^y + J_z \sigma_j^z \sigma_{j+1}^z \right)+ \sum_i B_z \sigma_i^z
\end{equation}
Here $B_z$ is the strength of the magnetic field in the $z$ direction and $J_{\{x, y, z\}}$ are coupling constants. Given the intuition of our composite channel, we expect this model to take advantage of our algorithm when the coupling constants largely differ in magnitude such that partitioning into Trotter and QDrift takes advantage of more Hamiltonian structure. Furthermore, introducing site-dependant coupling constants or writing down a highly disordered spin system could further add structure that the algorithm can take advantage of. A Hamiltonian of this nature would look something like the following: 
\begin{equation} \label{eq:spin_glass}
    H = \sum_j \left(J_x^{(j)} \sigma_j^x \sigma_{j+1}^x + J_y^{(j)} \sigma_j^y \sigma_{j+1}^y + J_z^{(j)} \sigma_j^z \sigma_{j+1}^z \right) + \sum_i B_z \sigma_i^z
\end{equation}
This Hamiltonian appears somewhat contrived for extracting the performance of the local composite channel. However, this system closely resembles the Edwards-Anderson model of a spin glass, a system of interest in condensed matter physics. In an attempt to create a sharp distribution, we simply sample $J_\nu^{(j)}$ from an exponential distribution with a scale parameter of 0.1. \\

\subsection{Partitioning Schemes} \label{sec:partition_scheme}
The main difficulty with deploying the composite simulation framework concerns  finding a good partitioning. In introducing composite simulations, Hagan and Wiebe suggested partitioning schemes derived on the basis of optimizing analytic cost functions both in deterministic and probabilistic settings \cite{hagan2022composite}. Here, we take a different approach involving the exact calculation of the simulation error, and an optimization routine that, given convergence, finds the optimal partitioning and gate count with respect to a chosen error tolerance and simulation time. This approach is used to answer the question regarding the best savings one can hope to achieve when deploying composite methods to simulate a specific Hamiltonian. We are not proposing this as a pre-processing routine, as it has complexity greater than that of the simulation itself, which is trivial as the optimization involves solving the simulation problem recursively. In addition, we also arrive at simple heuristics that can be used to partition certain Hamiltonians with little overhead, which we do propose as a strategy for using a composite approach. \\

\textit{Chop} is the partition that we introduce in this work. The idea is based on the heuristic of placing a few terms with larger spectral norms into Trotter-Suzuki channels and numerous small terms into QDrift, assuming that the Hamiltonian presents this structure. We start by sorting the terms by their spectral weights and introducing a ``chop threshold" $\omega_c \in [0, \max_i h_i]$. This scale will determine the partition such that if a term has spectral norm $h_i\geq\omega_c$ then $H_i\rightarrow A$ if $h_i < \omega_c$ then $H_i\rightarrow B$. Now we can express the error tolerance $\epsilon$ as a function of channel iterations $r$, with partitioning chop $\omega_c$ and a sample number $N_b$ that will be chosen in an optimization routine, and for a fixed initial state and time:
\begin{equation}\label{eq:errorfunc}
    \norm{(\mathcal{X}^{2k})^{\circ r}(\rho, t/r, \omega_c, N_b) - \evolchan{\rho, t}}_1 = \epsilon(r).
\end{equation}
By fixing an error tolerance for $\epsilon$, the exact cost of the simulation becomes a black box function with no closed form expression:
\begin{equation}\label{eq:cost_function}
    C_{comp} = f(\epsilon_{thresh}, r, \omega_c, N_b).
\end{equation}
This is the cost function we wish to minimize. However, we cannot do that by conventional methods such as with direct gradient descent. Also, with no strong intuition for a choice of $N_B$, if we wish to optimize this parameter we have to deal with integer optimization as well. The iterations $r$, however, while an integer, does not require optimization, but rather emits a search problem. If we allow an optimizer to pick initial random values for $N_b$ and $\omega_c$ from a fixed interval, then we must find the value $r$ required to meet the error threshold $\epsilon_{thresh}$, which will ultimately be determined by the optimizer's choice of the other two parameters. To complete this, we perform an exponential search on $r$ until we find some $r$ where $\epsilon(r) \geq \epsilon_{thresh}$ and set this as an upper bound on $r$. We then perform a binary search to find the smallest value of $r$ required to meet this condition and count the number of gates in the channel. This is a very expensive function given that we are precisely building the composite channel, applying it to the density matrix initial state in the problem, and counting the gates applied in each iteration of the search. The expensive nature emerges due to the sheer number of matrix multiplications required in performing this task, not in the search for $r$, which is nearly optimal. Note the importance of using the trace distance in this approach as it guarantees monotonicity of $\epsilon(r)$, which makes the search possible. This is not so in the framework of sampling the quantum infidelity, as finding the cost here would require other statistical methods (see Appendix \ref{sec:appendix_error}. \\

Now a glaring question left unanswered is the choice of an optimizer. We implement the Gradient-Boosted Regression Trees (GBRT) algorithm included in Sci-Kit Optimize \cite{pedregosa2011scikit}. This algorithm is specifically-designed to handle the optimization of very expensive functions. It is also convenient for our purposes given that it can handle both integer and real optimization parameters simultaneously. At a high level, the algorithm works by using a series of decision trees with an associated loss function. The decision trees perform regression to fit the input function and are iteratively generated based on the minimization of the loss function via gradient descent. This optimizer and cost function \ref{eq:cost_function} can then be easily generalized to the local composite channels where now we have an $N_b$ and $\omega_c$ for each blocking. As the number of local blocks grows, the optimization routine will need to take a larger number of input parameters in this prescription. However, the size of the system becomes classically intractable long before we would consider using this many local blocks, so this is far from a concern. \\

In some cases, models may exhibit a partitioning that is somewhat canonical and can lead to excellent performance of composite methods. This occurs when we have a Hamiltonian that fits naturally into the intuition behind the algorithm, such that we have a set $A$ containing large terms with small commutators and a set $B$ with small terms that are known not to commute in general. We are, therefore, proposing to use the chop partition but by choosing the chop threshold $\omega_c$ based on physical intuition regarding the Hamiltonian, rather than some expensive optimization routine. A perfect example of such a system is a Heisenberg model with weak coupling. In this case, looking at Equation \ref{eq:heisenberg}, we would set the chop threshold $\omega_c = \max\{J_x, J_y, J_z\}$, which implies we simulate the interactions with QDrift $\{J_\nu \sigma_j^\nu \sigma_{j+1}^\nu\}_{\nu = x,y,z} \rightarrow B$, and simulate the site energy terms with Trotter-Suzuki $\{B_z \sigma_j^z\} \rightarrow A$. In this way, the terms in the set $A$ all commute with each other, whereas the terms in the set $B$ are guaranteed to have a small spectral norm. We bring numerical evidence that this provides computational advantages in the sections below. In general, any system with perturbative interactions may benefit from this framework, given that the commutators within the system are small, as they will avail  this canonical partitioning. In cases where the partitioning is not as obvious, as is the case with {\rm H$_3$} and Jellium, we can achieve similar advantages by choosing $\omega_c = \max \frac{d \norm{H_j}}{dj} \textbf{s.t.} \norm{B} \geq \norm{A}$, meaning we sweep an ordered list of the Hamiltonian spectral norms and track the largest difference between terms, chopping the list where this occurs, given $j\geq \frac{L}{2}$. The final condition is just to ensure the majority of terms are simulated by QDrift. We also use this strategy throughout \ref{subsec:performance} and show advantages. 

\subsection{Error Measures} \label{sec:Error_Measure}
In this investigation, there is some arbitrariness in the error measure one can choose in order to quantify the performance of a simulation channel. In order to evaluate the resources required by an algorithm, one must evaluate the number of gates required to meet a certain $\epsilon$, which is calculated by said error measure. In the literature, this $\epsilon$ is often quantified by the diamond distance utilized in previous sections. However, while analytically convenient, for any reasonably-sized system, computing this quantity becomes computationally expensive. While it is possible to evaluate it efficiently, this requires finding the solution of a semi-definite program, which is much less efficient than using some other error measures. In addition, since we are not constrained to analytically solvable expressions or closed form equations with our numerical methods, we can optimize this cost in terms of some partition scheme. This is the idea behind the optimal chop partition, and doing so requires frequent computations of $\epsilon$. With this in mind, the error of our algorithm should be a quantity that we can compute in a reasonable amount of time while also being a fair error measure. The criteria for ``fairness" comes within the error measures treating each algorithm that comprises the composite channel on an equal footing. For example, if we are to optimize the partitioning with respect to the gate cost (which is dependent on $\epsilon$), then if Trotter is more performant with respect to QDrift in one error measure than in another, then our composite optimizer will favor Trotter, which will be reflected in the partition. As a result the total cost of the composite channel will be skewed by the error measure used. \\

We consider the infidelity and trace distance as possible measures of $\epsilon$. For definitions and an additional discussion regarding the scaling and complexity of computing these quantities, see Appendix \ref{sec:appendix_error}. Analytics are required in order to answer the question of which measure might provide a more fair comparison. More specifically, we can ask the question of how the error measures will scale with respect to the total simulation time $t$, and then test these results numerically. In terms of the infidelity, we provide the following Theorems:
\begin{restatable}[QDrift Infidelity Time Scaling]{prop}{qdinfscaling} \label{thm:qdinfscaling}
    Given a QDrift channel $\qdchan{\rho, t}$ and the standard evolutionary channel $\capU(\rho, \nicefrac{t}{N})$, for a density matrix $\rho$, time $t/N$, then the infidelity between the outputs of the channels $1- \fidel(\qdchan{\rho, t}, \capU(\rho, t/N)) \in \bigo{t^2}$.
 \end{restatable}

 \begin{restatable}[Trotter-Suzuki Infidelity Time Scaling]{prop}{TSfidelity} \label{thm:TSfidelity}
Given a Trotter-Suzuki channel $\tschan{2k}{\rho, t}$ and the standard evolutionary channel $\evolchan{\rho, t}$, for a density matrix pure state $\rho$ and time $t$, then the infidelity between the outputs of the channels $1 - \fidel(\tschan{2k}{\rho, t} \capU(\rho, t)) \in \bigo{(t^{2k+1})^2}$.
\end{restatable}

The proofs of these theorems are included in  Appendix \ref{sec:appendix_error}. Here, we obtain the non-trivial result that the infidelity is squared when considering Trotter-Suzuki formulas. This would lead to our optimized chop algorithm to heavily favor this channel over QDrift, and for this reason, we consider it an "unfair" error measure. On the other hand, if we consider how the trace distance scales with simulation time in both algorithms, we obtain the following Theorems:

\begin{restatable}[QDrift Trace Distance Time Scaling]{prop}{QDtracedist} \label{thm:QDtracedist}
Given a QDrift channel $\qdchan{\rho, t}$ and the standard evolutionary channel $\capU(\rho, \nicefrac{t}{N})$, for an arbitrary density matrix $\rho$ and time $t/N$, then the trace distance between the outputs of the channels $\tracedist(\qdchan{\rho, t}, \capU(\rho, t/N)) \in \bigo{t^2}$.
\end{restatable}

\begin{restatable}[Trotter-Suzuki Trace Distance Time Scaling]{prop}{TStracedist} \label{thm:TStracedist}
    Given a Trotter-Suzuki channel $\tschan{2k}{\rho, t}$ and the standard evolutionary channel $\evolchan{\rho, t}$, for an arbitrary density matrix $\rho$ and time $t$, then the trace distance between the outputs of the channels $\tracedist(\tschan{2k}{\rho, t}, \capU(\rho, t)) \in \bigo{t^{2k+1}}$.
\end{restatable}

Here, we see that no such squaring occurs, and the expected time-scaling is obtained. For this reason, we compute the entire density matrix and  $\epsilon$ using the trace distance in all of our numerical simulations. Proofs of the above theorems, as well as further discussions can again be found in Appendix \ref{sec:appendix_error}.

\subsection{Performance Results} \label{subsec:performance}
In this section, we first numerically analyze the real time quantum algorithm given by Hagan and Wiebe \cite{hagan2022composite} and then show equivalent numerical calculations of the imaginary time classical case. To accomplish this, we provide cost plots in which we provide the minimum $C$, or the number of rotation gates to achieve a desired simulation accuracy $\epsilon$ (calculated by the trace distance) for each point in time $t$ or $\beta$. To reiterate, here we exactly compute entire evolution channels with a random initial state $\rho$ sampled from the unit hyper-sphere, and directly apply and count gates. We conclude this section with a brief discussion about numerical studies for the local composite simulation algorithms. In these plots, we study variants of the composite channel and display results with the aforementioned notation with the addition of a tilde over the channel if the partition and $N_B$ have been optimized with GBRT. 
For example, a composite channel with inner-order 2 and outer-order 1 with an optimized partition and number of QDrift samples $N_B$ is written like so $\widetilde{\mathcal{X}}^{2,1}$. \\

Throughout the section, we normalize $\norm{H} = 1$ and run simulations for times $t \in (0, \frac{3\pi}{2}]$ so as to ensure the system undergoes non-trivial dynamics without overlapping the phase. This is done due to the fact that Trotter formulas have a periodic error for $\norm{H}t\geq 2\pi$, and running simulations in this range would lead to the optimizer finding the ``good points", where the error happens to be small, which would provide a very low cost simulation and a sharp drop in the cost trend. We also report the cost advantages achieved on \textit{crossover points}, which are values of $t$ such that $C_{QD}(t) = C_{TS}(t)$. We denote the composite channel advantage as $\xi \coloneqq C_{QD}(t') / C_\mathcal{X}(t') = C_{TS}(t') / C_\mathcal{X}(t')$. As we are unable to compute these times exactly we use interpolation methods to report values of $\xi$. This is motivated by the fact that analytics suggest this to be the point of greatest advantage for a composite channel \citenum{hagan2022composite}. This is intuitive, especially given higher-order Trotter-Suzuki formulas which are known to asymptotically outperform QDrift for large $t$, whilst QDrift is dominant in the small $t$ limit, suggesting a region where their strengths can be combined. 

\begin{table}[htbp!]
    \centering
    \begin{tabular}{| c | c | c | c |}
    \hline
        Hamiltonian & $\xi$ & Num. Terms & Time \\
        \hline
        Hydrogen-3 & 2.3 & 62 & Real \\
        5 Site Jellium & 9.2 & 56 & Real \\
        6 Site Jellium & 18.8 & 94 & Real \\
        7 Site Jellium & 10.4 & 197 & Real \\
        7 Spin Graph & 4.1 & 49 & Real \\
        8 Spin Graph & 3.9 & 64 & Real \\
        \hline
        8 Spin Heisenberg & 3.1 & 29 & Imag. \\
        Hydrogen-3 & 2.3 & 62 & Imag. \\
        6 Site Jellium & 18.8 & 94 & Imag. \\
        \hline
    \end{tabular}
    \caption{Summary of gate cost improvements observed (contingent on optimization convergence). We observe that savings tend to somewhat improve as the number of terms increases (within the same model), with the exception of Jellium 7 where optimizer struggles with partitioning due to the number of terms. This is evident in the lack of monotonicity of $C(\widetilde{\mathcal{X}^1}$ in Figure \ref{fig:Jellium56}. The most significant savings are seen for the Jellium models. Even in cases where the number of terms are comparable to other models, larger advantages are persistent in Jellium. This is further establishes the spectral norm distribution as one of the most important indicators of performance in the composite framework.}
    \label{tab:numerics_results}
\end{table} 
\FloatBarrier

\subsubsection{Real-Time Composite Quantum Simulations}
Beginning with Hydrogen chains, our results are highlighted in Figure \ref{fig:H3}. This plot reveals two interesting features: for long-time simulations, heuristics can be found that essentially match the performance of an optimized formula simply by inspecting the distribution of the norms of individual summands $\norm{H_j}$, and with the optimized routine, we find a significant improvement at the crossover point with $\xi = 2.3$. These plots begin flat for most of the simulation channels, which for the most part, indicates that one application of the channel achieves the desired $\epsilon$ for multiple sequential simulations at small times. This is expected, and is especially common with Trotter-Suzuki formulas, given that with one iteration they apply at least $L$ gates depending on the order, while QDrift provides the option of sampling single gates.

\begin{figure}[htbp!]
    \centering
    \includegraphics[width=0.6\textwidth]{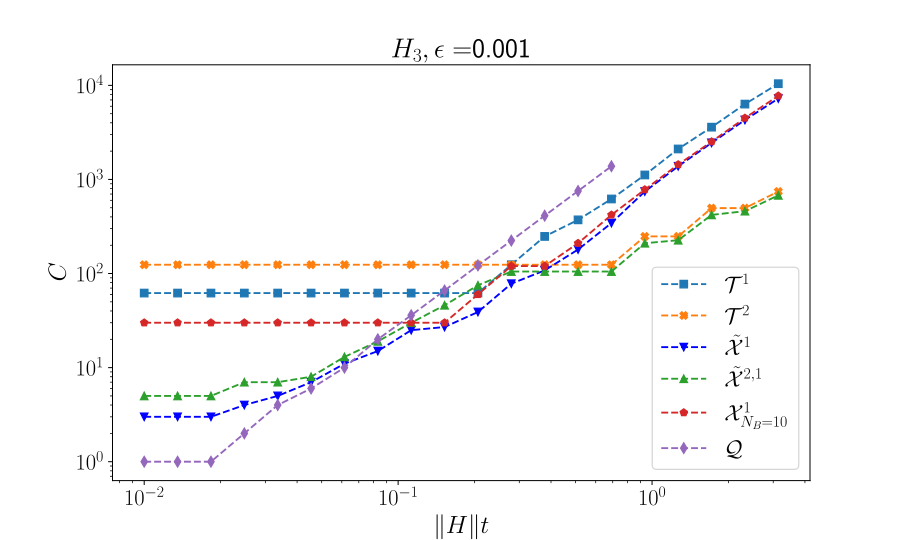}
    \caption{\textit{H$_3$ Cost Plot Simulation (Real-Time): }Log-log gate cost plot for the {\rm H$_3$} Hamiltonian generated with OpenFermion using three-dimensional Gaussians in a minimal basis. The bond distance is chosen to be that which minimizes the energy surface of {\rm H$_2$}, which is $\approx 0.8$ angstrom. We achieve a crossover advantage of $\xi = 2.3$, as well as remaining constant factor advantages at long times.} \label{fig:H3}
\end{figure} 
\FloatBarrier

Figure \ref{fig:H3} is interesting given that both the heuristic and optimized channels provide significant advantages at the crossover point, as well as the heuristic partition seems to match the asymptotic performance of the optimized channel. This demonstrates that optimization subroutines are not required to gain advantages in this framework. Furthermore, the second inner-order composite channel also shows a consistent advantage over the second-order Trotter channel. \\

To gain insight into effective choices of partition and QDrift samples, we also plot the optimized $N_B$ values and the ratio of Trotter terms to total terms $|A|/|H|$ with time in Figure \ref{fig:H3nbW}. Here, we find choices that somewhat agree with our prior intuition. For short times, places almost all terms into QDrift, and slowly increases $N_B$. As $t$ increases, more terms are placed into Trotter with the partition bouncing around in the regime where Trotter and QDrift have similar costs, which is also expected. The composite channel $\widetilde{\mathcal{X}}^{2,1}$  essentially places all terms into the Trotter simulator, given the favorable asymptotic performance of higher order Trotter formulas over QDrift, while $\widetilde{\mathcal{X}}^{1}$ finds a balance between the two at long times, likely due to their equivalent $t$ scaling. The most interesting behavior is that of $N_B$ at mid to long times. Here, $N_B$ peaks near the crossover point and then falls off as Trotter $t$-scaling becomes dominant in $\widetilde{\mathcal{X}}^{2,1}$. However, for the $\widetilde{\mathcal{X}}^{1}$ channel, $N_B$ experiences somewhat of a revival after the peak, and stabilizes at 15, which is about 24\% of the terms. We use this percentage to motivate future heuristic choices of $N_B$ in our investigation of Jellium in Figure $\ref{fig:Jellium56}$, which turns out to work quite well. 

\begin{figure}[h!] 
    \centering
        \begin{subfigure}[b]{.49\textwidth}
            \includegraphics[width=1\textwidth]{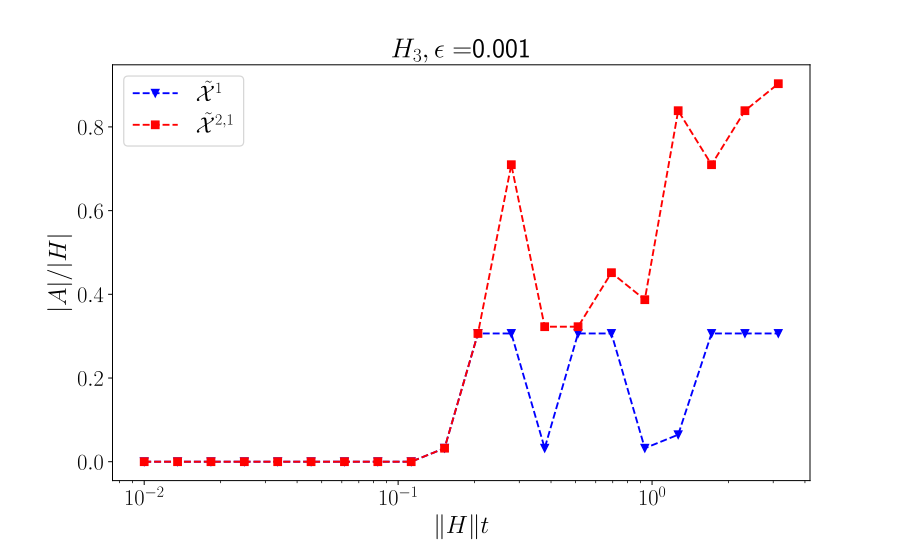}
            \caption{}
        \end{subfigure}
        \begin{subfigure}[b]{.49\textwidth}
            \includegraphics[width=1\textwidth]{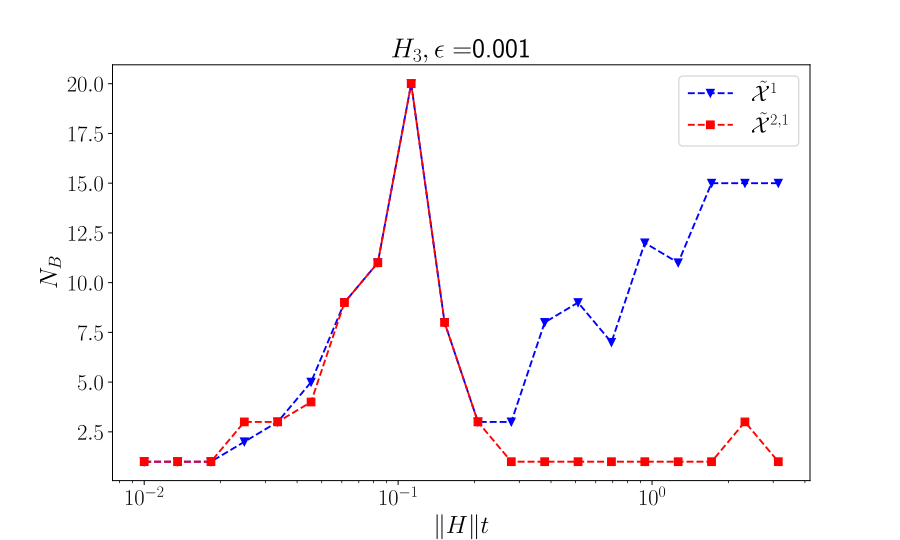}
            \caption{}
        \end{subfigure}
        \caption{\textit{Optimized H$_3$ Simulation Parameters:} Semi-log plots of parameters obtained by the GBRT optimization routine for the $\widetilde{\mathcal{X}^{1}}$ and $\widetilde{\mathcal{X}}^{2,1}$ real time channels. These parameter choices correspond to the H$_3$ simulation in Figure \ref{fig:H3}. In (a) we plot the cardinality of the $A$ set over the total number of terms, as a function of time. These values are dictated by GBRT optimized value of $\omega_c$. In (b) we present the equivalent plot with $N_B$.} 
        \label{fig:H3nbW}
\end{figure}
\FloatBarrier

When it comes to the simulation of Jellium, we find some of the most significant performance improvements within this section, including an order of magnitude cost difference at the crossover point (see Figure \ref{fig:Jellium56}). Specifically, in the case of 6-site Jellium, the Trotter and QDrift cost at the crossover point is approximately 100 gates, versus the composite channel, which achieves the same precision $\epsilon$ with only about 7 gates. Here, it is also shown that one can find an adequate partition that leads to advantages at longer times without the need for any optimization. This is also the only model whereby the optimization routine struggles to find optimal partitions in the neighbourhood of the crossover point. This leads to the Jellium 7 model inheriting a smaller $\xi$ than what is likely achievable. 
\begin{figure}[h!]
    \centering
        \begin{subfigure}[b]{.49\textwidth}
            \includegraphics[width=1\textwidth]{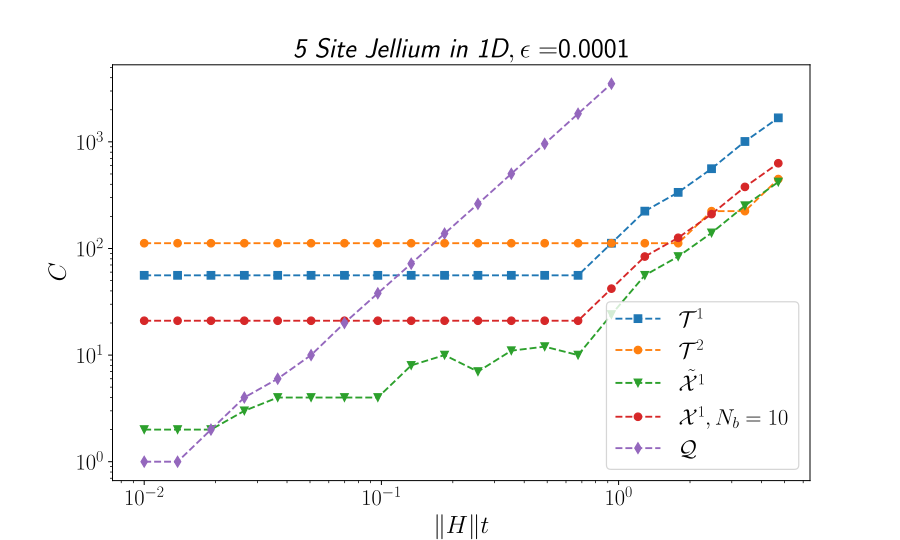}
            \caption{}
        \end{subfigure}
        \begin{subfigure}[b]{.49\textwidth}
            \includegraphics[width=1\textwidth]{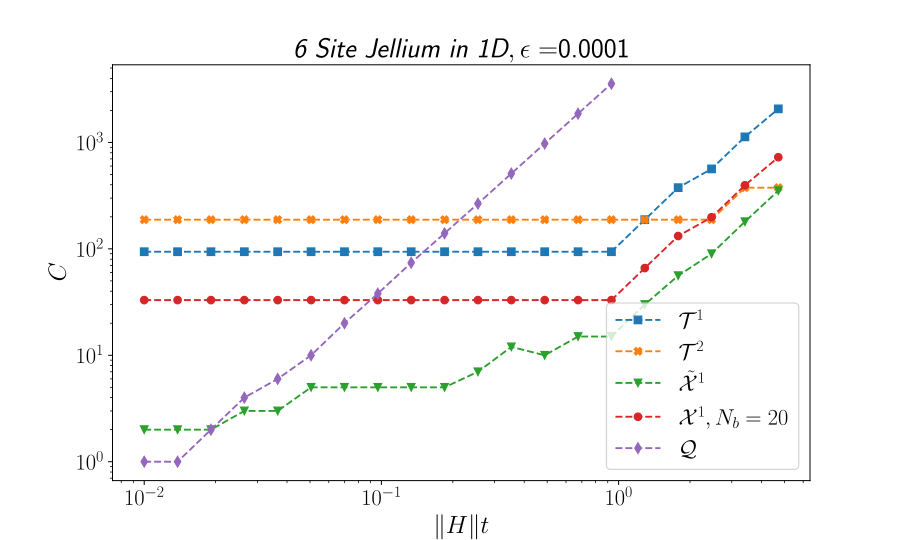}
            \caption{}
        \end{subfigure}
        \begin{subfigure}[b]{.49\textwidth}
            \includegraphics[width=1\textwidth]{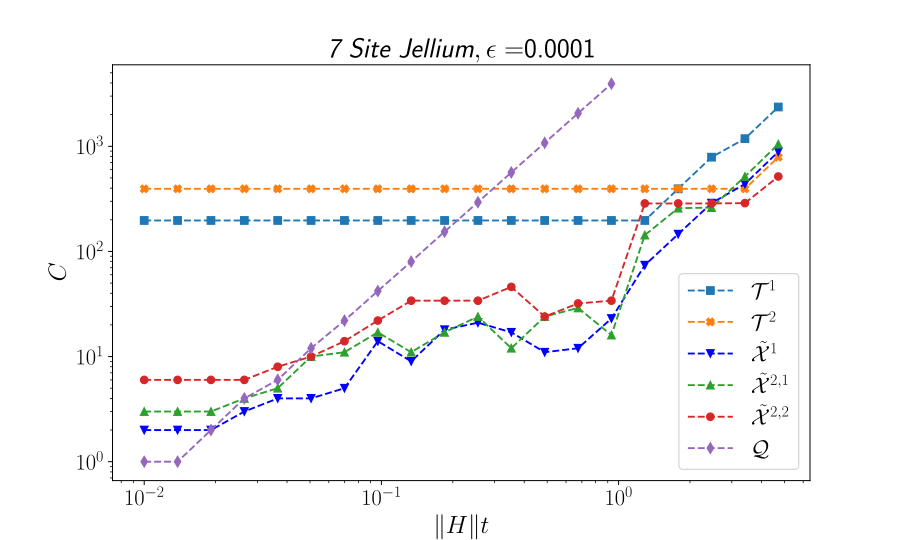}
            \caption{}
        \end{subfigure}
        \caption{\textit{Jellium Simulation Cost plots (Real-time): } Log-log cost plots of quantum simulations of Jellium with 5, 6, and 7 sites in (a), (b), and (c) respectively. In (a) and (b) we have in red, a chop heuristic where no optimization overhead is used. The distribution of Hamiltonian terms is chopped immediately before $\max \frac{d \norm{H_j}}{dj}$, and approximately $\frac{1}{5}L$ terms are sampled. This heuristic works quite well, although it is outperformed by the optimized version, especially at short times. In (a) we achieve $\xi = 9.2$. In (b) we achieve an impressive advantage of $\xi =  18.8$, the largest of all our real time results. In (c), we perform a similar analysis of the 7 site model with some higher order composite channels, but find that the optimizer has increased difficulty with larger numbers of Hamiltonian terms. Here $\xi = 10.4$, but through inspecting some neighbouring points of the crossover region, it surely has the potentially to be much larger.} \label{fig:Jellium56}
\end{figure}
\FloatBarrier

The final system we investigate in this section is that of the graph toy model with exponentially decaying interactions, which is a beyond nearest neighbour model. In Figure \ref{fig:graph_sim}, we study this model for chains of length 7 and 8, and find essentially identical behaviour. When moving from 7 to 8 spins, we only add 15 more terms to the Hamiltonian, which is clearly not enough to see any significant advantages. In fact, the crossover advantage is slightly smaller for the bigger model, but this could also be due to the optimizer not fully converging. 

\begin{figure}[htbp!]
\centering
    \begin{subfigure}[b]{.49\textwidth}
        \includegraphics[width=1\textwidth]{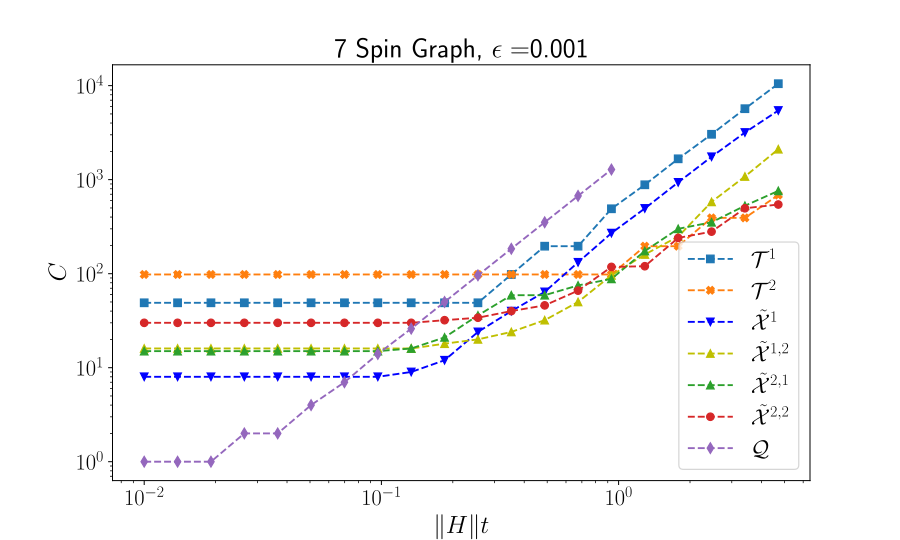}
        \caption{} 
    \end{subfigure}
    \begin{subfigure}[b]{.49\textwidth}
        \includegraphics[width=1\textwidth]{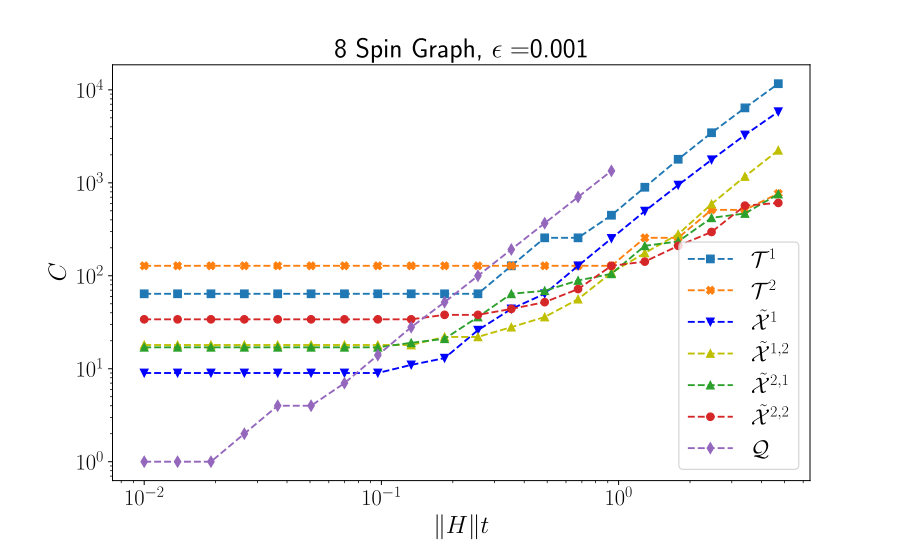}
        \caption{} 
    \end{subfigure}
    \caption{\textit{Toy Spin Graph Model Cost Plot:} In (a) the we have a cost plot for the 7 spin model where we obtain a crossover advantage of $\xi = 4.1$, which is fairly significant. The plot has multiple regions where different composite channels are optimal. In (b) we have the 8 spin model where we establish $\xi = 3.9$. This advantage, as well as the channel performance is almost identical to (a).}
    \label{fig:graph_sim}
\end{figure} 
\FloatBarrier

\subsubsection{Imaginary-Time Composite Channels} \label{subsubsec:iTime_results}
This section contains cost plots of simulations of the same aforementioned Hamiltonians, but with imaginary time propagators, where we present cost plots of the most interesting Hamiltonians from the previous section.\\

For simulations of the Heisenberg model, we find similar advantages to those in real time. In Figure \ref{fig:imag_sim}, we see that our proposed heuristic leads to an advantage over Trotter-Suzuki and QDrift in the regions of interest. What is different about this plot is that the optimizer finds the same $N_B$ and partitioning for all short times. This is an artefact of both the Hamiltonian and how the optimizer is programmed. Since all the splitting (single site) terms have equal spectral norm, the optimizer is placed in an all or nothing scenario, as choosing $\omega_c < 1$ immediately places all terms into QDrift. Given that after receiving the cost, the program then solves a search problem to find the minimal $r$ to achieve $\epsilon$ precision, it is rare that it finds the ideal conditions to build a pure QDrift channel with $r=1$. However, significant savings are still achieved at the crossover point, and over $\mathcal{T}^1$ and $\mathcal{T}^2$ at large $\beta$. Once again, the validity of heuristic partitions are shown, specifically in the 1st order composite channel, which exactly matches its optimized version at large $\beta$. \\

\begin{figure}[htbp!]
    \centering
    \includegraphics[width=0.7\textwidth]{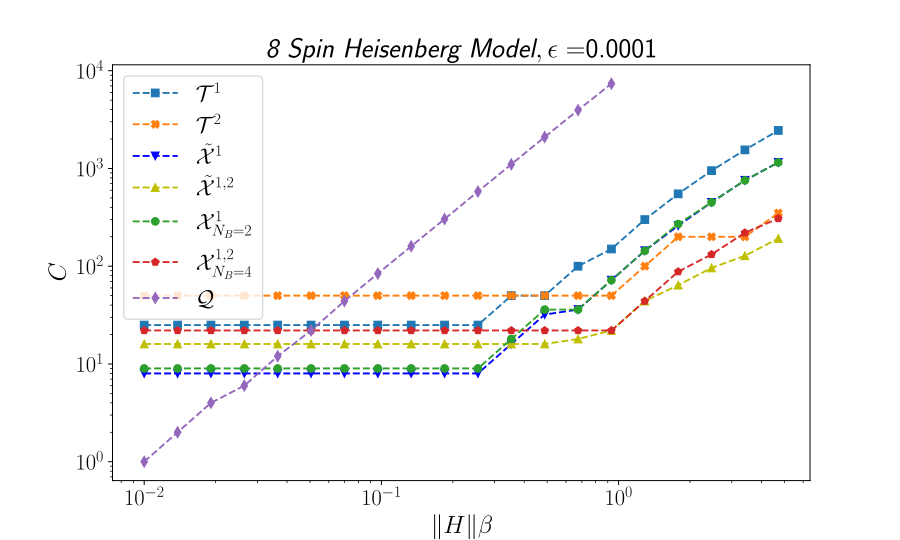}
    \caption{\textit{8 Spin Heisenberg Model Cost Plot:} In this imaginary time simulation we establish $\xi = 3.1$, as well as maintain advantages at large $\beta$. We also find that our chosen heuristics are essentially optimal at large $\beta$, with the green and dark blue lines overlapping.} \label{fig:imag_sim}
\end{figure} 
\FloatBarrier

For Hydrogen chains, we obtain strikingly similar results and compared with those in real time, as seen in Figure \ref{fig:iH3}. We once again obtain a significant crossover advantage, as well as constant factor advantages at large $\beta$, or low temperature. Heuristics are also shown to continue to hold in their effectiveness, in this case, from the crossover point and onward.

\begin{figure}[htbp!]
    \centering
    \includegraphics[width=0.7\textwidth]{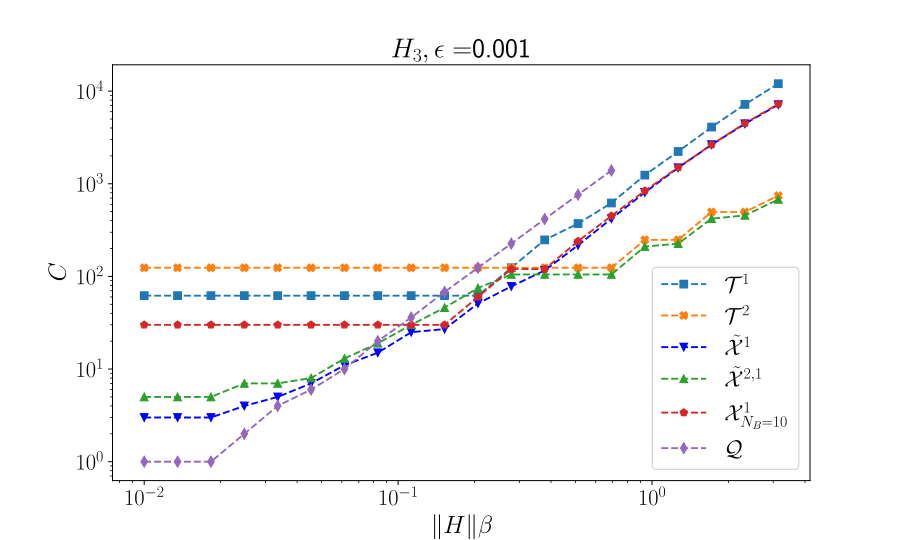}
    \caption{\textit{H$_3$ Simulation Cost Plot (Imaginary-Time):} The same parameters are used to build the 3 atom Hydrogen chain as we done in real time. We achieve incredibly similar results, and again recover the real time result of $\xi = 2.3$ that was achieved in Figure \ref{fig:H3}.} \label{fig:iH3}
\end{figure} 
\FloatBarrier

For Jellium, we choose to investigate the system size with the best-behaved optimizer, as well as the largest $\xi$, which occurs for the 6-site model. In imaginary time, we once again reproduce a significant advantage, shown in Figure \ref{fig:iJellium6}. As in the case with {\rm H$_3$}, this plot is quite similar to the real time case in Figure \ref{fig:Jellium56}. However, here at large $\beta$ the composite channel seems to do better in imaginary time given that even the first order composite channel (with optimization) outperforms the second order Trotter channel. This happens in the final point of the plot where $\mathcal{T}^2$ is no longer in the``flat-regime". While this is very interesting, it is unclear analytically why this occurs, and we would likely not expect this trend to continue asymptotically. 

\begin{figure}[htbp!]
    \centering
    \includegraphics[width=0.7\textwidth]{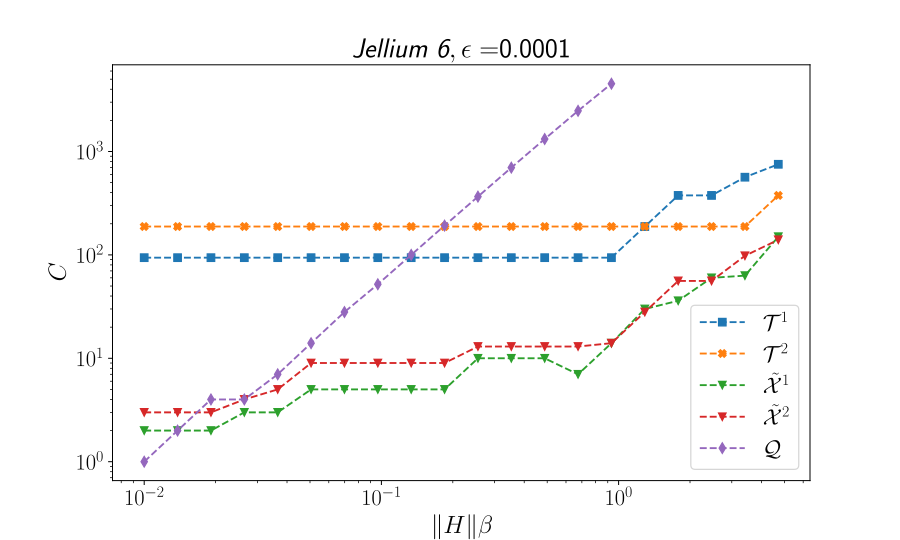}
    \caption{\textit{6 Site Jellium Simulation Cost Plot (Imaginary-Time):} Here we recover the  crossover advantage of $\xi = 18.8$ from the real time simulation in Figure \ref{fig:Jellium56}, which is also the largest advantage achieved in our imaginary time simulations. We additionally achieve advantages over second order Trotter at large $\beta$.} \label{fig:iJellium6}
\end{figure} 
\FloatBarrier

Overall, this section nicely complements some of the analytics in Section \ref{sec:analysis} both by reinforcing the fact that composite quantum channels allow for similar advantages in both real and imaginary time, as well as through calculation of exact constant factor advantages. In other words, this section provides convincing evidence on the applicability of composite formulas to classical imaginary-time Monte-Carlo algorithms.

\subsubsection{Local Composite Quantum Channels}
Distinct from the previous two algorithms, with the local composite channels we do not immediately expect to see significant simulation advantages for the small systems we can compute. Recall, this algorithm makes use of the Lieb-Robinson velocity $v_{LR}$ that limits the propagation of information and thus correlations in a local lattice model. In our above numerics, our lattices contain $\leq 8$ sites, meaning even for small $v_{LR}$, the lattice can still become quickly entangled. In this section, it is important to understand where the observed advantages are originating, whether they are from the local decomposition, or something else. For example, the results in prior sections already suggest that composite channels can outperform Trotter and QDrift channels in certain regimes. If a block decomposition is introduced, we pay a small gate cost to break the simulation into subsets (as gates on the boundary are applied more than once), but the advantages of the composite simulation are almost guaranteed to outweigh this cost. Thus we wish to find a regime in which we can perform calculations of exact costs with local composite channels that explicitly gain advantages via block decomposition. Otherwise, we will observe essentially the same behaviour as before, but with slightly smaller constant factors. There are two ways to go about achieving this; one is to add more sites to the model, which quickly becomes computationally intractable with standard methods. The second strategy is to decrease the coupling between sites in the lattice, which naturally decreases $v_{LR}$. This is the strategy we utilize. In reference to Equation \ref{eq:heisenberg} we perform our cost plots on Heisenberg models with 8 sites with $B_z = 1$, and $J_\nu^{(j)}$ sampled from an exponential distribution with a scale parameter (serving as a coupling constant) of 0.00005. To allow for a fair simulation, we then choose $\epsilon = 0.000001$, such that statistically, $~98\% $ of terms will be greater than $\epsilon$, which can be seen from a simple integration of the PDF. Results of this simulation are shown in Figure \ref{fig:local_heisen}.  \\

\begin{figure}[htbp!]
    \centering
    \includegraphics[width=0.7\textwidth]{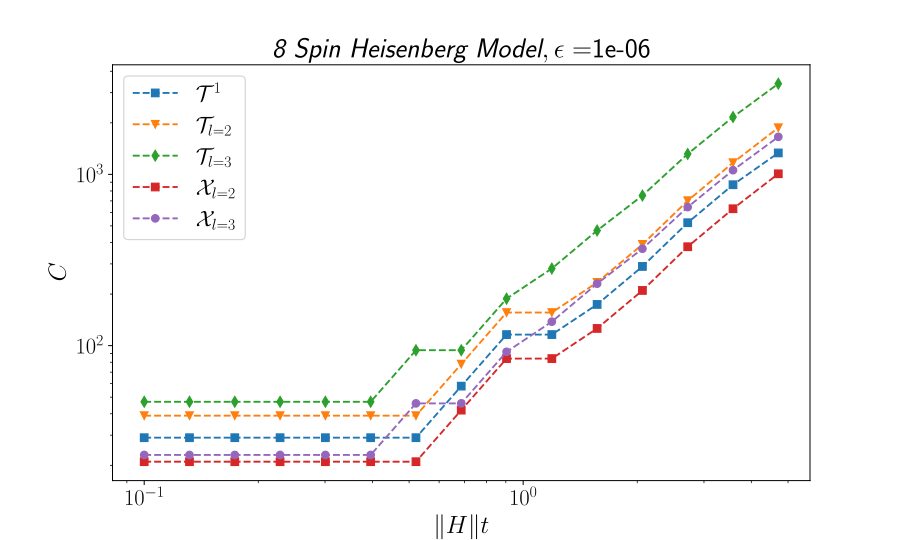}
    \caption{\textit{8 Spin Heisenberg Model Simulated by localized and standard channels:} This cost plot compares Trotter to its  localized versions, with the subscript $l$ indicating the overlap of the boundary region in the block-local simulation.  Channels with no $l$ subscript are the standard simulations from prior sections. As expected, we observe that the composite channel is most efficient, however, given that the standard Trotter algorithm outperforms the localized Trotter channel, we can conclude that we are not in the regime where locality is providing advantages. The same heuristics were used for the composite channel as in Figure \ref{fig:imag_sim}, with $N_B = (4,1,4)$ on lattice subsets $(A, Y, B)$.} \label{fig:local_heisen}
\end{figure} 
\FloatBarrier

Given the gap between $\mathcal{T}^1$ and $\mathcal{T}_{l=2}$ in this very weak coupling regime, it is unclear whether our methods (exact gate counts) provide the means for investigating advantages gained by locality. In Ref. \cite{haah2021quantum}, via computations of bounds, numerics did not show advantages (in the form of T-gates) until approximately 100 sites were included in a more strongly-coupled model with $J_\nu^{(j)} \in [-1, 1]$ sampled i.i.d., so our results with far fewer spins are not unexpected. However, we still theoretically expect to see advantages in the limit where system sizes are large, and we can take advantage of the Lieb-Robinson bound.

\section{Discussion and Summary}
\label{sec:summary}
The main contributions of this work are the extension of the composite channel to imaginary time and to local Hamiltonians, and the construction of a composite simulation library for numerically evaluating algorithmic performance. The imaginary-time bounds provided advocate for the potential use of QDrift and Composite QDrift-product formulas in classical simulations where they can potentially improve the efficiency/accuracy of Quantum Monte-Carlo simulations. This is possible via shortening the path length of discrete time path integrals in these calculations given that our algorithms can approximate $e^{-\beta H}$ with fewer propagators than Trotter-Suzuki formulas in most cases. Our results also highlight that depending on the choice of $\beta$, one may also perform an imaginary time simulation using QDrift, particularly for high temperature physics. It is an open question as to whether these methods can be used to improve the overall accuracy before one runs into the sign problem. Nevertheless, these methods are interesting as one can use physical information about the Hamiltonian to choose a partitioning that possibly grants a computational advantage.  
\\

Furthermore, the extension of the composite channel to local Hamiltonians provides an interesting approach to lattice simulations in which there is a priori knowledge about the lattice structure. The diamond distance bound on this algorithm is significant in that the cost of establishing local composite blocks seemingly minimal in terms of asymptotics, making it a potentially-useful algorithm. However, it seems more powerful numerical methods may be required to explore the regimes in which this approach shines, given the size of the systems likely required. An interesting future trade-off to be explored here is the size of the blocking $l$, and the size of the commutator sums in Theorem \ref{thm:localcomp_diamond}. Furthermore, one can investigate optimizations of partitions and QDrift samples on lattice subsets. Despite our library containing these methods, we did not employ them as they can lead to misleading figures. Here, one must tread cautiously, as adding more blocks with parameterized weights and QDrift samples can lead to an optimizer finding lower costs, similar to how adding hidden units to a neural network can lead to the over-fitting of a data set in machine learning. In this scenario, it may be possible to find localized composite channels that outperform the standard channels, simply because the increase in parameters allowed us to find a more tailored partitioning.\\ 

One of the difficulties with the techniques introduced in Ref. \citenum{hagan2022composite} is finding good parameters to yield cost improvements. Our numeric results seem to suggest that even heuristic approaches for parameters, such as $N_B$, can yield at least factors of 2 reductions for small Hamiltonians, as seen in Fig. \ref{fig:H3}. The difficulties of multiple parameter optimization becomes more of an issue when moving to the local simulation framework. With $m$ blocks there are at least $2m$ parameters that need to be determined with a chop partitioning scheme. However, simple heuristics still lead to performance advantages over localized Trotter-Suzuki formulas in Figure \ref{fig:local_heisen}. The fact that said heuristics provide advantages with only an 8 spin model suggests that cost savings should scale favorably with increased lattice sizes. \\

In Section \ref{sec:Error_Measure}, a series of short but significant proofs were also provided to remind the reader to be cautious in examining error within hybridized algorithms. In our case, the infidelity treated QDrift differently than Trotter by a square, and using it would have skewed the partitions and provided inaccurate costs on the composite channel. \\

Finally, the library built for this project is easy to implement. It is also flexible, and it reliably evaluates the exact number of operator of exponentials $e^{-iH_jt}$ required to execute on a quantum computer (or $e^{-iH_j \beta}$ for classical) to time evolve a system within a given error tolerance. To reiterate, the reason this is important is because the cost of these algorithms is highly dependent on the Hamiltonian, specifically in the number and size of the terms, and their commutators. Additionally, the costs and errors discussed in this paper are derived from analytic upper bounds on the diamond distance and induced Schatten $1 \rightarrow 1$ norm, which are often bounded with repeated applications of the triangle inequality and using $\max$ operations, and can therefore be loose. These exact numerics can capture the true costs and errors of these algorithms, and advantages can motivate further studies into these approaches. This library is available upon reasonable request to the authors.

\section{Acknowledgements}
M.P. was supported through NSERC discovery grants, N.W. and M.H. were supported using grants from US DOE National Quantum Information Science Research Centers, Co-design Center for Quantum Advantage (C2QA) under contract number DE-SC0012704 as well as support from Google Inc.  
D.S. acknowledges support from an NSERC Discovery Grant and the Canada Research Chair program. JC acknowledges support from the Natural Sciences and Engineering Research Council (NSERC), the Shared Hierarchical Academic Research Computing Network (SHARCNET), Compute Canada, and the Canadian Institute for Advanced Research (CIFAR) AI chair program. Resources used in preparing this research were provided, in part, by the Province of Ontario, the Government of Canada through CIFAR, and companies sponsoring the Vector Institute \url{www.vectorinstitute.ai/#partners}.

\bibliographystyle{unsrt}
\bibliography{bib.bib}

\begin{widetext}

\appendix
\section{Proofs of Imaginary Time Bounds} \label{sec:imag_proof}
\subsection{Imaginary time QDrift}

\imagQD* 

\begin{proof}
Following the supplementary material of \cite{campbell2019random}, we will make use of the Liouvillian super-operator formalism for the imaginary time channel. In imaginary time we have:

\begin{equation}
    \mathcal{L}(\rho) = \{H, \rho\} = H \rho + \rho H,
\end{equation}
which generates the following channel: 
\begin{equation}
    \frac{e^{-\beta \mathcal{L(\rho)}}}{\Tr e^{-\beta \mathcal{L(\rho)}}} = \frac{e^{-\beta H} \rho e^{-\beta H}}{\Tr(e^{-\beta H} \rho e^{-\beta H})}.
\end{equation}

Throughout this section, whenever $\liouv$ appears, it is implied that we mean $\liouv(\rho)$, we use the former for more compact notation. The difference between $\liouv$ defined here and that in \cite{campbell2019random} is that the commutator becomes an anti-commutator, and the $i$ vanishes due to the nature of the imaginary time channel. Equivalence can be seen via the Taylor series expansion of $e^{-\beta \liouv(\rho)}$, where $\liouv^n$ means $n$ applications of the super-operator: $\{H, \{H, ... , \{H, \{H, \rho \}  \} ... \}\}$ From \cite{campbell2019random}, this super-operator yields the following properties: 

\begin{equation} \label{eq:liouv_prop1}
    \indone{\liouv} \leq 2 \norm{H} \leq 2\lambda .
\end{equation}

Similarly, if we define $\liouv_j$ that generate unitaries under each $H_j$ as in \cite{campbell2019random}, and given that the Hamiltonian terms are normalized by $h_j$ the following also holds:

\begin{equation} \label{eq:liouv_prop2}
    \indone{\liouv_j} \leq 2\norm{H_j} = 2 .
\end{equation}

We are now ready to proceed with the proof. To avoid cumbersome notation we will write $\liouv(\rho)$ simply as $\liouv$. Let us first examine the induced 1-norm of the difference between exact channel $\mathcal{U}_N$ that acts for imaginary time $\nicefrac{\beta}{N}$ and the QDrift channel $\mathcal{Q}$, and the generalize this for $N$ compositions of the QDrift channel representing $N$ samples:

\begin{align}
    \indone{\frac{\mathcal{U}_N}{\Tr \mathcal{U}_N} - \frac{\mathcal{Q}}{\Tr \mathcal{Q}}} &= \indone{\frac{e^{-\nicefrac{\beta}{N} \liouv}}{\Tr e^{-\nicefrac{\beta}{N} \liouv}} - \frac{\sum_i \frac{h_i}{\lambda} e^{-\nicefrac{\beta \lambda}{N} \liouv_i}}{\Tr(\sum_j \frac{h_j}{\lambda} e^{-\nicefrac{\beta \lambda}{N} \liouv_j})}}\\
    &=  \indone{\frac{e^{-\nicefrac{\beta}{N} \liouv}}{1+\Tr(\sum_{n=1}^\infty \frac{(-\beta \liouv)^n}{N^n n!})} - \frac{\sum_i \frac{h_i}{\lambda} e^{-\nicefrac{\beta \lambda}{N} \liouv_i}}{1 + \Tr(\sum_j \frac{h_j}{\lambda} {\sum_{m=1}^\infty \frac{(-\beta \liouv_j)^m}{N^m m!}})}} \\
    &\leq \indone{\frac{e^{-\nicefrac{\beta}{N} \liouv}}{1+\Tr(\sum_{n=1}^\infty \frac{(-\beta \liouv)^n}{N^n n!})} - \frac{\sum_i \frac{h_i}{\lambda} e^{-\nicefrac{\beta \lambda}{N} \liouv_i}}{1+\Tr(\sum_{n=1}^\infty \frac{(-\beta \liouv)^n}{N^n n!})}} \nonumber \\
    \text{  }&+ \indone{\frac{\sum_i \frac{h_i}{\lambda} e^{-\nicefrac{\beta \lambda}{N} \liouv_i}}{1+\Tr(\sum_{n=1}^\infty \frac{(-\beta \liouv)^n}{N^n n!})} - \frac{\sum_i \frac{h_i}{\lambda} e^{-\nicefrac{\beta \lambda}{N} \liouv_i}}{1 + \Tr(\sum_j \frac{h_j}{\lambda} {\sum_{m=1}^\infty \frac{(-\beta \liouv_j)^m}{N^m m!}})}} \\
    &\coloneqq a + b
\end{align}

Now we shall proceed by dealing with each of these terms individually. The denominator of the $a$ can be expanded via its geometric series:

\begin{align}
    a &= \indone{\parens{\sum_{k=0}^\infty \bigl (-\Tr(\sum_{n=1}^\infty \frac{(-\beta \liouv)^n}{N^n n!}) \bigr)^k}  (e^{\nicefrac{-\beta \liouv}{N}} - \sum_i \frac{h_i}{\lambda} e^{\nicefrac{-\beta \liouv_i}{N}})} \\
    &= \abs{\sum_{k=0}^\infty (-\Tr\bigl (\sum_{n=1}^\infty \frac{(-\beta \liouv)^n}{N^n n!}) \bigr)^k} \indone{e^{\nicefrac{-\beta \liouv}{N}} - \sum_i \frac{h_i}{\lambda} e^{\nicefrac{-\beta \lambda \liouv_i}{N}}} \\ 
\end{align}

We will first focus on the leftmost factor. To bound this term we will apply the triangle inequality to apply the norm to each term in the sum. We will also use the fact that the absolute value of the trace is trivially upper-bounded by the induced 1-norm: $\abs{\Tr \liouv(\rho)} \leq \indone{\liouv(\rho)}$. Here it is understood that this absolute value is maximized over all $\rho : \norm{\rho}=1$. We then use the sub-multiplicative property $\indone{\liouv^n}\leq \indone{\liouv}^n$:

\begin{align}
    a &\leq \sum_{k=0}^\infty \bigl (\sum_{n=1}^\infty \frac{\beta^n \abs{\Tr(\liouv)^n}}{N^n n!} \bigr )^k \indone{e^{\nicefrac{-\beta \liouv}{N}} - \sum_i \frac{h_i}{\lambda} e^{\nicefrac{-\beta \lambda \liouv_i}{N}}} \\
    &\leq \sum_{k=0}^\infty \bigl (\sum_{n=1}^\infty \frac{\beta^n \indone{\liouv}^n}{N^n n!} \bigr )^k \indone{e^{\nicefrac{-\beta \liouv}{N}} - \sum_i \frac{h_i}{\lambda} e^{\nicefrac{-\beta \lambda \liouv_i}{N}}} \\
    &\leq \sum_{k=0}^\infty \bigl (\sum_{n=1}^\infty \frac{\beta^n 2^n \lambda^n}{N^n n!} \bigr )^k \indone{e^{\nicefrac{-\beta \liouv}{N}} - \sum_i \frac{h_i}{\lambda} e^{\nicefrac{-\beta \lambda \liouv_i}{N}}} \\
    &= \sum_{k=0}^\infty (e^{\nicefrac{2 \lambda \beta }{N}} -1)^k \indone{e^{\nicefrac{-\beta \liouv}{N}} - \sum_i \frac{h_i}{\lambda} e^{\nicefrac{-\beta \lambda \liouv_i}{N}}} \\
    & = \frac{\indone{e^{\nicefrac{-\beta \liouv}{N}} - \sum_i \frac{h_i}{\lambda} e^{\nicefrac{-\beta \lambda \liouv_i}{N}}} }{2 - e^{\nicefrac{2 \beta \lambda}{N}}}
\end{align}

Where we require $|e^{\nicefrac{2 \beta \lambda}{N}} - 1| < 1$. Next, inspecting the numerator, this is of the form bounded in the Appendix of \cite{campbell2019random}, therefore, the following holds trivially:
\begin{equation}
    a \leq \frac{4 \beta^2 \lambda^2}{N^2} \frac{e^{\nicefrac{2 \lambda \beta }{N}} }{2 - e^{\nicefrac{2 \beta \lambda}{N}}}
\end{equation}

Next, let us proceed to bound the $b$ term:
\begin{align}
    b &= \indone{\frac{\sum_i \frac{h_i}{\lambda} e^{-\beta \liouv_i}}{1+\Tr(\sum_{n=1}^\infty \frac{(-\beta \liouv)^n}{n!})} - \frac{\sum_i \frac{h_i}{\lambda} e^{-\beta \liouv_i}}{1 + \Tr(\sum_j \frac{h_j}{\lambda} {\sum_{m=1}^\infty \frac{(-\beta \liouv_j)^m}{m!}})}} \\
    &= \indone{\sum_i \frac{h_i}{\lambda} e^{-\beta \liouv_i}} \abs{\sum_{k=0}^\infty \bigl (-\Tr(\sum_{n=1}^\infty \frac{(-\beta \liouv)^n}{n!}) \bigr)^k - \sum_{l=0}^\infty \bigl (-\Tr(\sum_j \frac{h_j}{\lambda} \sum_{n=1}^\infty \frac{(-\beta \liouv_j)^n}{n!}) \bigr)^l} \\
    &\leq e^{\nicefrac{2\beta \lambda}{N}} \abs{\sum_{k=0}^\infty \bigl (-\Tr(\sum_{n=1}^\infty \frac{(-\beta \liouv)^n}{n!}) \bigr)^k - \sum_{l=0}^\infty \bigl (-\Tr(\sum_j \frac{h_j}{\lambda} \sum_{n=1}^\infty \frac{(-\beta \liouv_j)^n}{n!}) \bigr)^l}
\end{align}
Now, carefully expanding the double sums up to second order:

\begin{equation}
    = e^{\nicefrac{2\beta \lambda}{N}} \abs{ \begin{multlined} \sum_{k=2}^\infty \bigl (-\Tr(\sum_{n=2}^\infty \frac{(-\beta \liouv)^n}{n!}) \bigr)^k - \sum_{l=2}^\infty \bigl (-\Tr(\sum_j \frac{h_j}{\lambda} \sum_{n=2}^\infty \frac{(-\beta \liouv_j)^n}{n!}) \bigr)^l + \\ \Tr \sum_j \frac{h_j}{\lambda} \sum_{m=2}^\infty \frac{(-\beta \lambda \liouv_j)^m}{N^m m!} - \Tr \sum_{n=2}^\infty \frac{(-\beta \lambda \liouv)^n}{N^n m!} + \Tr \frac{\beta \liouv}{N} - \Tr \sum_j \frac{h_j}{\lambda} \frac{\beta \lambda \liouv_j}{N} \end{multlined}} 
\end{equation}

Using the fact that $\sum_j h_j \liouv_j = \liouv$, the final two terms cancel and we have:

\begin{align}
    &= e^{\nicefrac{2\beta \lambda}{N}} \Biggr|\sum_{k=2}^\infty \bigl (-\Tr(\sum_{n=2}^\infty \frac{(-\beta \liouv)^n}{N^n n!}) \bigr)^k - \sum_{l=2}^\infty \bigl (-\Tr(\sum_j \frac{h_j}{\lambda} \sum_{n=2}^\infty \frac{(-\beta \liouv_j)^n}{N^n n!}) \bigr)^l \nonumber\\
    &\qquad+ \Tr \sum_j \frac{h_j}{\lambda} \sum_{m=2}^\infty \frac{(-\beta \lambda \liouv_j)^m}{N^m m!} - \Tr \sum_{n=2}^\infty \frac{(-\beta \lambda \liouv)^n}{N^n m!} \Biggr| \\
    &\leq e^{\nicefrac{2\beta \lambda}{N}} \medmath{\parens{ \abs{\sum_{k=2}^\infty \bigl (-\Tr(\sum_{n=2}^\infty \frac{(-\beta \liouv)^n}{N^n n!}) \bigr)^k} + \abs{\sum_{l=2}^\infty \bigl (-\Tr(\sum_j \frac{h_j}{\lambda} \sum_{n=2}^\infty \frac{(-\beta \liouv_j)^n}{N^n n!}) \bigr)^l} + \abs{\Tr \sum_j \frac{h_j}{\lambda} \sum_{m=2}^\infty \frac{(-\beta \lambda \liouv_j)^m}{N^m m!}} + \abs{\Tr \sum_{n=2}^\infty \frac{(-\beta \lambda \liouv)^n}{N^n m!}} }} \\
    &\leq e^{\nicefrac{2\beta \lambda}{N}} \medmath{\parens{ \sum_{k=2}^\infty \bigl (\sum_{n=2}^\infty \frac{\beta^n \abs{\Tr{\liouv}^n}}{N^n n!} \bigr)^k + \sum_{l=2}^\infty \bigl (\sum_j \frac{h_j}{\lambda} \sum_{n=2}^\infty \frac{\beta^n \abs{\Tr{\liouv}^n}}{N^n n!} \bigr)^l + \sum_j \frac{h_j}{\lambda} \sum_{m=2}^\infty \frac{\beta^m \lambda^m \abs{\Tr{\liouv_j}^n}}{N^m m!} + \sum_{n=2}^\infty \frac{\beta^n \lambda^n \abs{\Tr{\liouv}^n}}{N^n m!} }} \\
    \end{align}
    
    Now we once again use the property $\abs{\Tr \liouv(\rho)} \leq \indone{\liouv(\rho)}$:
    
    \begin{align}
    &\leq e^{\nicefrac{2\beta \lambda}{N}} \medmath{\parens{ \sum_{k=2}^\infty \bigl (\sum_{n=2}^\infty \frac{\beta^n \indone{\liouv}^n}{N^n n!} \bigr)^k + \sum_{l=2}^\infty \bigl (\sum_j \frac{h_j}{\lambda} \sum_{n=2}^\infty \frac{\beta^n \indone{\liouv_j}^n}{N^n n!} \bigr)^l +  \sum_j \frac{h_j}{\lambda} \sum_{m=2}^\infty \frac{\beta^m \lambda^m \indone{\liouv_j}^m}{N^m m!} +  \sum_{n=2}^\infty \frac{\beta^n \lambda^n \indone{\liouv}^n}{N^n m!} }} \\
    &\leq e^{\nicefrac{2\beta \lambda}{N}} \parens{ \sum_{k=2}^\infty \bigl (\sum_{n=2}^\infty \frac{2^n \beta^n \lambda^n}{N^n n!} \bigr)^k + \sum_{l=2}^\infty \bigl (\sum_j \frac{h_j}{\lambda} \sum_{n=2}^\infty \frac{2^n \beta^n \lambda^n}{N^n n!} \bigr)^l +  \sum_j \frac{h_j}{\lambda} \sum_{m=2}^\infty \frac{2^m \beta^m \lambda^m}{N^m m!} +  \sum_{n=2}^\infty \frac{2^n \beta^n \lambda^n}{N^n m!} } \\
    &\leq e^{\nicefrac{2\beta \lambda}{N}} \parens{ 2 \sum_{k=2}^\infty \bigl (\frac{2\beta^2 \lambda^2}{N^2} e^{\nicefrac{2\lambda \beta}{N}} \bigr)^k +  2\frac{2\beta^2 \lambda^2}{N^2} e^{\nicefrac{2\lambda \beta}{N}}}\\
    &=  2 \frac{\frac{4\beta^4 \lambda^4}{N^4} e^{\nicefrac{6\lambda \beta}{N}}}{1-\frac{2\beta^2 \lambda^2}{N^2} e^{\nicefrac{2\lambda \beta}{N}}} +  2\frac{2\beta^2 \lambda^2}{N^2} e^{\nicefrac{4\lambda \beta}{N}}\\
\end{align}

Where in the last 3 lines we used both the tail-bound on an exponential sum; $\sum_{n=2}^\infty \frac{x^n}{n!} \leq \frac{x^2}{2} e^x$ and the tail-sum of the geometric series; $\sum_{n=2}^\infty x^n = \frac{x^2}{1-x}$. The latter of the two requires that $|x|<1$ for convergence, so we therefore, require $|\frac{2\beta^2 \lambda^2}{N^2} e^{\nicefrac{2\lambda \beta}{N}}| < 1$, which is not a strong requirement given we are interested in large $N$. Now combining the results for $a$ and $b$ we obtain the final result for a channel that simulates a fraction $\nicefrac{1}{N}$ of the dynamics:
\begin{align}
    \indone{\mathcal{U}_N - \Lambda_{QD}} &\leq a + b \\
    &\leq \frac{4 \beta^2 \lambda^2}{N^2} \frac{e^{\nicefrac{2 \lambda \beta }{N}} }{2 - e^{\nicefrac{2 \beta \lambda}{N}}} + 2 \frac{\frac{4\beta^4 \lambda^4}{N^4} e^{\nicefrac{6\lambda \beta}{N}}}{1-\frac{2\beta^2 \lambda^2}{N^2} e^{\nicefrac{2\lambda \beta}{N}}} +  2\frac{2\beta^2 \lambda^2}{N^2} e^{\nicefrac{4\lambda \beta}{N}}
\end{align}

Using the fact that the norm is sub-additive under composition
\begin{align}
   \indone{\mathcal{U}_N^{\circ N} - \Lambda_{QD}^{\circ N}} &\leq N \indone{\mathcal{U}_N- \Lambda_{QD}}, 
\end{align}
we are left with the following result for the full-time channel
\begin{align} \label{eq:pre-linear}
    \indone{\mathcal{U} - \Lambda_{QD}^{\circ N}} &\leq \frac{4 \beta^2 \lambda^2}{N} \frac{e^{\nicefrac{2 \lambda \beta }{N}} }{2 - e^{\nicefrac{2 \beta \lambda}{N}}} + 2 \frac{\frac{4\beta^4 \lambda^4}{N^3} e^{\nicefrac{6\lambda \beta}{N}}}{1-\frac{2\beta^2 \lambda^2}{N^2} e^{\nicefrac{2\lambda \beta}{N}}} +  2\frac{2\beta^2 \lambda^2}{N} e^{\nicefrac{4\lambda \beta}{N}}.
\end{align}

Next, we wish to linearize this bound for a clearer interpretation of the error. We can do so by utilizing the conditions required to hold in the proof: 
\begin{equation} \label{eq:qd_cond1}
    |e^{\nicefrac{2 \beta \lambda}{N}} - 1| < 1 ,
\end{equation}

\begin{equation} \label{eq:qd_cond2}
    |\frac{2\beta^2 \lambda^2}{N^2} e^{\nicefrac{2\lambda \beta}{N}}| < 1. 
\end{equation}
 
The first condition (\ref{eq:qd_cond1}) gives $e^{\nicefrac{2\beta \lambda}{N}} \in (0,2)$ which implies $\frac{2\beta \lambda}{N} \in (-\infty, \ln{2})$, and given that we require this quantity necessarily be positive we have $\frac{2\beta \lambda}{N} \in [0, \ln{2})$. Now, our second condition (\ref{eq:qd_cond2}), given the positivity constraint, further implies $\frac{x^2}{2} e^x \in [0,1)$ with $x=\nicefrac{2 \beta \lambda}{N}$. Using max value obtained from Equation 
(\ref{eq:qd_cond1}), that $x < \ln{2}$, Equation (\ref{eq:qd_cond2}) reduces to $\frac{x^2}{2} e^x \leq 2\ln{2} \in [0, 1)$. Given that Equation (\ref{eq:qd_cond2}) is a product of two monotonic functions in the domain $(0,\infty)$, any value of $x$ less than the maximum obtained in the first condition also satisfies the second. Therefore, going forward we focus only on satisfying the constraint in Equation (\ref{eq:qd_cond1}).\\

Going forward in linearizing Equation \ref{eq:pre-linear} we will take $\lambda$ and $N$ to be constant and consider variable $\beta$, and continue by upper bounding each term. Beginning with the first term we have:

\begin{equation} \label{eq:qd_term1}
    \frac{4 \beta^2 \lambda^2}{N} \frac{e^{\nicefrac{2 \lambda \beta }{N}} }{2 - e^{\nicefrac{2 \beta \lambda}{N}}} = \frac{2 \beta^2 \lambda^2}{N} \frac{1}{y - \frac{1}{2}}, \; \; \mathbf{s.t.} \; y(\beta) = e^{-\nicefrac{2 \lambda \beta }{N}} \; \text{and} \;  \frac{2}{3} \leq y(\beta) \leq 1
\end{equation}

Now we rewrite the aforementioned constraint $1\leq e^{\nicefrac{2 \lambda \beta }{N}} < 2$, to be $1\leq e^{\nicefrac{2 \lambda \beta }{N}} \leq \frac{3}{2}$ to prevent one from approaching a divergence in the proof arbitrarily closely. Utilizing the following property:
\begin{equation}
    f(b) = f(a) + \int_a^b \frac{df}{dy} dy \leq f(a) + (b-a) \max_{y \in [a,b]}{\abs{\frac{df(y)}{dy}}} ,
\end{equation}

with $f(\beta) = \frac{1}{y(\beta) - \frac{1}{2}}$ Equation \ref{eq:qd_term1} is upper bounded in the following way:

\begin{equation}
    \frac{4 \beta^2 \lambda^2}{N} \frac{4\beta^2 \lambda^2}{N} e^{\nicefrac{4\beta \lambda}{N}} \leq \frac{4 \beta^2 \lambda^2}{N} (3 + 12 \frac{\lambda}{N})
\end{equation}

Moving on to the third term in Equation (\ref{eq:pre-linear}), we have:

\begin{equation}
    \frac{4\beta^2 \lambda^2}{N} e^{\nicefrac{4\beta \lambda}{N}} = \frac{4\beta^2 \lambda^2}{N} y^2, \; \; \mathbf{s.t.} \;\; y = e^{\nicefrac{2\beta \lambda}{N}} \; \text{and} \; 1\leq y(\beta) \leq \frac{3}{2} ,
\end{equation}

and using the same integral bound we obtain:

\begin{equation}
    \frac{4\beta^2 \lambda^2}{N} e^{\nicefrac{4\beta \lambda}{N}} \leq \frac{4\beta^2 \lambda^2}{N} (1+3 \frac{\lambda}{N}).
\end{equation}

Finally, we are left to bound the second term in Equation \ref{eq:pre-linear}, which we can write in the following way:

\begin{equation}
    \parens{\frac{4 \beta^2 \lambda^2}{N} e^{\nicefrac{4\lambda \beta}{N}}} \frac{\frac{4\beta^2 \lambda^2}{N^2} e^{\nicefrac{2\lambda \beta}{N}}}{2-\frac{4\beta^2 \lambda^2}{N^2} e^{\nicefrac{2\lambda \beta}{N}}} = 2 \parens{\frac{4 \beta^2 \lambda^2}{N} e^{\nicefrac{4\lambda \beta}{N}}} \frac{y\ln{y}}{1-y\ln{y}}, \; \; \mathbf{s.t.} \; y(\beta) = e^{\nicefrac{2\lambda \beta}{N}} \; \text{and} \; 1\leq y(\beta) \leq \frac{3}{2}.
\end{equation}

The function $f(y) = \frac{y\ln{y}}{1-y\ln{y}}$ is monotonic over the constrained domain and can be upper bounded by inserting the largest value that $\ln{y}$ achieves. Also, notice the term in parentheses is precisely the 3rd term in Equation \ref{eq:pre-linear} that we bounded above. Using these two observations:

\begin{align}
    2 \parens{\frac{4 \beta^2 \lambda^2}{N} e^{\nicefrac{4\lambda \beta}{N}}} \frac{y\ln{y}}{1-y\ln{y}} &\leq \frac{4 \beta^2 \lambda^2}{N}
    2 \parens{1 + 3 \frac{\lambda}{N}} \frac{y\ln{\frac{3}{2}}}{1-y\ln{\frac{3}{2}}} \\
    &\leq \frac{4 \beta^2 \lambda^2}{N} 2 \ln{\frac{3}{2}} \parens{1 + 3 \frac{\lambda}{N}} \parens{\frac{1}{\frac{2}{3} - \ln{\frac{3}{2}}} + \frac{2\lambda}{3N} \frac{1}{(\frac{2}{3} - \ln{\frac{3}{2}})^2}}
\end{align}
where in the last step we once again used the integral identity following the same procedure as we did for the first term. Now that we have linearized each term, we combine all the results to obtain:
\begin{equation}
    \indone{\mathcal{U} - \Lambda_{QD}^{\circ N}} \leq \frac{\beta^2 \lambda^2}{N}\parens{a + b\frac{\lambda}{N} + c\frac{\lambda^2}{N^2}}
\end{equation}
where $a\approx 28.41845 , \; b\approx 128.95110\; \text{and} \; c\approx 95.08717$. If we make the promise that $\frac{\lambda}{N} \leq 0.01$, the bound simplifies to 

\begin{equation}
     \indone{\mathcal{U} - \Lambda_{QD}^{\circ N}} \leq \frac{C\beta^2 \lambda^2}{N}
\end{equation}
where $C\approx 29.71747$, completing the proof.
\end{proof}

This condition is not very restrictive on the imaginary time $\beta$. Considering the condition given in Equation (\ref{eq:qd_cond1}), we have $\beta < \frac{\ln{2} N}{2\lambda} \approx 34.657$, so $\beta < 10\pi$, meaning that given this promise on the ratio $\frac{\lambda}{N}$, the bound holds for a simulation time that allows the phase to oscillate up to 5 times if we were dealing with real-time.
 

\subsection{Imaginary Time Trotter-Suzuki Formulas}
Following the procedure in \cite{hagan2022composite}, we must first convert the existing bound to a bound on channels and then bound the error with the renormalizing operations. \\

\imagTS*

\begin{proof}
First consider the induced 1-norm of the difference between the channels after a single iteration or time-step $r$. As in the previous proof, absolute values containing a channel operation are understood to be maximized over all inputs $\rho : \norm{\rho}_1 = 1$:
    \begin{align}
        &\indone{\frac{\evolchan{\rho, \nicefrac{\beta}{r}, \nicefrac{\beta}{r}}}{\Tr(\evolchan{\rho, \nicefrac{\beta}{r}, \nicefrac{\beta}{r}})} - \frac{\tschan{2k}{\rho, \nicefrac{\beta}{r}, \nicefrac{\beta}{r}}}{\Tr(\tschan{2k}{\rho, \nicefrac{\beta}{r}, \nicefrac{\beta}{r}})}} \leq \indone{\frac{\evolchan{\rho, \nicefrac{\beta}{r}}- \tschan{2k}{\rho, \nicefrac{\beta}{r}}}{\Tr \evolchan{\rho, \nicefrac{\beta}{r}}}} \nonumber\\
        &\qquad+ \indone{\tschan{2k}{\rho, \nicefrac{\beta}{r}} \parens{\frac{1}{\Tr \tschan{2k}{\rho, \nicefrac{\beta}{r}}}- \frac{1}{\Tr \evolchan{\rho, \nicefrac{\beta}{r}}}}} \\
        &= \abs{\frac{1}{\Tr \evolchan{\rho, \nicefrac{\beta}{r}}}} \indone{\evolchan{\rho, \nicefrac{\beta}{r}}-\tschan{2k}{\rho, \nicefrac{\beta}{r}}} + \indone{\tschan{2k}{\rho, \nicefrac{\beta}{r}}} \abs{\frac{\Tr(\evolchan{\rho, \nicefrac{\beta}{r}} - \tschan{2k}{\rho, \nicefrac{\beta}{r}})}{\Tr(\evolchan{\rho, \nicefrac{\beta}{r}}) \Tr(\tschan{2k}{\rho, \nicefrac{\beta}{r}})}} \\
        &\leq \abs{\frac{1}{\Tr \evolchan{\rho, \nicefrac{\beta}{r}}}} \indone{\evolchan{\rho, \nicefrac{\beta}{r}}-\tschan{2k}{\rho, \nicefrac{\beta}{r}}} \nonumber\\
        &\qquad+ \indone{\tschan{2k}{\rho, \nicefrac{\beta}{r}}} \indone{\evolchan{\rho, \nicefrac{\beta}{r}} - \tschan{2k}{\rho, \nicefrac{\beta}{r}}}\abs{\frac{1}{\Tr(\evolchan{\rho, \nicefrac{\beta}{r}}) \Tr(\tschan{2k}{\rho, \nicefrac{\beta}{r}})}} \\
        & = \indone{\evolchan{\rho, \nicefrac{\beta}{r}} - \tschan{2k}{\rho, \nicefrac{\beta}{r}}} \parens{\frac{\indone{\tschan{2k}{\rho, \nicefrac{\beta}{r}}} }{\abs{\Tr(\evolchan{\rho, \nicefrac{\beta}{r}}) \Tr(\tschan{2k}{\rho, \nicefrac{\beta}{r}})}} + \frac{1}{\abs{\Tr \evolchan{\rho, \nicefrac{\beta}{r}}}}}
    \end{align}
    where in the second line we simply cross-multiply the traces of the outputs of the channels (since they are just numbers), and in third line we use the fact that the absolute value of the trace is trivially upper bounded by the diamond norm. Going forward we need to lower bound the traces in the denominator, which can be done using the well-known von Neumann trace inequality: $\sum_{i=1}^N \sigma^A_i \sigma^B_{N-i+1} \leq \Tr(AB)$ where $\sigma^{(A, B)}_i$ represent the singular values of the matrix $A$ or $B$ respectfully, and they are ordered such that $\sigma_i \leq \sigma_{i+1}$. The absolute values in the denominator also need not be written as these traces are positive by definition. For the exact evolution channel we have:
    \begin{align}
        \Tr \evolchan{\rho, \nicefrac{\beta}{r}} &= \Tr(e^{-\beta H} \rho e^{-\nicefrac{\beta}{r} H}) \\
        &= \Tr (e^{-2\nicefrac{\beta}{r} H} \rho) \\
        & \geq \sum_i \sigma_i(e^{-2\nicefrac{\beta}{r} H}) \sigma_{N-i+1}(\rho) \\
        & \geq \sigma_{min}(e^{-2\nicefrac{\beta}{r} H}) \Tr(\rho) \\
        & \geq e^{-2\nicefrac{\beta}{r} \norm{H}}, 
    \end{align}
    where we use the fact that the trace of $\rho$ is 1, along with the fact that $\rho$ is positive semi-definite and hermitian, such that the sum of its singular values is equal to its trace. Next, we follow a similar procedure in dealing with the trace of the output of the Trotter-Suzuki channel:
    \begin{align}
        \Tr \tschan{2k}{\rho, \nicefrac{\beta}{r}} &= \Tr(\prodform(\nicefrac{\beta}{r}) \rho \prodform^\dagger(\nicefrac{\beta}{r})) \\
        &\geq \sigma_{min}\parens{\prodform^\dagger(\nicefrac{\beta}{r}) \prodform(\nicefrac{\beta}{r})} \Tr \rho \\
        & = \sigma_{min}\parens{\prodform^\dagger(\nicefrac{\beta}{r}) \prodform(\nicefrac{\beta}{r})} .
    \end{align}
    In order to understand how this term behaves with respect to $\nicefrac{\beta}{r}$ and $H$, we require more careful analysis. To do this it will be useful to introduce some more convenient notation. Just like an arbitrary product formula can be defined in the following way (following the notation of \cite{childs2021theory}):
    \begin{equation}
        \prodform(\nicefrac{\beta}{r}) = \prod_{\upsilon=1}^\Upsilon \prod_{l=1}^L e^{-\nicefrac{\beta}{r} s_{\upsilon , l} H_{\pi_{\upsilon(l)}}} . \\
    \end{equation}
    Where $s$ controls the pre-factor that modifies the time-step and $\pi$ controls the ordering of the $H$ summands in each stage of the formula. We utilize the following property of singular values: given two matrices $A, \; B \in \mathbb{C}$ the following holds $\sigma_n (A) \sigma_i (B) \leq \sigma_i (AB) \leq \sigma_1 (A) \sigma_i (B)$, given that singular values are ordered such that $\sigma_1 \geq \sigma_2 \geq ... \geq \sigma_n$. Also, note that since the individual evolution operators are no longer anti-hermitian, coupled with the fact that Trotter-Suzuki formulas are symmetric, $\prodform^\dagger = \prodform$. 

    \begin{align}
        \Tr \tschan{2k}{\rho, \nicefrac{\beta}{r}} &\geq \sigma_{min}\parens{\prodform^2(\nicefrac{\beta}{r})} \\ 
        & \geq \sigma^2_{min}\parens{\prod_{\upsilon=1}^\Upsilon \prod_{l=1}^L e^{-\nicefrac{\beta}{r} s_{\upsilon , l} H_{\pi_{\upsilon(l)}}}} \\
        &\geq \parens{\prod_{\upsilon=1}^\Upsilon \prod_{l=1}^L \sigma_{min}\parens{e^{-\nicefrac{\beta}{r} s_{\upsilon , l} H_{\pi_{\upsilon(l)}}}}}^2 \\
        &\geq \parens{\prod_{\upsilon=1}^\Upsilon \prod_{l=1}^L e^{-\nicefrac{\beta}{r} s_{\upsilon , l} \norm{H_{\upsilon(l)}}}}^2 \\
        &= \parens{e^{\sum_{\upsilon=1}^\Upsilon \sum_{l=1}^L-\nicefrac{\beta}{r} s_{\upsilon , l} \norm{H_{\upsilon(l)}}}}^2 \\
        &= e^{-2 \nicefrac{\beta}{r} \lambda}
    \end{align}
    where each inequality step is a repeated application of the aforementioned property. In the last step, we complete the using the property of Trotter-Suzuki formulas that $\sum_{\upsilon, l} s_{\upsilon, l} = 1$ and the definition $\sum_l \norm{H_l} = \lambda$. \\

    We also need to deal with the most important term which is the actual induced 1-norm of the difference between the channels. To do this, we will massage this term into a form where we can apply a result from \cite{childs2021theory}. To do this we follow the same procedure from Section \ref{sec:LR}:
    \begin{align}
        \indone{\capU(\beta) - \tschan{2k}{\beta}} :=& \indonedef{\capU(\beta) - \tschan{2k}{\beta} } \\
        =& \max_{\rho : \norm{\rho}_1 = 1} \norm{e^{-H\beta} \rho e^{- H \beta}  - \prodform(\beta) \rho \prodform  (\beta) }_1 \\
        \leq& \max_{\rho : \norm{\rho}_1 = 1} \norm{e^{-H\beta} \rho e^{- H \beta} -e^{- H \beta} \rho \prodform  (\beta) }_1 \\
        \text{ }& +\max_{\rho : \norm{\rho}_1 \leq 1} \norm{ e^{- H \beta}   \rho \prodform  (\beta)  - \prodform(\beta)  \rho \prodform  (\beta) }_1 \\
        =& \max_{\rho : \norm{\rho}_1 = 1} \norm{e^{-\beta H}\rho \parens{e^{- H \beta} - \prodform  (\beta)} }_1 + \max_{\rho : \norm{\rho}_1 = 1} \norm{\parens{e^{- H \beta} - \prodform(\beta)}  \rho \prodform(\beta)}_1 \\
        \leq & \parens{\norm{e^{-\beta H}}_\infty + \norm{\prodform(\beta)}_\infty} \norm{e^{- H \beta} - \prodform(\beta)}_\infty \max_{\rho : \norm{\rho}_1 \leq 1} \norm{\rho}_1 \\
        = & 2 e^{\beta \lambda} \norm{e^{- H \beta} - \prodform(\beta)}_\infty \\
        = & 2 e^{\beta \lambda} \norm{U(\beta) - \prodform(\beta)}_\infty .
    \end{align}
    Next we apply the following bound from Appendix E in \cite{childs2021theory}:
    \begin{equation}
        \norm{U(\nicefrac{\beta}{r}) - \prodform(\nicefrac{\beta}{r})}_\infty \leq 2\Upsilon^{2k+1} \frac{\alpha_{comm}(H,2k)}{(2k+1)!}  \frac{\beta^{2k+1}}{r^{2k+1}} e^{4\frac{\beta}{r} \Upsilon \lambda}
    \end{equation}
    Where $\Upsilon$ and $\lambda$ are again the number of stages of the product formula and the sum of the spectral norms of the Hamiltonian summands respectively. Now we are simply left with a single term to bound in the diamond norm of the Trotter-Suzuki channel. We can similarly define the Trotter-Suzuki channel as a composition of multiple stages of channels, each containing compositions of evolution channels corresponding to each Hamiltonian summand. We write this in the following way: 
    \begin{equation}
        \tschan{2k}{\nicefrac{\beta}{r}} \equiv \bigcirc_{\upsilon = 1}^\Upsilon \bigcirc_{l=1}^L e^{-\frac{\beta}{r} s_{\upsilon, l} \liouv_{\pi_{\upsilon(l)}} } ,
    \end{equation}
    where we once again have the Liouvillian generator for imaginary time $\liouv(\rho) = \{H, \rho \}$, and all other machinery is defined in the same was as above. In this way it is more natural to proceed in using the sub-multiplicative property of the norm:
    
    \begin{align}
        \indone{\tschan{2k}{\nicefrac{\beta}{r}}} &\leq \prod_{\upsilon=1}^\Upsilon \prod_{l=1}^L \indone{e^{-\frac{\beta}{r} s_{\upsilon, l} \liouv_{\pi_{\upsilon(l)} }}} \\
        &\leq \prod_{\upsilon=1}^\Upsilon \prod_{l=1}^L  e^{2\frac{\beta}{r} s_{\upsilon, l} \norm{H_{\pi_{\upsilon(l)}}}} \\
        &= e^{\sum_{\upsilon=1}^\Upsilon \sum_{l=1}^L 2\frac{\beta}{r} s_{\upsilon , l} \norm{H_{\upsilon(l)}}} \\
        &= e^{2 \frac{\beta}{r} \lambda}.
    \end{align}
    In the second line we use the Taylor expansion of the exponential of the Liouvillian to apply the diamond norm to the exponent as seen in the QDrift proof. We make use of the aforementioned fact that $\indone{\liouv_l} \leq 2h_l$ where $h_l$ is the spectral norm of the $H_l$ Hamiltonian term. Now, combining all of our results we are left with the following:
    \begin{equation} \label{eq:single_step_itime_trot}
        \indone{\frac{\evolchan{\rho, \beta}}{\Tr(\evolchan{\rho, \beta})} - \parens{\frac{\tschan{2k}{\rho, \nicefrac{\beta}{r}}}{\Tr(\tschan{2k}{\rho, \nicefrac{\beta}{r}})}}^{\circ r}} \leq 4\Upsilon^{2k+1} \frac{\alpha_{comm}(H,2k)}{(2k+1)!} \frac{\beta^{2k+1}}{r^{2k+1}} e^{4\frac{\beta}{r} \Upsilon \lambda} e^{\beta \lambda /r}\parens{e^{4\frac{\beta}{r} \lambda} e^{2\frac{\beta}{r} \norm{H}} + e^{2\frac{\beta}{r} \norm{H}}}
    \end{equation}
    and we wish to compare the exact evolution to the channel after $r$ iterations. In other words, we use the property $\indone{X^{\circ r} - Y^{\circ r}} \leq r \indone{X-Y}$, proven earlier (for the diamond norm, but it also holds here) in Equation \ref{eq:normsum}, 
    to obtain the following:
    \begin{equation}
        \indone{\frac{\evolchan{\rho, \beta}}{\Tr(\evolchan{\rho, \beta})} - \parens{\frac{\tschan{2k}{\rho, \nicefrac{\beta}{r}}}{\Tr(\tschan{2k}{\rho, \nicefrac{\beta}{r}})}}^{\circ r}} \leq 4\Upsilon^{2k+1} \frac{\alpha_{comm}(H,2k)}{(2k+1)!} \frac{\beta^{2k+1}}{r^{2k}} e^{4\frac{\beta}{r} \Upsilon \lambda} e^{\beta \lambda /r} \parens{e^{4\frac{\beta}{r} \lambda} e^{2\frac{\beta}{r} \norm{H}} + e^{2\frac{\beta}{r} \norm{H}}}
    \end{equation}
    Similar to the QDrift case, we wish to eliminate the exponential factors to give a more straightforward interpretation of scaling of the error bound. However, in this case the proof did not impose restrictions on the size of these exponentials. Therefore, we instead upper-bound each exponential term in a common way and then make on observation about the scaling of $r$, to show that the exponential factors are $\bigo{1}$. We begin by using the fact that $\norm{H} \leq \lambda$, and  $e^{\nicefrac{\beta \lambda}{r}} \leq e^{\Upsilon \nicefrac{\beta \lambda}{r}}$ to obtain the following:
    \begin{equation}
        \indone{\frac{\evolchan{\rho, \beta}}{\Tr(\evolchan{\rho, \beta})} - \parens{\frac{\tschan{2k}{\rho, \nicefrac{\beta}{r}}}{\Tr(\tschan{2k}{\rho, \nicefrac{\beta}{r}})}}^{\circ r}} \leq 4\Upsilon^{2k+1} \frac{\alpha_{comm}(H,2k)}{(2k+1)!} \frac{\beta^{2k+1}}{r^{2k}} \parens{e^{\frac{11 \Upsilon \beta \lambda}{r} } + e^{\frac{7 \Upsilon \beta \lambda} {r} }} 
    \end{equation}
    Now, let us posit that $r \in \Omega(\Upsilon^a \beta^b \lambda^c) \; | \; a,b,c \geq 1$. Focusing on the terms in parentheses above, the terms in the exponents are positive by definition, so we have $\parens{e^{\nicefrac{10 \Upsilon \beta \lambda}{r} } + e^{\nicefrac{6 \Upsilon \beta \lambda} {r} }} \geq 1$. If the posited scaling holds, then $\parens{e^{\nicefrac{11 \Upsilon \beta \lambda}{r} } + e^{\nicefrac{7 \Upsilon \beta \lambda} {r} }} \in \bigo{1}$ such that this term is asymptotically upper-bounded by some constant $C$. In fact, if $a, b$ or $c >1$, then C=1 in the infinite limit. From \cite{hagan2022composite}, we use the fact that $\alpha_{comm}(H, 2k) \leq 2^{2k} \lambda^{2k+1}$ to verify this intuition. 

    \begin{equation}
        \indone{\frac{\evolchan{\rho, \beta}}{\Tr(\evolchan{\rho, \beta})} - \parens{\frac{\tschan{2k}{\rho, \nicefrac{\beta}{r}}}{\Tr(\tschan{2k}{\rho, \nicefrac{\beta}{r}})}}^{\circ r}} \leq 4^{2k+1} \Upsilon^{2k+1} \frac{\lambda^{2k+1}}{(2k+1)!} \frac{\beta^{2k+1}}{r^{2k}} C < \epsilon, 
    \end{equation}
    where solving for $r$ in terms of $\epsilon$ gives the following:
    \begin{equation}
        \parens{4 \Upsilon \lambda \beta}^{1 + \frac{1}{2k}} \parens{\frac{C}{\epsilon(2k+1)!}}^{\frac{1}{2k}} < r .
    \end{equation}
    Asymptotically, this gives the following asymptotic bound: 
    \begin{equation}
        r \in \Omega(\frac{\Upsilon^{1+\frac{1}{2k}} \lambda^{1+\frac{1}{2k}} \beta^{1+\frac{1}{2k}} }{\epsilon^\frac{1}{2k}}).
    \end{equation}
    Here, we enforce $\epsilon \leq 1$, which then confirms the initial intuition about the asymptotic scaling of $r$, such that that $r \in \Omega(\Upsilon^a \beta^b \lambda^c)$, where $a=b=c=1+\nicefrac{1}{2k}$, and the exponential term vanishes in the infinite limit, yielding the result below:
    \begin{equation}\indone{\frac{\evolchan{\rho, \beta}}{\Tr(\evolchan{\rho, \beta})} - \parens{\frac{\tschan{2k}{\rho, \nicefrac{\beta}{r}}}{\Tr(\tschan{2k}{\rho, \nicefrac{\beta}{r}})}}^{\circ r}} \in \bigo{ \Upsilon^{2k+1}\frac{\alpha_{comm}(H,2k)}{(2k+1)!}\frac{\beta^{2k+1}}{r^{2k}} },
    \end{equation}
    thus completing the proof.
\end{proof}

\subsection{Imaginary Time Composite Channels} 

    \imagCOMP*

    \begin{proof}
        We begin by expanding the imaginary time composite channel $\mathcal{X}(\rho, \beta)$ in terms of familiar quantities that have been previously bounded. That being said, this proof relies almost entirely on Theorems \ref{thm:imagQD} and \ref{thm:imagTS}. Throughout, notation is used such that $\evolchan{\rho, \beta}_A$ represents an exact evolution channel under the Hamiltonian $A$ from the the theorem statement. It is implied throughout that terms in $A$ are simulated by the Trotter-Suzuki channel, while terms in $B$ are simulated by QDrift. Therefore, we write $\achan{\rho, \nicefrac{\beta}{r}}$ and $\bchan{\rho, \nicefrac{\beta}{r}}$ as Trotter and QDrift channels generated by the Hamiltonians $A$ and $B$ for clearer notation, and once again we initial consider a time-step $\nicefrac{\beta}{r}$. Expanding the composite channel we have:
        \begin{equation}
            \indone{\frac{\evolchan{\rho, \nicefrac{\beta}{r}}}{\Tr \evolchan{\rho, \nicefrac{\beta}{r}}} - \frac{\mathcal{X}^{2k}(\rho, \nicefrac{\beta}{r})}{\Tr \mathcal{X}^{2k}(\rho, \nicefrac{\beta}{r})}} = 
            \indone{\frac{\capU_B (\rho, \nicefrac{\beta}{r}) \circ \capU_A (\rho, \nicefrac{\beta}{r}) + E_{A,B}(\nicefrac{\beta}{r})}{\Tr\capU (\rho, \nicefrac{\beta}{r}) } - \frac{\bchan{\rho, \nicefrac{\beta}{r}} \circ \achan{\rho, \nicefrac{\beta}{r}}}{\Tr \bchan{\rho, \nicefrac{\beta}{r}} \circ \achan{\rho, \nicefrac{\beta}{r}}}} ,
        \end{equation}
        where $E_{A,B}(\nicefrac{\beta}{r})$ is an error term that comes from partitioning the exact evolution $\capU$ into two evolution channels over $A$ and $B$. We leave the denominator un-expanded for convenience. As before, we can add and subtract a convenient term to give us terms with common denominators that are more straightforward to bound, and then apply the triangle inequality.
        \bea
            &&= \indone{\frac{\capU_B (\rho, \nicefrac{\beta}{r}) \circ \capU_A (\rho, \nicefrac{\beta}{r}) - \bchan{\rho, \nicefrac{\beta}{r}} \circ \achan{\rho, \nicefrac{\beta}{r}}}{\Tr\capU (\rho, \nicefrac{\beta}{r}) }} 
            \nonumber\\
            &&+ \indone{\frac{E_{A,B}(\nicefrac{\beta}{r})}{\Tr\capU (\rho, \nicefrac{\beta}{r})}} 
            + \indone{\frac{\bchan{\rho, \nicefrac{\beta}{r}} \circ \achan{\rho, \nicefrac{\beta}{r}}}{\Tr\capU (\rho, \nicefrac{\beta}{r})} - \frac{\bchan{\rho, \nicefrac{\beta}{r}} \circ \achan{\rho, \nicefrac{\beta}{r}}}{\Tr \bchan{\rho, \nicefrac{\beta}{r}} \circ \achan{\rho, \nicefrac{\beta}{r}}}}.
        \eea
        Now we proceed by bounding each of these terms individually. Starting with the first term:
        \begin{align}
            &\indone{\frac{\capU_B (\rho, \nicefrac{\beta}{r}) \circ \capU_A (\rho, \nicefrac{\beta}{r}) - \bchan{\rho, \nicefrac{\beta}{r}} \circ \achan{\rho, \nicefrac{\beta}{r}}}{\Tr\capU (\rho, \nicefrac{\beta}{r}) }} \leq \frac{\indone{\capU_B (\rho, \nicefrac{\beta}{r}) - \bchan{\rho, \nicefrac{\beta}{r}}} + \indone{\capU_A (\rho, \nicefrac{\beta}{r}) - \achan{\rho, \nicefrac{\beta}{r}}} }{\Tr\capU (\rho, \nicefrac{\beta}{r})} \\
            &\leq \parens{\indone{\capU_B (\rho, \nicefrac{\beta}{r}) - \bchan{\rho, \nicefrac{\beta}{r}}} + \indone{\capU_A (\rho, \nicefrac{\beta}{r}) - \achan{\rho, \nicefrac{\beta}{r}}} } e^{2\nicefrac{\norm{H} \beta}{r}}
        \end{align}
        where in the second line we use the trace bound found in the proof of Theorem \ref{thm:imagTS}. Each of the diamond norms remaining above have been bounded in Theorems \ref{thm:imagQD} and \ref{thm:imagTS}, and we will plug in said results at the end of the proof. Proceeding to the second term we have:
        \begin{align}
            & \indone{\frac{\bchan{\rho, \nicefrac{\beta}{r}} \circ \achan{\rho, \nicefrac{\beta}{r}}}{\Tr\capU (\rho, \nicefrac{\beta}{r})} - \frac{\bchan{\rho, \nicefrac{\beta}{r}} \circ \achan{\rho, \nicefrac{\beta}{r}}}{\Tr \bchan{\rho, \nicefrac{\beta}{r}} \circ \achan{\rho, \nicefrac{\beta}{r}}}} \\
            &\leq \indone{\bchan{\rho, \nicefrac{\beta}{r}}} \indone{\achan{\rho, \nicefrac{\beta}{r}}} \abs{\frac{\Tr \parens{ \capU (\rho, \nicefrac{\beta}{r}) -  \bchan{\rho, \nicefrac{\beta}{r}} \circ \achan{\rho, \nicefrac{\beta}{r}} }}{\Tr \parens{\capU (\rho, \nicefrac{\beta}{r})} \Tr \parens{\bchan{\rho, \nicefrac{\beta}{r}} \circ \achan{\rho, \nicefrac{\beta}{r}}}}} \\
            &= \indone{\bchan{\rho, \nicefrac{\beta}{r}}} \indone{\achan{\rho, \nicefrac{\beta}{r}}} \frac{\abs{\Tr \parens{ \capU_B(\rho, \nicefrac{\beta}{r}) \circ \capU_A(\rho, \nicefrac{\beta}{r}) + E_{A,B}(\nicefrac{\beta}{r}) -  \bchan{\rho, \nicefrac{\beta}{r}} \circ \achan{\rho, \nicefrac{\beta}{r}} }}}{\Tr \parens{\capU (\rho, \nicefrac{\beta}{r})} \Tr \parens{\bchan{\rho, \nicefrac{\beta}{r}} \circ \achan{\rho, \nicefrac{\beta}{r}}}} \\
            &\leq e^{\frac{2\beta \lambda_B}{r}} e^{\frac{2\beta \lambda_A}{r}} e^{\frac{2\beta \norm{H}}{r}} \frac{\indone{ \capU_B(\rho, \nicefrac{\beta}{r}) \circ \capU_A(\rho, \nicefrac{\beta}{r}) + E_{A,B}(\nicefrac{\beta}{r}) -  \bchan{\rho, \nicefrac{\beta}{r}} \circ \achan{\rho, \nicefrac{\beta}{r}} }}{\Tr \parens{\bchan{\rho, \nicefrac{\beta}{r}} \circ \achan{\rho, \nicefrac{\beta}{r}}}} \\
            &\leq e^{\frac{2\beta \lambda}{r}} e^{\frac{2\beta \norm{H}}{r}} \parens{\frac{\indone{\capU_B(\rho, \nicefrac{\beta}{r})   -  \bchan{\rho, \nicefrac{\beta}{r}} } + \indone{\capU_A(\rho, \nicefrac{\beta}{r}) - \achan{\rho, \nicefrac{\beta}{r}} }}{\Tr \parens{\bchan{\rho, \nicefrac{\beta}{r}} \circ \achan{\rho, \nicefrac{\beta}{r}}}} + \frac{\indone{E_{A,B}(\nicefrac{\beta}{r})}}{\Tr \parens{\bchan{\rho, \nicefrac{\beta}{r}} \circ \achan{\rho, \nicefrac{\beta}{r}}}}} ,
        \end{align}
        where in the last line we once again apply the triangle inequality, and also observe that $\lambda_A + \lambda_B = \lambda$. Now we only have 2 unfamiliar terms left to bound, the partitioning error $E_{A,B}(\nicefrac{\beta}{r})$ and the trace of the composition of the QDrift and Trotter channels. To analyze the trace, we expand both of the channels into their operator form:
        \begin{align}
            \Tr \parens{\bchan{\rho, \nicefrac{\beta}{r}} \circ \achan{\rho, \nicefrac{\beta}{r}}} &= \Tr \parens{\sum_{\mathbf{j}} P_\mathbf{j} \mathbf{V_j} \prodform \rho \prodform \mathbf{V_j}^\dagger} \\
            &=  \sum_{\mathbf{j}} P_\mathbf{j} \Tr \parens{\mathbf{V_j} \prodform \rho \prodform \mathbf{V_j}^\dagger} \\
            & = \sum_{\mathbf{j}} P_\mathbf{j} \Tr \parens{\prodform \mathbf{V_j}^\dagger \mathbf{V_j} \prodform \rho } \\ 
            &\geq \sum_{\mathbf{j}} P_\mathbf{j} \sigma_{min} \parens{\prodform \mathbf{V_j}^\dagger \mathbf{V_j} \prodform } \Tr \rho \\
            &\geq \sigma_{min}^2 \parens{\prodform} \sum_{\mathbf{j}} P_\mathbf{j}  \sigma_{min}^2 \parens{\mathbf{V_j}} \\
            &\geq e^{-\frac{2 \beta \lambda_A }{r}} \sum_{\mathbf{j}} P_\mathbf{j}  \sigma_{min}^2 \parens{\mathbf{V_j}}
        \end{align}
        Now recall the definition of the product formula from the multi-sample QDrift channel defined in \ref{def:QD} 
        $\mathbf{V_j} = \prod^N_{k=1} e^{-iH_{j_k} \tau}$ with $\tau = \nicefrac{\beta \lambda}{N}$. To reiterate, the index $\mathbf{j}$ is a vector of length $N$, and the index set contains every possible length-$N$ permutation of the $L_B$ Hamiltonian terms in the QDrift channel. Since we already have exponential terms in our bound, we can use a naive lower bound here, which coincides with the most likely case of drawing the largest term $N$ times. This gives:
        \begin{align}
            &\geq e^{-\frac{2 \beta \lambda_A }{r}} \sum_{\mathbf{j}} P_\mathbf{j}  \prod^N_{k=1} \sigma_{min}^2 \parens{ e^{-B_{j_k} \frac{\beta \lambda_B}{Nr}}} \\
            &\geq e^{-\frac{2 \beta \lambda_A }{r}}  e^{-  \frac{2 \norm{B} \beta \lambda_B}{r}} \\
            &= e^{-\frac{2 \beta \lambda }{r}}   ,
        \end{align}
        where in the last line $\norm{B}=1$ based by definition of QDrift, and $\lambda_A + \lambda_B = \lambda$. 

        Now combining all of our results we have the following:
        \begin{equation}
            \parens{\indone{\capU_B(\rho, \nicefrac{\beta}{r})   -  \bchan{\rho, \nicefrac{\beta}{r}}} + \indone{\capU_A(\rho, \nicefrac{\beta}{r})   -  \achan{\rho, \nicefrac{\beta}{r}}} + \indone{E_{A,B}(\nicefrac{\beta}{r})}} e^{\frac{2 \beta \norm{H} }{r}} \parens{1+ e^{\frac{4 \beta \lambda }{r}}}
        \end{equation}
        This term is a sum of the errors of each respective algorithm, with the addition of the partitioning error, multiplied by some exponential terms that come from lower bounding traces, which have been common throughout. Therefore, this is exactly what we might expect a priori. Going forward we wish to represent this error for the entire simulation rather than just a single time step, and then show that the exponential terms become irrelevant in the limit of interest. For the full time simulation we have
        \begin{equation}
            \indone{\parens{\frac{\evolchan{\rho, \nicefrac{\beta}{r}} }{\Tr \evolchan{\rho, \nicefrac{\beta}{r}} }}^{\circ r}- \parens{\frac{\mathcal{X}^{2k}(\rho, \nicefrac{\beta}{r})}{\Tr \mathcal{X}^{2k}(\rho, \nicefrac{\beta}{r})}}^{\circ r}} \leq r \indone{\frac{\evolchan{\rho, \nicefrac{\beta}{r}} }{\Tr \evolchan{\rho, \nicefrac{\beta}{r}} }- \frac{\mathcal{X}^{2k}(\rho, \nicefrac{\beta}{r})}{\Tr \mathcal{X}^{2k}(\rho, \nicefrac{\beta}{r})}} ,
        \end{equation}
        where the quantity on the right is that which has been bounded. Inserting the imaginary time QDrift bound from Theorem \ref{thm:imagQD}, and the single step Trotter bound from Equation \ref{eq:single_step_itime_trot} we have: 
        \begin{equation}
            \parens{\frac{C\beta^2 \lambda_B^2}{N_B r} + 2\Upsilon^{2k+1} \frac{\alpha_{comm}(A,2k)}{(2k+1)!} \frac{\beta^{2k+1}}{r^{2k}} e^{4\nicefrac{\beta}{r} \Upsilon \lambda_A}\parens{e^{4\nicefrac{\beta}{r} \lambda_A} e^{2\nicefrac{\beta}{r} \norm{A}} + e^{2\nicefrac{\beta}{r} \norm{A}}} + r \indone{E_{A,B}(\nicefrac{\beta}{r})}} e^{\frac{2 \beta \norm{H} }{r}} \parens{1+ e^{\frac{4 \beta \lambda }{r}}}. 
        \end{equation}
        Next we deal with the partitioning error $\indone{E_{A,B}(\nicefrac{\beta}{r})} = \indone{\evolchan{\rho, \nicefrac{\beta}{r}} - \mathcal{U}_B(\rho, \nicefrac{\beta}{r}) \circ \mathcal{U}_A(\rho, \nicefrac{\beta}{r})}$. We know from Theorem \ref{thm:imagTS} that $\indone{\evolchan{\beta, \rho} - \mathcal{U}_B(\rho, \nicefrac{\beta}{r}) \circ \mathcal{U}_A(\rho, \nicefrac{\beta}{r})} \leq 2 e^{\frac{\beta \lambda}{r}}\norm{\evolchan{\beta, \rho} - \mathcal{U}_B(\rho, \nicefrac{\beta}{r}) \circ \mathcal{U}_A(\rho, \nicefrac{\beta}{r})}_\infty$, then once again using the bound from Appendix E of  \cite{childs2021theory}, we have:
        \begin{equation}
            \parens{\frac{C\beta^2 \lambda_B^2}{N_B r} + 2\Upsilon^{2k+1} \frac{\alpha_{comm}(A,2k)}{(2k+1)!} \frac{\beta^{2k+1}}{r^{2k}} e^{4\nicefrac{\beta}{r} \Upsilon \lambda_A}\parens{e^{4\nicefrac{\beta}{r} \lambda_A} e^{2\nicefrac{\beta}{r} \norm{A}} + e^{2\nicefrac{\beta}{r} \norm{A}}} + \norm{[A,B]} e^{\frac{\beta \lambda}{r}} \frac{2\beta^2}{r}e^{\frac{4\beta \lambda}{r}}}  e^{\frac{2 \beta \norm{H} }{r}} \parens{1+ e^{\frac{4 \beta \lambda }{r}}}
        \end{equation}
        We are now left with an equation non-linear in $r$, which makes it more difficult to make asymptotic arguments similar to those in made Theorem \ref{thm:imagTS}. Since, we wish to eliminate the exponentials from this bound, we first upper-bound all exponentials by $e^{\nicefrac{n \Upsilon \beta \lambda}{r}}$ where we keep the coefficients $n$ for each of them. This allows us to collect the exponentials and simplify the expression to the following:

        \begin{equation}
            \parens{\frac{C\beta^2 \lambda_B^2}{N_B r} + 2\Upsilon^{2k+1} \frac{\alpha_{comm}(A,2k)}{(2k+1)!} \frac{\beta^{2k+1}}{r^{2k}} \parens{e^{10\nicefrac{\Upsilon \beta \lambda}{r}}  + e^{6\nicefrac{\Upsilon \beta \lambda}{r}}}+ \norm{[A,B]} \frac{2\beta^2}{r}e^{\nicefrac{5 \Upsilon \beta \lambda}{r}}}  \parens{e^{\nicefrac{2 \Upsilon \beta \lambda }{r}}  + e^{\nicefrac{6 \Upsilon \beta \lambda }{r}}}
        \end{equation}
    Now, we introduce the constraint that $r\geq \Upsilon \beta \lambda$. Imposing this constraint over $r$ allows all of the exponentials to be upper-bounded by a constant such that in the limit that all other parameters approach infinity, this bound becomes the following:
    \begin{equation}
        \indone{\frac{\evolchan{\rho, \beta} }{\Tr \evolchan{\rho, \beta} }- \frac{\mathcal{X}^{2k}(\rho, \beta)}{\Tr \mathcal{X}^{2k}(\rho, \beta)}} \in \bigo{\frac{\beta^2 \lambda_B^2}{N_B r} + \Upsilon^{2k+1} \frac{\alpha_{comm}(A,2k)}{(2k+1)!} \frac{\beta^{2k+1}}{r^{2k}}   + \norm{[A,B]} \frac{\beta^2}{r}},
    \end{equation}
    thus completing the proof. 
    \end{proof} 

\section{Numerical Error Analysis} \label{sec:appendix_error}
\subsection{Types of Error Measures}
 Here, we propose two different error measures and discuss the advantages associated with each.

\subsubsection{Infidelity of Quantum States} \label{sec:infidelity}
The fidelity is a quantity that serves as a measure of the closeness between two quantum states, of which can be either pure or mixed states and is summarized in the following definition.

\begin{definition}[Fidelity] 
    Let $\rho, \sigma \in \mathbb{C}^{2^n \times 2^n}$ be density matrices that represent quantum states with equal Hilbert space dimension. The fidelity between the states is then defined as: 
    \begin{equation}
        \fidel(\rho \sigma) := \Tr\brackets{\sqrt{\sqrt{\rho}\sigma\sqrt{\rho}}}^2 \\
    \end{equation}
    Where $\fidel(\rho, \sigma) \in [0,1]$ and $\fidel(\rho, \sigma) = \fidel(\sigma, \rho)$.
\end{definition}
The infidelity is then simply defined as $1-\fidel$ and represents a measure of distance between two states with the latter properties in the definition still holding. The computational advantage of using infidelity over some other measures is that it is relatively cheap to compute, especially for large systems. This does not appear obvious on first glance given the need to compute 2 matrix square roots which often involves diagonalization. However, we can take advantage of a nice property: $\fidel = \Tr(\rho, \sigma)$ if one of $\rho$ or $\sigma$ is pure. A simple proof is provided in the following lemma.

\begin{lemma}[Pure-Non-Pure Fidelity Property] \label{lem:PureNonPure}
Given a two density matrices $\rho$ and $\sigma$ where at least one of which is pure, then the expression for the fidelity simplifies to $\fidel(\rho, \sigma) = \Tr(\rho \sigma)$.
\begin{proof}
\begin{align}
    \fidel(\rho, \sigma) :=& \Tr\brackets{\sqrt{\sqrt{\rho}\sigma\sqrt{\rho}}}^2 \\
    &= \Tr\brackets{\sqrt{\sqrt{\ketbra{\psi}{\psi}}\sigma\sqrt{\ketbra{\psi}{\psi}}}}^2 \\
    &= \Tr\brackets{\sqrt{\ketbra{\psi}{\psi} \sigma \ketbra{\psi}{\psi}}}^2 \\
    &= \Braket{\psi}{\sigma}{\psi}\Tr(\ketbra{\psi}{\psi}) \\
    &= \Tr(\rho \sigma)
\end{align}  
Here we choose $\rho$ to be a pure state $\ketbra{\psi}{\psi}$ which is a projector and thus $ \rho = \rho^2 = \sqrt{\rho}$. We use these properties along with the fact that the trace of a pure state is equal to one to complete the proof. Note, $\fidel$ is also known to be symmetric $\fidel(\rho, \sigma) = \fidel(\sigma, \rho)$ but we do not show this here. This implies the statement is true for either $\rho$ or $\sigma$ pure.
\end{proof}
\end{lemma}

The important part of this lemma is that we only require one of the density matrices to be a pure state. To evaluate the error of a composite channel, we will be taking the infidelity between channel outputs, and due to the nature of QDrift the channel will always output a mixed state. However, the exact evolution channel providing the means for comparison only produces pure states given pure state input. With these definitions and properties we can then present the following strategy: Reformulate the fidelity expression such that it can be written as an expectation value that only requires matrix vector multiplication to compute, and then use Monte-Carlo sampling to approximate this quantity. This is shown via the following derivation: 

\begin{align}
    \fidel(\capU(\rho, t), \mathcal{X}^{2k}(\rho, t)) &= \Tr(e^{-iHt}\ketbra{\psi}{\psi}e^{iHt} e^{-iAt}e^{-iBt}\ketbra{\psi}{\psi}e^{iBt}e^{iAt}) \\
    &= \Tr(e^{-iHt}\ketbra{\psi}{\psi}e^{iHt} \trotterchan{2k}{-iAt} \sum_j P_j e^{-iB_j\tau} \ketbra{\psi}{\psi} e^{iB_jt} \trotterchan{2k}{-iAt}^\dagger) \\
    &= \sum_j P_j \Tr(\bra{\psi} e^{iB_jt} \trotterchan{2k}{-iAt}^\dagger e^{-iHt}\ketbra{\psi}{\psi} e^{-iHt} \trotterchan{2k}{-iAt}e^{-iB_j\tau} \ket{\psi}) \\ 
    &= \sum_j P_j |\bra{\psi} e^{iB_j\tau} \trotterchan{2k}{-iAt}^\dagger e^{-iHt} \ket{\psi}|^2 \\
    &= \expect{|\bra{\psi} e^{iB_j\tau} \trotterchan{2k}{-iAt}^\dagger e^{-iHt} \ket{\psi}|^2} \\
    & \approx \frac{1}{N} \sum_j^N |\bra{\psi} e^{iB_j\tau} \trotterchan{2k}{-iAt}^\dagger e^{-iHt} \ket{\psi}|^2 \label{derivation}
\end{align}

Now when numerically evaluating the error, only matrix-vector multiplication is needed as opposed to matrix-matrix before. Here we no longer need to build exact density matrices outputted by the composite channel, and the step involving the matrix square root has been eliminated. However, the trade-off is that we have to use Monte-Carlo sampling to estimate this expectation value. Given that the naive methods for matrix-vector and matrix-matrix multiplication scale like $\bigo{n^2}$ and $\bigo{n^3}$ respectively, for small systems the overhead of Monte-Carlo sampling may lead to a loss of any computational advantage. However, for large systems where $n^3 \gg n^2$, this sampling method is quite advantageous. It is also important to note that sampling makes QDrift channel much cheaper to study as well. For a QDrift channel with $L$ terms and $N$ samples, building the exact channel yields additional matrix multiplication overhead $\bigo{NL}$ (not considering arithmetic overhead). While this is does not seem prohibitive, it does become expensive when we have a Hamiltonian with many terms that we wish to sample many times for a high degree of accuracy. It should be noted, however, that this is not as bad as it initially seems given that there are $L^N$ terms QDrift channel sum. The numerical routine that computes this sum was optimized to achieve this, by using the fact that channel samples are i.i.d. With intermediate summation to produce a new input $\rho$ after each sample vector, we do not need to compute each term in this sum individually:

\begin{equation}
    \sum_j p_j e^{-iH_j t} \rho_i e^{iH_j t} \rightarrow \rho_{i+1}
\end{equation}
repeating this procedure $N$ times until $\rho_i \rightarrow \rho_N$ which constructs the exact QDrift channel output with $\bigo{NL}$ matrix multiplications. We mention this to highlight the advantages of Monte-Carlo sampling where it can be done reliably, as well as to highlight how QDrift is numerically computed exactly in our work (which will be made necessary by results in the following section). 

\subsubsection{Trace Distance over Density Matrices} \label{sec:tracedist} 
The next error measure we wish to investigate is the trace distance, which is equivalent to the Schatten 1-norm. The trace distance is a metric over the space of density matrices that can be defined in the following way:

\begin{definition}[Trace Distance] \label{def:qdrift_channel}
    Let $\rho, \sigma \in \mathbb{C}^{2^n \times 2^n}$ be density matrices that represent quantum states with equal Hilbert space dimension. The trace distance between the states is then defined as: 
    \begin{equation}
        \tracedist(\rho, \sigma) := \norm{\rho - \sigma}_1 =  \Tr\brackets{\sqrt{(\rho-\sigma)(\rho-\sigma)^\dagger}} \\
    \end{equation}
    Where $\tracedist(\rho, \sigma) \in [0,1]$.
\end{definition}

 Given that $\rho$ and $\sigma$ are necessarily hermitian, the trace distance simplifies to $\tracedist(\rho, \sigma) = \sum_i \abs{\lambda_i}^2$ where $\lambda_i$ are eigenvalues of the matrix $(\rho - \sigma)$. There are no simplifications in our approach to computing this numerically. Here, we must compute the exact density matrix output of the composite channel (a  necessarily mixed state outputted by the QDrift channel component) and the output of the unitary channel. Exact diagonalization is performed to obtain the eigenvalues, and the sum above is then computed. Providing a rigorous argument for crossover points in the compute time between the fidelity and trace distance (as a function of $N, L$) is beyond the scope of this paper, however, given the arguments from the previous section our intuition is that for small systems the trace distance will be computable in a reasonable amount of time whereas Monte-Carlo sampling the infidelity will become more advantageous as the system size grows.

\subsection{Comparison of Error Measures}
In the following section we provide short proofs regarding how the error of each of the QDrift and Trotter-Suzuki algorithms scale with time, both in the infidelity and trace distance framework. It is important to note that we are specifically using the Schatten 1-norm over density matrices and not an induced 1-norm over channels. Therefore, we are computing this quantity using the outputs of the two channels, and the same goes for the infidelity which is only defined on the density matrices that are the outputs of the channels in question. We begin with the infidelity:

 \qdinfscaling*
\begin{proof}
\begin{align}
    1- \fidel(\qdchan{\rho, t}, \evolchan{\rho, \nicefrac{t}{N}}) :=& 1- \Tr(\qdchan{\rho, t} \capU(\rho, \nicefrac{t}{N})) \\
    =& 1- \Tr(\sum_j P_j e^{-i H_j \tau} \ketbra{\psi}{\psi} e^{i H_j \tau} e^{-i H t/N} \ketbra{\psi}{\psi} e^{i H t/N}) \\
    =& 1- \sum_j P_j \bra{\psi} e^{i H t/N}e^{-i H_j \tau} \ketbra{\psi}{\psi} e^{i H_j \tau} e^{-i H t/N} \ket{\psi} \\
    =& 1- \sum_j P_j |\bra{\psi}(\openone - i H t/N - \frac{H^2 t^2/N^2}{2!} + ...)(\openone + i H_j \tau - \frac{H_j^2 \tau^2}{2!} + ...)\ket{\psi}|^2 \label{eq:QDriftcomm} \\ 
    =& 1- \sum_j P_j |\bra{\psi}(\openone - i H t/N + i H_j \tau + H H_j t \tau/N - \frac{H^2 t^2/N^2}{2!} - \frac{H_j^2 \tau^2}{2!} + ...)\ket{\psi}|^2 \\
    =& 1- \sum_j P_j |\bra{\psi}(\openone - i H t/N + i H_j t \lambda/N + H H_j t^2 \lambda/N^2  \\
    \text{      } &- \frac{H^2 t^2/N^2}{2!} - \frac{H_j^2 t^2 \lambda^2 /N^2}{2!} + ...)\ket{\psi}|^2 \\
    =& 1- \sum_j P_j (1 - 2 \Braket{\psi}{H}{\psi}\Braket{\psi}{H_j}{\psi} t^2 \lambda/N^2 + \Braket{\psi}{H}{\psi}^2 t^2/N^2 + \Braket{\psi}{H_j}{\psi}^2 t^2 \lambda^2/N^2 + c t^2 ...) \\
    =& t^2 \sum_j P_j (2 \Braket{\psi}{H}{\psi}\Braket{\psi}{H_j}{\psi} \lambda/N^2 - \Braket{\psi}{H}{\psi}^2/N^2 - \Braket{\psi}{H_j}{\psi}^2 \lambda^2/N^2 - c ...) \\
    \leq & t^2  \sum_j P_j (|2 \Braket{\psi}{H}{\psi}\Braket{\psi}{H_j}{\psi}| \lambda/N^2 + |\Braket{\psi}{H}{\psi}^2|/N^2 + |\Braket{\psi}{H_j}{\psi}^2| \lambda^2/N^2 + |c| ...) \in \bigo{t^2}
\end{align}  
We use the cyclic property of the trace and the fact that $\rho$ is a pure state $\ketbra{\psi}{\psi}$ and then replace the operator exponentials with their respective Taylor series expansions, keeping terms only up to $t^2$ throughout. In the final step, the modulus is taken and the triangle inequality applied to upper bound the sum. This is done for the purpose of ensuring the signs of the terms are ``well behaved" such that $\bigo{}$ notation can be aptly applied. As well, $c$ is used for brevity where $c = \Braket{\psi}{(H H_j \lambda/N^2 - \frac{H^2 }{2!}/N^2 - \frac{H_j^2 \lambda^2 /N^2}{2!})}{\psi} + h.c $. Here the notation $\bigo{t^2}$ is understood in the limit as $t \rightarrow 0$.
\end{proof}

Given the diamond distance bound in \cite{campbell2019random}, the $\bigo{t^2}$ infidelity scaling is not surprising. Not only this, but recall from Equation \ref{eq:QDriftcomm}, where the largest order in commutator error from the product formula was quadratic, and given that QDrift applies no symmetrization strategy, we did not expect this to disappear. We would like to remind the reader that the shining feature of QDrift is that the cost is independent of the number of terms in the Hamiltonian. Next, we examine how a Trotter-Suzuki channel behaves in the infidelity measure:

\TSfidelity*
\begin{proof}
\begin{align}
    1- \fidel(\tschan{2k}{\rho, t},  \capU(\rho, t)) :=& 1- \Tr(\tschan{2k}{\rho, t} \evolchan{\rho, t}) \\
    =& 1- \Tr( \trotterchan{2k}{-i H t} \ketbra{\psi}{\psi} \trotterchan{2k}{-i H t}^\dagger e^{-i H t} \ketbra{\psi}{\psi} e^{i H t}) \\
    =& 1- \bra{\psi} e^{i H t} \trotterchan{2k}{-i H t} \ketbra{\psi}{\psi} \trotterchan{2k}{-i H t}^\dagger e^{-i H t} \ket{\psi} \\
    =& 1- |\bra{\psi} e^{i H t} \trotterchan{2k}{-i H t} \ket{\psi}|^2 \\
    =& 1-|\bra{\psi} (\trotterchan{2k}{-i H t}^\dagger \pm i t^{2k+1} R_{2k+1} \mp i t^{2k+3} R_{2k+3} + ...) \trotterchan{2k}{-i H t} \ket{\psi}|^2 \\
    =& 1-|\bra{\psi} \openone + \pm i t^{2k+1} R_{2k+1} \trotterchan{2k}{-i H t} \ket{\psi} \mp i t^{2k+3} R_{2k+3} \trotterchan{2k}{-i H t} \ket{\psi} + ...|^2  \\
    =& 1- |1 \pm i t^{2k+1}\Braket{\psi}{R_{2k+1} \trotterchan{2k}{-i H t}}{\psi} \mp i t^{2k+3} \Braket{\psi}{R_{2k+3} \trotterchan{2k}{-i H t}}{\psi} \pm ...|^2  \\
    \leq & 1- |1|^2 +  (t^{2k+1})^2|\Braket{\psi}{R_{2k+1} \trotterchan{2k}{-i H t}}{\psi}|^2 + ... \\
    =& (t^{2k+1})^2|\Braket{\psi}{R_{2k+1} \trotterchan{2k}{-i H t}}{\psi}|^2 + ... \in \bigo{(t^{2k+1})^2} 
\end{align}  
We again use the cyclic property of the trace and the fact that $\rho$ is a pure state $\ketbra{\psi}{\psi}$. The evolution operator $e^{-iHt}$ is rewritten in terms of the product formula and corrections at significant orders. $R_n$ are remainder operators of order $n$ of which we keep to order $2k+3$. In the second last step, we use the triangle inequality and ignore terms of higher order. Keeping higher order terms is unnecessary in that we cannot use them to generate terms of order $\leq (t^{2k+1})^2$. Once again, $\bigo{(t^{2k+1})^2}$ is understood in the limit as $t \rightarrow 0$.
\end{proof}

This theorem provides a peculiar result. The expected error in the output of a Trotter-Suzuki channel from section \ref{sec:imag_TS} is squared by the nature of the infidelity calculation. Given that $\epsilon^2 \leq \epsilon \leq 1$, error squaring can lead to a much lower gate cost, and given that Trotter experiences this phenomenon while QDrift does not, a numerical optimizer may heavily favour Trotter and the potential advantages from the composite approach will likely fade. This is the main issue with choosing infidelity as an error measure whilst hybridizing algorithms with different error scaling. The following two theorems show that the trace distance suffers no such side effect:

\QDtracedist*
\begin{proof}
\begin{align}
    \tracedist(\qdchan{\rho, t}, \capU(\rho, \nicefrac{t}{N})) :=& \Tr\brackets{\sqrt{(\qdchan{\rho, t}- \capU(\rho, \nicefrac{t}{N}))(\qdchan{\rho, t}- \capU(\rho, t/N))^\dagger}} \\
    =& \Tr\brackets{\sqrt{(\qdchan{\rho, t}- \capU(\rho, \nicefrac{t}{N}))^2}} \\
    =& \Tr\brackets{\sqrt{((\sum_j P_j e^{-i H_j \tau} \rho e^{i H_j \tau})-(e^{-i H t/N} \rho e^{i H t/N}))^2}}  \\
    =& \Tr\brackets{\sqrt{((\sum_j P_j (\openone - i H_j \tau - H_j^2 \tau^2 +...) \rho (\openone +i H_j\tau -H_j^2 \tau^2 + ... )-(e^{i H t/N} \rho e^{-i H t/N}))^2}} \\
    =& \Tr\brackets{\sqrt{(\sum_j P_j (\rho + i\lambda t/N [H_j, \rho] + \lambda^2 t^2/N^2(H_j \rho H_j - i \{H_j, \rho\}) + ...)-(e^{i H t/N} \rho e^{-i H t/N}))^2}} \\
    =& \Tr\brackets{\sqrt{((\rho + \sum_j \frac{h_j}{\lambda} (i\lambda t/N [H_j, \rho] + \lambda^2 t^2/N^2(H_j \rho H_j - i \{H_j, \rho\}) + ...)-(e^{i H t/N} \rho e^{-i H t/N}))^2}} \\
    =& \Tr\brackets{\sqrt{(\rho +  i t/N [H, \rho] + \sum_j \frac{h_j}{\lambda}(\lambda^2 t^2/N^2(H_j \rho H_j - i \{H_j^2, \rho\}) + ...)-(e^{i H t/N} \rho e^{-i H t/N}))^2}} \\
    =& \Tr\brackets{\sqrt{\sum_j \frac{h_j}{\lambda}(\lambda^2 t^2/N^2(H_j \rho H_j - i \{H_j^2, \rho\}  ...) + (t^2 (H\rho H + i\{H^2, \rho\}) + ...))^2}} \\
    =& t^2 \Tr\brackets{\sqrt{(R(t))^2}} \in \bigo{t^2}
\end{align}  
Here we expand the Qdrift channel with the Taylor series of the exponential operators and show that the terms of this series cancel with that of the Taylor series of $\capU_N$ up to order $t^2$. To achieve this we simply use the definitions of $P_j = h_j / \lambda$ and $\sum_j H_j = H$. In the last line $R(t)$ is considered a remainder operator which is only a function of polynomials $t^n, n\geq 0$. The notation $\bigo{t^2}$ is understood in the limit as $t \rightarrow 0$.
\end{proof}

\TStracedist*

\begin{proof}
\bea
    &&\tracedist(\tschan{2k}{\rho, t}, \capU(\rho, t)) :=\Tr\brackets{\sqrt{(\tschan{2k}{\rho, t}- \capU(\rho, t))(\tschan{2k}{\rho, t}- \capU(\rho, t))^\dagger}} \\
    &&= \Tr\brackets{\sqrt{(\tschan{2k}{\rho, t}- \capU(\rho, t))^2}} \\
    &&= \Tr\brackets{\sqrt{((\trotterchan{2k}{-iHt} \rho \trotterchan{2k}{-iHt}^\dagger)-(e^{-i H t} \rho e^{i H t}))^2}}  \\
    &&= \Tr\brackets{\sqrt{((\trotterchan{2k}{-iHt} \rho \trotterchan{2k}{-iHt}^\dagger)-((\trotterchan{2k}{-i H t}^\dagger \pm i t^{2k+1} R_{2k+1} + ...) \rho (\trotterchan{2k}{-i H t}^\dagger \mp i t^{2k+1} R_{2k+1}^\dagger + ...))^2}}\nonumber\\ \\
    &&= \Tr\brackets{\sqrt{((\mp it^{2k+1} \trotterchan{2k}{-iHt}\rho R^\dagger_{2k+1} \pm it^{2k+1} R_{2k+1} \trotterchan{2k}{-iHt}^\dagger + ...))^2}} \\
    &&= t^{2k+1}\Tr\brackets{\sqrt{(R(t))^2}} \in \bigo{t^{2k+1}}
\eea
Here expand the exponential operators in the evolutionary channel $\capU$ in terms of their error terms for some general $2k$th order Trotter-Suzuki channel, where $R_n$ are remainder operators of order $n$. As the algebra proceeds we only keep terms of order $\leq t^{2k+1}$ throughout. Keeping higher order terms is unnecessary in that we cannot use them to generate terms of order $\leq t^{2k+1}$. In the last line $R(t)$ is considered a general remainder operator which is only a function of polynomials $t^n, n\geq 0$. Once again, $\bigo{t^{2k+1}}$ is understood in the limit as $t \rightarrow 0$.
\end{proof}

\begin{figure}[h!]
    \centering
        \begin{subfigure}[b]{.32\textwidth}
            \includegraphics[width=1\textwidth]{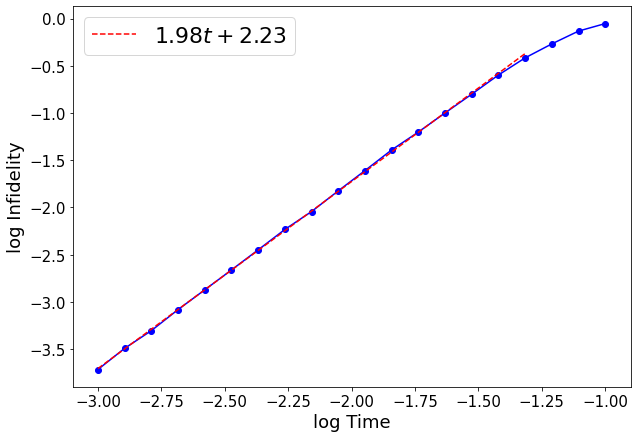}
            \caption{QDrift Infidelity}
        \end{subfigure}
        \begin{subfigure}[b]{.32\textwidth}
            \includegraphics[width=1\textwidth]{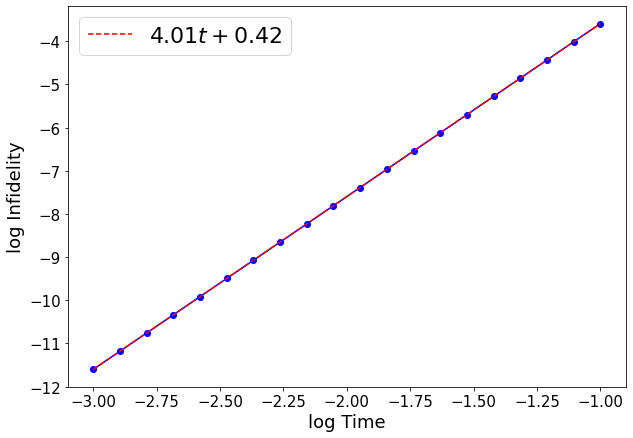}
            \caption{Trotter Infidelity}
        \end{subfigure}
        \begin{subfigure}[b]{.32\textwidth}
            \includegraphics[width=1\textwidth]{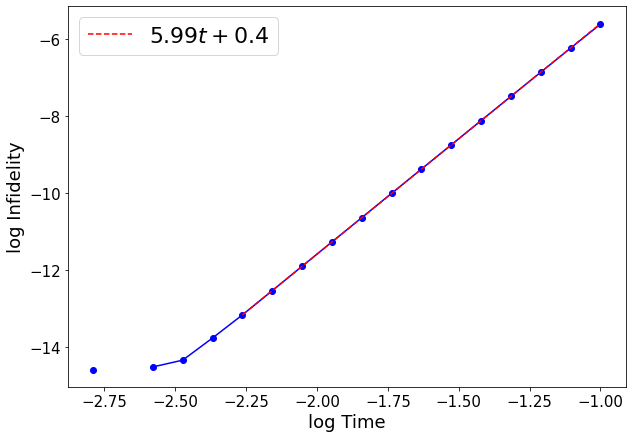}
            \caption{TS Order $2k=2$ Infidelity}
        \end{subfigure}
        \begin{subfigure}[b]{.32\textwidth}
            \includegraphics[width=1\textwidth]{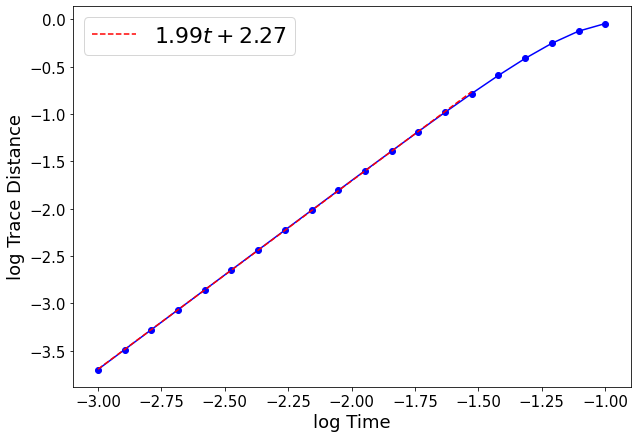}
            \caption{QDrift Trace Distance}
        \end{subfigure}\begin{subfigure}[b]{.32\textwidth}
            \includegraphics[width=1\textwidth]{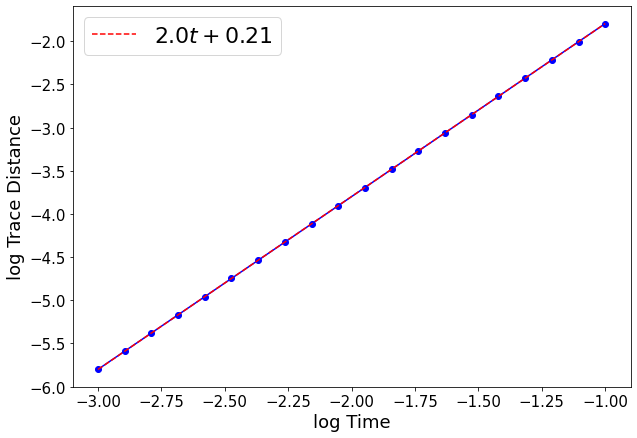}
            \caption{Trotter Trace Distance}
        \end{subfigure}\begin{subfigure}[b]{.32\textwidth}
            \includegraphics[width=1\textwidth]{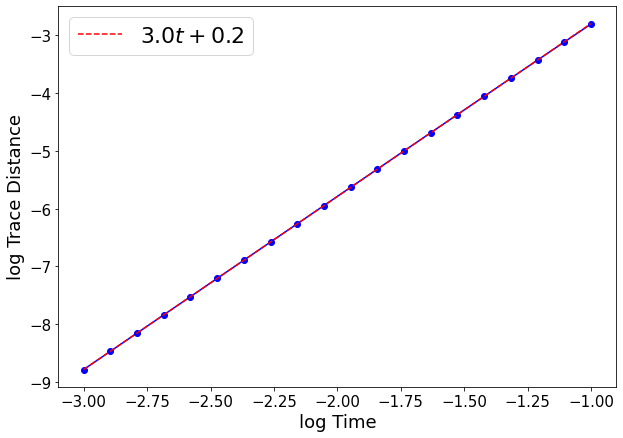}
            \caption{TS order $2k=2$ Trace Distance}
        \end{subfigure}
        \caption{Numerical plots of log infidelity and trace distance scaling with log time for both Trotter-Suzuki and QDrift channels. Plots are all constructed using 1 iteration or sample respectively. Note that the slope of the plots follow the analytical predictions exactly, with the infidelity squaring the expected Trotter error in (b) and (c). Note that in (c) we run into machine precision error at small times and for long times in (a) and (d) QDrift no longer converges. The plots were compute using 100 Monte-Carlo samples for QDrift Infidelity, and exact density matrices elsewhere for a Graph Hamiltonian with 5 spins.} \label{fig:error}
\end{figure}

There are no surprises with the trace distance error as it yields the expected time scaling from the results of section \ref{sec:tracedist}. Numerical investigations also support our proofs (see Figure \ref{fig:error}). Therefore, while it may be more expensive to compute for larger systems, this makes it favourable to work with as it provides a more ``fair" measure, treating the algorithms on more equal footing. Thus, it is expected that the composite channel will appear more performant using this metric. In addition, using this framework requires no sampling, which means that monotonicity of the cost will be guaranteed with respect to channel iterations. This will prove useful in the cases in which we choose to optimize over possible partitions with respect to the cost. Now, out of interest, we also wish to numerically investigate the trace distance of the composite channel. Analytically, this is likely a messy problem given the partition, but intuitively we expect the scaling to be some linear combination $\alpha t^2 + \beta t^3 +\gamma$ with the slope of the log plot being between 2 and 3 (see Figure \ref{fig:elbow}).
\begin{figure}[h!]
    \centering
        \includegraphics[width=0.65\textwidth]{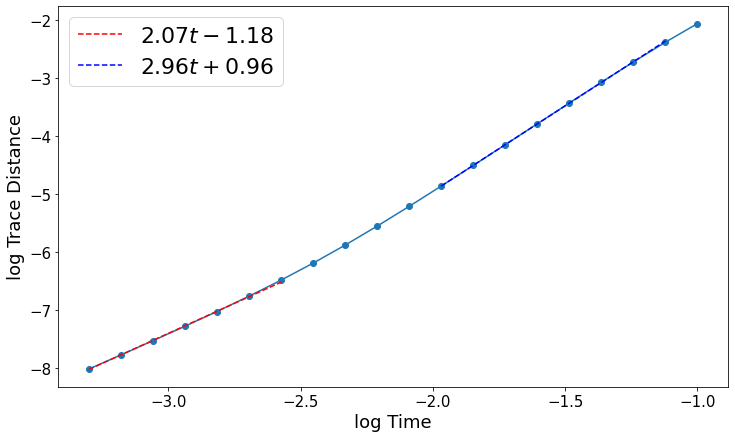}
        \caption{Trace distance time scaling for a composite channel in inner order $2k=2$. This plot emits an interesting structure in that the time dependence of the trace distance goes through a ``phase transition" between QDrift (red) and 2nd order Trotter-Suzuki (blue) dominated regions. Given that QDrift excels in short time simulations and whereas the higher order Trotter-Suzuki dominates at longer simulation times this is not entirely unexpected.} \label{fig:elbow}
\end{figure}
\end{widetext}
\end{document}